\documentclass{./styl_files/llncs}
\input{llncs_header.tex}
\title{A Broadcast Authenticated Encryption with Keyword Search in the Standard Model}
\subtitle{Tightly Secure in Multi-User, Multi-Challenge Settings}

\author{Sayantan Mukherjee}
\institute{Department of Computer Science and Engineering,\\
Indian Institute of Technology, Jammu\\\email{csayantan.mukherjee@gmail.com}}%\footnote{corresponding author}}}

\begin{document}

\maketitle

%!TEX spellcheck = en_US
%!TEX root = ../main.tex

\begin{abstract}
    Searchable Encryption (SE) allows users to perform a keyword search over encrypted documents. In Eurocrypt'04, Boneh et al. introduced Public-key Encryption with Keyword Search (PEKS). Broadcast Encryption with Keyword Search (BEKS) is a natural progression that allows for some amount of access control. However, PEKS and BEKS suffer from keyword-guessing attacks (KGA). In the case of KGA, an adversary guesses the keyword encoded in a trapdoor by creating a ciphertext on a sequence of keywords of its choice and testing them against the trapdoor. In ACISP'21, Liu et al. introduced a variant of BEKS called Broadcast Authenticated Encryption with Keyword Search (BAEKS) as a generalization of Public-Key Authenticated Encryption with Keyword Search (PAEKS) in troduced in Information Sciences in 2017. This was followed with a number works on BAEKS and PAEKS.

    However, no known work considered the functionality requirement in its most realistic setting. We propose a new security definition of BAEKS in the multi-user (with adaptive corruptions) and multi-challenge (both in terms of ciphertext and trapdoor in an interleaved manner) settings. We also study the question of the unforgeability of BAEKS. In fact, our strong hiding requirement already implies a significant amount of unforgeability. We then propose a new BAEKS construction in the bilinear pairing groups. We prove this scheme achieves adaptive tight full-hiding security under (almost) standard MDDH assumptions. Restricting our BAEKS construction for a single receiver also gives an efficient and tightly secure PAEKS construction. We further run experiments to implement and evaluate our scheme. 
\end{abstract}

\section{Introduction}
\label{sec:Intro}
PEKS \cite{EC:BDOP04} is an well known primitive to solve the problem of searching over encrypted data.
In a PEKS, $\Test(\Trp,\Ct)$ outputs $1$ if the ciphertext-attribute $y$ (encoded in ciphertext $\Ct$ by $\SEnc$) \emph{matches} the key-attribute $x$ (encoded in trapdoor $\Trp$ by $\Tgen$) where $\SEnc$ is a publicly evaluation function and $\Tgen$ uses the master secret key $\msk$.
This line of work \cite{TCC:BonWat07,JC:ABCKKL08,C:BonGenWat05,AC:AttFurIma06,INDOCRYPT:ChaMuk18,PKC:LZCHQ24} has produced several PEKS of different security levels and functionalities. %EC:KatSahWat08,
Keyword guessing attacks \cite{BRPL06,YHG08} capture the question if a trapdoor hides encoded key-attribute (in this case the keyword).
Being public key encryptions, however, all the above PEKS schemes were vulnerable to keyword guessing attacks as an adversary can create a ciphertext (using $\SEnc$) and run $\Test$ on the challenge trapdoor to comprehend the keyword encoded in the trapdoor.
There have been many efforts to introduce various query, advanced security requirements, efficiency improvements and expansions. 
In the literature, two prominent directions to overcome key guessing attacks are:
\begin{itemize}
    \item Restrict searching capability: This line of work studied designated-tester PEKS \cite{BaeNaiSus08,FSGW09,EMRO15}. Here, $\Test$ uses the tester's secret key.
    \item Restrict encryption capability: This line of work studied public key authenticated encryption with keyword search (PAEKS) \cite{HuaLi17,ChiQinZhe20,PanLi21,ESORICS:CheMen22,LuLi22,APKC:Emura22,LHHS23,YWYLWJ23}. 
    Here, $\SEnc$ uses the sender's secret key.
\end{itemize}

% The \emph{public key authenticated encryption with keyword search} (PAEKS) \cite{HuaLi17} addressed key guessing attacks by restricting the adversaries encryption capability. 
% For a PAEKS $\SEnc$ uses the sender's secret key.

% In the literature, two prominent directions to overcome key guessing attacks are:
% \begin{itemize}
%     \item Restrict searching capability: This line of work studied designated-tester PEKS \cite{BaeNaiSus08,FSGW09,EMRO15}. Here, $\Test$ uses the tester's secret key.
%     \item Restrict encryption capability: This line of work studied public key authenticated encryption with keyword search (PAEKS) \cite{HuaLi17,ChiQinZhe20,PanLi21,ESORICS:CheMen22,LuLi22,APKC:Emura22,LHHS23,YWYLWJ23}. 
%     Here, $\Enc$ uses the sender's secret key.
% \end{itemize}

\noindent Over the years, PAEKS got much attention \cite{HuaLi17,ChiQinZhe20,InfoSc:QCHLZ20,JISA:PanLi21,ESORICS:CheMen22,LuLi22,APKC:Emura22,LHHS23,InfoSc:YWYLWJH23,InfoSc:CQFM23,CiC:LiBoy24}.
PAEKS allows a sender $\Sndr$ to compute a searchable ciphertext on $\kw$ for a receiver $\Rcvr$ using both sender's secret key $\sk_{\Sndr}$ and the receiver's public key $\pk_{\Rcvr}$.
To check if a searchable ciphertext on $\kwd$ is from a sender $\Sndrr$ intended for a receiver $\Rcvrr$, the receiver creates a trapdoor that is computed using the sender's public key $\pk_{\Sndrr}$ and the receiver's secret key $\sk_{\Rcvrr}$.
Given a ciphertext and a trapdoor, any public entity runs $\Test$ to find if $\Sndr=\Sndrr$, $\kw=\kwd$ and $\Rcvr=\Rcvrr$.

In ACISP 2021, Liu \etal \cite{ACISP:LHYSTH21} extended this primitive towards supporting multiple receivers.
In particular, they introduced the notion of \emph{broadcast authenticated encryption with keyword search} (BAEKS) to allow a sender to encrypt a keyword for a set of authorized receivers.\footnote{Obviously PAEKS is a BAEKS where the underlying authorized set is singleton. This work focuses on BAEKS and the results rightaway applies to PAEKS. See \Cref{rem:PAEKS-BAEKS} for comments on BAEKS and PAEKS relationship.}
This is followed up with a \emph{statistically consistent} BAEKS construction by Mukherjee \cite{ACISP:Mukherjee23} and a generic construction of \emph{computationally consistent} BAEKS proposed by Emura \cite{IET:Emura23}.
Both the works defined different versions of consistency notions and argued security in the standard model.
However, none of these works considered security in the multi-challenge settings which is more realistic for an \emph{authenticated encryption} typically deployed for a set of users. 
In case of BAEKS (or PAEKS), the encryption uses both secret and public keys, multi-challenge security of BAEKS (or the simpler variant PAEKS) does not necessarily follow from their single-challenge variant.
Moreover, the security models should also ensure that they withstand multi-challenge queries made in an interleaved manner.
As far we could find, there are no works on BAEKS (or for that matter PAEKS) with sufficient formalism.

% Tightness in multi-challenge security has been studied extensively in the recent past for the (Hierarchical) Identity-based Encryption \cite{C:HofJag12}.
In this work we have considered large-universe setting where number of users $|\Users|$ could be exponential in the security parameter $\secpp$. 
Thus the security reduction advantage is upper bounded by $|\Users|\cdot \epsilon$ for some negligible $\epsilon$ does not give an effective solution. 
Moreover, in a multi-user, multi-challenge security, security reductions often suffer from a nontrivial multiplicative security loss $J$ which is in the order of number of queries adversary has made.
As an example, both the most prominent BAEKS \cite{ACISP:Mukherjee23,IET:Emura23} (and PAEKS \cite{JISA:PanLi21,ESORICS:CheMen22,LuLi22,APKC:Emura22,LHHS23,InfoSc:YWYLWJH23,InfoSc:CQFM23,CiC:LiBoy24}) suffer from a multiplicative loss of $\bigTh{|\Qct|+|\Qtrp|}$ where $\Qct$ (and $\Qtrp$) is the number of ciphertext oracle queries (resp. trapdoor oracle queries).
As \cite{EC:BelRog96,C:HofJag12} has argued, we should always aim for tight security where security reductions does not suffer from such a loss.
This motivates us to explore tight security of BAEKS (and thereby of PAEKS) in the multi-user, multi-challenge setting.

\begin{remark}\label{rem:PAEKS-BAEKS}
    On a closer look, both \cite{ACISP:Mukherjee23,IET:Emura23} developed BAEKS using PAEKS although \cite{ACISP:Mukherjee23} did not mention it explicitly.
    Thus, in both \cite{ACISP:Mukherjee23,IET:Emura23} size of BAEKS ciphertext = $\bigO{|\RSet|}\times$ size of PAEKS ciphertext where $\RSet$ is the set of authorized receivers.
    This is indeed in line with state-of-the-art anonymous Broadcast encryption schemes \cite{PKC:LibPatQua12,ACNS:LiGon18} where anonymity was attained at the cost of ciphertext size linearly dependent on $\bigO{|\RSet|}$.
\end{remark}

\subsection{Motivation}
As a motivation for BAEKS functionality, let us consider an hospital where a doctor encrypts patient's health record with with regards to hospital's public key.
Using PEKS \cite{EC:BDOP04}, hospital could allow other doctor's to search for keywords on this health record by providing them trapdoors. Inquisitive insurance agents, too nosy about patient's record, can mount key-guessing attack on the trapdoors. 
To avoid this, a hospital can decide using PAEKS where encryption using hospital's public key and the encrypting doctor's secret key.
Consider a generalization of this scenario, where a doctor wishes to share this record with different doctors from other hospitals. 
This is where BAEKS turns out to be an important primitive towards solution of this problem.

% wants to check a report by a colleague who is trying to whistleblow about all the deaths due to the nipah cases in his locality.
% Obviously, the whistleblower would like to send this information to multiple officers such that the officers wouldn't know who else got this, the health department administration will not know it. 
% However, an officer who is among the intended receivers will be able to efficiently test and identify the files that keeps track of nipah deaths. 
% This basically ensures any person can whistleblow secretly without loosing anonymity.

% {\color{red}Give a picture?}

\subsection{Related Works} \label{sec:Related-Works}
Pan and Li \cite{JISA:PanLi21} introduced multi-ciphertext indistinguishability (MCI) and multi-trapdoor indistinguishability (MTI) and proposed a PAEKS secure in the random oracle model.\footnote{A discussion of these security notions is given in \Cref{sec:Different-Security_Notions}.} 
Cheng and Meng \cite{JISA:CheMen21} showed \cite{JISA:PanLi21} to be insecure and proposed a PAEKS \cite{ESORICS:CheMen22} under LWE assumption.
Liu \etal \cite{ASIACCS:LTTMC22} introduced an informal argument of consistency in PAEKS and proposed a new generic construction along with a lattice-based instantiation.
Emura \cite{APKC:Emura22} argued the generic construction of \cite{ASIACCS:LTTMC22} is not practical, formally introduced computational consistency in PAEKS and proposed a generic construction and proved it to be IND-CKA and IND-IKGA secure.\footnotemark[2] Emura also mentioned MTI security to be an open problem.  
This was followed with \cite{TIFS:LHHS23} who proposed PAEKS with updatability.
Cheng \etal \cite{InfoSc:CQFM23} proposed a new PAEKS construction under random-oracle and followed it up with \cite{TIFS:CheMen23} a server aided version for better efficiency.
Recently Li and Boyen \cite{CiC:LiBoy24} proposed a generic construction of PAEKS.
Despite being so many works in this field, the security definition is not yet standardadized \cite[Table 1]{InfoSc:CQFM23}. Moreover, the notion of consistency is an core objective of a searchable encryption \cite{JC:ABCKKL08} which a number of PAEKS constructions did not consider. Given the current status of PAEKS constructions available, we now take a look at existing BAEKS constructions to identify research gaps.

Mukherjee \cite{ACISP:Mukherjee23} introduced a statistical consistency definition for BAEKS. They also improved the adaptive security notions of \cite{ACISP:LHYSTH21} towards reality and proposed a standard model construction attaining the new stronger security notions. They also proposed a novel technique to argue statistical consistency of their construction.

Emura \cite{IET:Emura23} introduced a computational consistency definition for BAEKS extending the computational consistency definition of \cite{APKC:Emura22,IEICE:Emura24}. 
Their security model allows user corruptions\footnote{Corruption of users are pretty common in a real network scenario. 
Corruptions of users during a security game done via adversary receiving secret key of users it selected and therefore captures stronger security.} which was not allowed in \cite{ACISP:Mukherjee23}. 
However, they didn't consider sender anonymity unlike \cite{ACISP:Mukherjee23}. 
They then proposed a BAEKS construction from generic PAEKS construction of \cite{APKC:Emura22}. 
They then claim consistency and security of this construction due to consistency and security of the underlying PAEKS \cite{APKC:Emura22}. 
Our critical analysis of the consistency proof of \cite{APKC:Emura22} (and thereby \cite{IET:Emura23}) brings up several questions rendering the proof inconclusive. 
We present an analysis of the consistency proof of PAEKS \cite{APKC:Emura22} (and thereby BAEKS \cite{IET:Emura23}) in \Cref{sec:Emura-Correctness}.

In PKC2024, Ling \etal \cite{PKC:LZCHQ24} proposed two PEKS constructions in the multi-user, multi-challenge setting. 
They proposed two PEKS constructions in the composite order bilinear groups and utilized Dual system technique \cite{C:Waters09} to give a tight security proof. 
Since these constructions are PEKS constructions, they are susceptible to the dreaded keyword guessing attacks \cite{BRPL06,YHG08}.
{\color{\confcolor}In fact, our analysis of proofs of both the constructions brings up several questions rendering the proofs of \cite{PKC:LZCHQ24} incomplete. We present an analysis of the security proof of PEKS in \cite{PKC:LZCHQ24} in \Cref{sec:Ling-Security}.}

% {\color{red}Add more recent works}

% \subsection{Comparison}
{\setlength{\tabcolsep}{1.2pt}
\begin{figure}[h]
    \scriptsize
    \centering
    \begin{tabular}{|c|c|c|c|c|c|}
        \hline
        Schemes & Consistency & CT-Security & Trap-Security & Corruption & Tightness\\ \hline
        \cite{ACISP:Mukherjee23} & statistical & Confidentility & Confidentility & No & No\\
        & & $\Sndr$\&$\Rcvr$ anonymity & $\Sndrr$\&$\Rcvrr$ anonymity & & \\ \hline
        \cite{IET:Emura23} & computational\!$?$ & Confidentility & Confidentility & Adaptive & No\\
        & & $\Rcvr$ anonymity & $\Rcvrr$ anonymity & & \\ \hline
        This paper & statistical & \multicolumn{2}{|c|}{$\fullcpa$ (see \Cref{sec:def-fullcpa})} & Adaptive & Yes\\
        % & & $\Sndr$\&$\Rcvr$ anonymity & $\Sndrr$\&$\Rcvrr$ anonymity & & \\
        % & & Unforgeability & Unforgeability & & \\ 
        \hline
    \end{tabular}
    \caption{Comparison between standard model BAEKS schemes. $\Sndr$ and $\Sndrr$ anonymity stands for sender anonymity in ciphertext $(\Ct)$ and trapdoor $(\Trp)$ respectively. $\Rcvr$ and $\Rcvrr$ anonymity stands for receiver anonymity in $\Ct$ and $\Trp$ respectively. Confidentiality in CT-Security and Trap-Security respectively denotes hiding the keywords in $\Ct$ and $\Trp$ respectively. In \Cref{sec:Related-Works}, we justify that the proof of computational consistency is not complete in \cite{IET:Emura23} and therefore mark it as computational\!$?$.} %In \Cref{sec:Related-Works}, we  brief justification of .} %and comparison of different security notions is given in \Cref{sec:Different-Security}.}
    \label{fig:BAEKS-Comparison}
\end{figure}
}
% \cite{IET:Emura23} is small universe.
% 

And therefore, the following question arises naturally:
\emph{Can we obtain a multi-user BAEKS (and thereby a PAEKS) scheme in the prime order groups that is tightly secure in multi-user, multi-challenge setting under adaptive corruptions?}

\subsection{Our Contribution}
We propose a novel definition of $\fullcpa$ security for the BAEKS primitive, which improved upon the security definitions of existing PAEKS \cite{JISA:PanLi21,ESORICS:CheMen22,LuLi22,APKC:Emura22,LHHS23,InfoSc:YWYLWJH23,InfoSc:CQFM23,CiC:LiBoy24} and BAEKS \cite{ACISP:Mukherjee23,IET:Emura23}.
Looking ahead, the $\fullcpa$ security notion \emph{allows adversary to challenge on ciphertext and challenge on trapdoors in an interleaved manner while allowing adaptive corruptions}.
We emphasize that, our $\fullcpa$ security notion \emph{hides the keyword}, \emph{the (honest) sender information}, and \emph{the (honest) receiver(s) information} in both the trapdoor and the ciphertext simultaneously in presense of \emph{adaptive corruptions} of users in the system.
We also propose two unforgeability security notions $(i)$ $\ctcma$ and $(ii)$ $\trapcma$ to capture the requirement that none but the authentic user would be able to create a ciphertext or a trapdoor in this system. 
We then show that, $\fullcpa$ security already ensures both types of unforgeabilities.
% We also propose two security models, one to capture \emph{trapdoor security} and one for \emph{ciphertext security}, where both consider hiding the keyword, the sender information, and the receiver(s) information simultaneously.
% Our security models allow user public key queries not allowed in \cite{ACISP:LHYSTH21}.
% {\color{\new}We believe, these new security definitions would serve as a good reference for BAEKS which in a way generalizes PAEKS.}
Then we give a new \emph{statistically-consistent} construction of \emph{broadcast authenticated encryption with keyword search} scheme called $\baeks$ and argue its statistical consistency.
Finally, we argue \emph{tight}-$\fullcpa$ security in the standard model under standard matrix diffie-hellman-based assumptions.
We also have implemented and evaluated the construction.

\subsection{Technical Overview} %{\color{red}REWRITE. How is corruption allowed?}
In a BAEKS scheme, as introduced above informally, a ciphertext $\Ct\leftarrow \SEnc(\sk_{\Sndr},\kw,\pk_{\RSet})$ is generated wrt a sender $\Sndr$, a keyword $\kw$ and some (at most $\ell$) receivers $\RSet$ \footnote{For PAEKS, $|\RSet|=1$.} and a trapdoor $\Trp\leftarrow \Tgen(\pp,\pk_{\Sndrr},\kwd,\sk_{\Rcvrr})$ is generated wrt a sender $\Sndrr$, a keyword $\kwd$ and a receiver $\Rcvrr$.
Mukherjee \cite{ACISP:Mukherjee23} represented Sender and Receiver information as (common) variables  ($\W[\Sndr]$ and $\W[\Rcvr]$ respectively) mimicking dual system encryption \cite{EC:CheGayWee15} and encoded $\kw$ as the Boneh-Boyen Hash wrt the common variables. 
In particular, for the above scenario-example, \cite{ACISP:Mukherjee23} produces ciphertext-encodings as $\left\{(\W[\Sndr]\kw+\W[\Rcvr])\right\}_{\Rcvr\in\RSet}$ and trapdoor encoding as $(\W[\Sndrr]\kwd+\W[\Rcvrr])$. 
Observe that, if $\Sndrr=\Sndr$, $\kwd=\kw$ and $\Rcvrr\in\RSet$, there will be a \emph{match} inbetween the encodings in terms of membership (i.e. $(\W[\Sndrr]\kwd+\W[\Rcvrr])\in \left\{(\W[\Sndr]\kw+\W[\Rcvr])\right\}_{\Rcvr\in\RSet}$).
On the other hand, if one of the above equalities or membership (i.e. $(\Sndrr,\kwd,\Rcvrr)\notin \left\{(\Sndr,\kw,\Rcvr)\right\}_{\Rcvr\in\RSet}$) does not hold, the encodings are independent (even against an unbounded adversary) which is a crucial part of the proof of \cite{ACISP:Mukherjee23}.
It is easy to see that, the independence breaks down if more than two encoding instances are considered.
Therefore, Mukherjee's BAEKS construction does not seem useful for tight security in multi-users, multi-challenge settings.

Following footstep of \cite{ACISP:Mukherjee23}, we set our goal to find an ciphertext-trapdoor encoding that allows tight security argument. Intuitively, we would like to replace the Boneh-Boyen Hash encoding of \cite{ACISP:Mukherjee23} with this encoding that allows tight security argument.
We, start taking a look at existing tightly secure IBE \cite{C:BlaKilPan14,AC:HofJiaPan18}, HIBE \cite{PKC:LanPan19,PKC:LanPan20} constructions.
All these works randomize multiple encodings simultaneously due to underlying (randomized) Naor-Reingold PRF structure \cite{FOCS:NaoRei97}.
However, all these schemes depend on a trusted third party (KGC) having a master secret key. 
To get a secret key, each user has to request KGC with their identity.
%This introduces the \emph{key escrow} problem \cite{BonFra03}.
In BAEKS, each user should be able to choose their own public-private key pairs.
Moreover, every user could be a sender in certain communication and a receiver in others.

To give a brief overview of all these tightly secure (H)IBEs, built on (randomized) Naor-Reingold PRF structure, define master secret keys $\left\{\W[j,b]\right\}_{(j,b)\in[m]\times \{0,1\}}$ where $m$ is bit-length of identity. 
Ciphertext and key-encoding with respect to any identity $\id\in\{0,1\}^m$, is $\Sum{j}{[m]}\W[j,\id_j]$.
% This, therefore, demands each user to have two sets of keys: $(i)$ being a sender $(ii)$ being a receiver.
Following this, we could let each user $j$ have two keys $(i)$ $\left\{\U[j,b]\right\}_{(j,b)\in[m]\times \{0,1\}}$ for being a sender and $(ii)$ $\left\{\V[j,b]\right\}_{(j,b)\in[m]\times \{0,1\}}$ for being a receiver. 
A ciphertext-trapdoor encoding for keyword $\kywd\in\{0,1\}^{\klen}$ from the sender $\Sndr$ to a receiver $\Rcvr$ would then be $\Sum{j}{[\klen]} (\U[\Sndr,{\msgii[\kywd]{j}}]+\V[\Rcvr,{\msgii[\kywd]{j}}])$.
We could follow the blueprint of all those tightly secure (H)IBEs and inject random functions
$(\Sum{j}{[\klen]} \U[\Sndr,{\msgii[\kywd]{j}}]+\RF[]{{\Sndr}}{\kywd})+(\Sum{j}{[\klen]}\V[\Rcvr,{\msgii[\kywd]{j}}]+\RF[{\Rcvr}]{{\Rcvr}}{\kywd})$ via a hybrid argument to all the ciphertext and trapdoor encodings simultaneously. 
Finally, we could aim to replace $\RF[]{{\Sndr}}{\kywd}+\RF[{\Rcvr}]{{\Rcvr}}{\kywd}$ with random choices.
Recall that, this argument requires $\left\{\RF[\Sndr]{{\Sndr}}{\kywd}+\RF[\Rcvr]{{\Rcvr}}{\kywd}\right\}_{(\Sndr,\Rcvr)\in \SSet\times \RSet}$ is independent and therefore utlize an information theoretic step.
It is fairly easy to verify that this information theoretic step cannot be utlized in BAEKS (or PAEKS) as $\left\{\RF[\Sndr]{{\Sndr}}{\kywd}+\RF[\Rcvr]{{\Rcvr}}{\kywd}\right\}_{(\Sndr,\Rcvr)\in \SSet\times \RSet}$ for $|\SSet|,|\RSet|>2$ are not independent.
We cannot directly apply techniques from all these tightly secure (H)IBEs such as \cite{PKC:LanPan20}.

We modify the representations of users.
In our construction, we denote $j^{th}$ user as $(i)\SumiDT[i]{1}{\ulen}\U[i,{\msgii[j]{i}}]$ for being a sender $(ii)\SumiDT[i]{1}{\ulen}\V[i,{\msgii[j]{i}}]$ for being a receiver. 
% In reality, we define different projections of these keys as users' public and private key pairs.
We also let each keyword $\kw\in\{0,1\}^{\klen}$ be encoded by $\SumiDT[i]{1}{\klen}\K[i,{\msgii[\kw]{i}}]$. 
Thus, a tuple $(\Sndr,\kw,\Rcvr)$ would be encoded as $\SumiDT[i]{1}{\ulen}\U[i,{\msgii[\Sndr]{i}}]+\SumiDT[i]{1}{\klen}\K[i,{\msgii[\kw]{i}}]+\SumiDT[i]{1}{\ulen}\V[i,{\msgii[\Rcvr]{i}}]$.
This representation allows us to inject randomness into the trapdoor and ciphertexts such that $\SumiDT[i]{1}{\ulen}\U[i,{\msgii[\Sndr]{i}}]+\SumiDT[i]{1}{\klen}\K[i,{\msgii[\kw]{i}}]+\SumiDT[i]{1}{\ulen}\V[i,{\msgii[\Rcvr]{i}}]+\RF[]{}{{\Sndr||\kw||\Rcvr}}$. 
HIBE construction of \cite{PKC:LanPan19,PKC:LanPan20} intutivelly uses similar structure except we need to address different parties taking part is defining different variables. 
In particular, in our case $\U$ is chosen by the sender, $\V$ is chosen by the receiver and $\K$ is chosen by the trusted third party.
This makes the situtation a bit more complicated. 
% We let the trusted party generate the system parameters $(\On{\A},\Tw{\A},\On{\B},\Tw{\B})$ where intutivelly, $\On{\A}$ (resp. $\Tw{\B}$) will be the generator for ciphertext (resp. trapdoor) for $\On{\A}=\go^{\A}\in\Go^{\nRA\times\nCA}$ and $\Tw{\B}=\gt^{\B}\in\Gt^{\nRB\times\nCB}$ \cite{JC:EHKRV17}.
% Looking ahead, each user $j$ defines $\left\{\On{\Ut[j,\ii,b]\A},\Tw{\V[j,\ii,b]\B}\right\}_{(\ii,b)\in L\times\{0,1\}}$ as their secret keys and 
% $\left\{\On{\U[j,\ii,b]\B},\Tw{\Vt[j,\ii,b]\A}\right\}_{(\ii,b)\in L\times\{0,1\}}$ as their public keys.
% This work shows that even after different entities who chose their own secret and public  keys, contributions are 
A ciphertext on $(\Sndr,\kw,\Rcvr)$ is computed as $(\ct[{0}]={\A\s},\ct[{1}]={(\SumiDT[i]{1}{\ulen}\U[i,{\msgii[\Sndr]{i}}]+\SumiDT[i]{1}{\klen}\K[i,{\msgii[\kw]{i}}]+\SumiDT[i]{1}{\ulen}\V[i,{\msgii[\Rcvr]{i}}])\A\s})$. 
A trapdoor is computed as $(\trp[0],\trp[1])$ where $(\trp[{0}]={\B\r},\trp[{1}]={(\SumiDT[i]{1}{\ulen}\Ut[i,{\msgii[\Sndr]{i}}]+\SumiDT[i]{1}{\klen}\Kt[i,{\msgii[\kw]{i}}]+\SumiDT[i]{1}{\ulen}\Vt[i,{\msgii[\Rcvr]{i}}])\B\r})$.

% We take a look at simpler multi-challenge security where $(\Sndr,\kw,\Rcvr)$ is not queried more than once.
% To achieve multi-challenge security in such a simple setting, we could follow \cite{PKC:LanPan20} which was inspired from \cite{EC:GHKW16,AC:HofJiaPan18}.
% We could simply iterate over each bits of $\kywd\in\{0,1\}^{L}$, and inject a random function $\RF[\Usr]{}{\kywd}$ to all the ciphertext and trapdoor encodings simultaneously and then aim to replace $\RF[\Usr]{}{\kywd}$ with random choices.
% In particular, for a tuple $(\Sndr,\kw,\Rcvr)$, applying \cite{PKC:LanPan20} directly would result in $\RF[\Sndr]{}{\kw}+\RF[\Rcvr]{}{\kw}$.
% We here note that such a representation would not allow us to replace $\RF[\Sndr]{}{\kw}+\RF[\Rcvr]{}{\kw}$ by a random choice.
% Observe that $\RF[\Sndr]{}{\kw}+\RF[\Rcvr]{}{\kw}$, $\RF[\Sndr]{}{\kw}+\RF[\Rcvrr]{}{\kw}$, $\RF[\Sndrr]{}{\kw}+\RF[\Rcvr]{}{\kw}$ and $\RF[\Sndrr]{}{\kw}+\RF[\Rcvrr]{}{\kw}$ are not independent.
% Therefore, an adversary could make queries on tuples $(\Sndr,\kw,\Rcvr)$, $(\Sndr,\kw,\Rcvrr)$, $(\Sndrr,\kw,\Rcvr)$ and $(\Sndrr,\kw,\Rcvrr)$ which allow the adversary to break indistinguishability.

To withstand such a situtation, we focus on the non-triviality condition of the $\fullcpa$-security \Cref{sec:def-fullcpa} (or the Core-Lemma \Cref{fig:Core-Lemma} which reciprocates the $\fullcpa$ security model).
In particular, we make use of the restriction that $\AA$ will not query $\sk_{\Usr}$ for those users on which ciphertext and trapdoor queries are made (informally, $\Usr\in\{\RSet^{\xpz}\cup\RSet^{\xpo}\cup \{\Sndr^{\xpz},\Sndr^{\xpo},\Rcvr^{\xpz},\Rcvr^{\xpo}\}\}$). 
Looking ahead, we essentially manage to inject a random function $\RF{}{\Usr||\kywd||\Usrr}$ to all the ciphertext and trapdoor encodings of $(\Usr,\kywd,\Usrr)$ without affecting the secret keys queried.
We give a complete proof in \Cref{sec:Core}.

We now discuss adaptive corruption queries informally.
This discussion requires us to informally introduce the public and secret keys of every users as well as the public parameters of the system.
The public parameters are $\left\{\On{\Kt[\ii,b]\A},\Tw{\K[\ii,b]\B}\right\}_{\substack{{\ii\in[\klen]}\\{b\in\{0,1\}}}}$\footnote{$\On{A}=\go^A$ and $\Tw{B}=\gt^B$ where $\go,\gt$ are generators of prime order bilinear groups $\Go$ and $\Gt$ respectively.}, and each $j^{th}$ user has secret key $\sk_j=\left(\On{\Ut[j]\A},\Tw{\V[j]\B}\right)$ and has public key $\pk_j=\left(\Tw{\U[j]\B},\On{\Vt[j]\A}\right)$.
Corruption queries on user $j$ leak their secret keys $\left\{\On{\Ut[j]\A},\Tw{\V[j]\B}\right\}$  and remember that, the public keys $\left\{\On{\U[j]\B},\Tw{\Vt[j]\A}\right\}$ are already available to everybody. 
We already know that, projections $\On{\Ut[j]\A}$ and $\On{\U[j]\B}$ allow to retain some of the entropy of $\U[j]$ \cite{EC:CheGayWee15,AC:Attrapadung16}. 
On top of this, our Core-Lemma-based proof utilizes the fact that $\AA$ will not query $\sk_{\Usr}$ for those users on which challenge ciphertext and challenge trapdoor queries are made and therefore, we can insert additional randomness in $\U[j]$ and $\V[j]$ required for the proof but is not visible to the adversary. 

\subsection{Organization of the paper}
In \Cref{sec:Preliminaries}, we discuss some basic notations and descriptions of mathematical tools.
We define \emph{broadcast authenticated encryption with keyword search} and its security in \Cref{sec:Definition}. 
In \Cref{sec:Construction}, we propose a new construction $\baeks$ along with a security proof. %and present a comparison with the state of the art in \Cref{sec:Comparison}.
In \Cref{sec:Core}, we argue that the core-lemma is secure.
\Cref{sec:Comparison} provides our experimental evaluations of our construction.
Finally, we conclude the paper in \Cref{sec:Conclude}.
%We revisit the security models of \cite{ACISP:LHYSTH21} in \Cref{sec:ACISP-Model}.

% \subsection{Motivations}

% - not a PK primitive (\cite{ACISP:Mukherjee23})

% - should consider multiple challenges ct\& trap

% - What about the question of unforgeability 

% - existing schemes consistency issue

\section{Mathematical Tools and Preliminaries}
\label{sec:Preliminaries}
%!TEX spellcheck = en_US
%!TEX root = ../main.tex

\paragraph{Notations.}
For $a,b\in\mathbb{N}$ such that $a\leq b$, we often use $[a,b]$ to denote $\{a,\ldots,b\}$ and write $[b]:=[1,b]$. For a set $D$, we write $x\sample D$ to denote $x$ is sampled uniformly at random from $D$. 
The abbreviation $\ppt$ stands for probabilistic polynomial time. 
The small-case bold letters $\s$ denote a column vector, and large-case bold letters $\S=\Big(s_{ij}\Big)_{(i,j)\in I\times J}$ to denote a matrix of appropriate size. $\dforall i,j\in I$ denotes for all distinct $i$ and $j$ from $I$. We use $\dontcare$ as the \emph{don't care} symbol.
For $\tau\in\{0,1\}^L$, $\msgi[\tau]{i}$ denotes $i$-bit sub-string $\tau_1||\ldots||\tau_i$ and $\msgi[\tau]{i}||\dontcare||\ldots||\dontcare$ denotes $L$ bit string $\tau_1||\ldots||\tau_i||\dontcare||\ldots||\dontcare$ having last $L-i$ bits containing \emph{don't care} symbol $\dontcare$. 
We use $A\sqcup B$ to denote disjoint union.

% \paragraph{Games.} {\color{\confcolor}Following \cite{C:BlaKilPan14}, we use games for our security reductions. A game $G$ is defined by procedures $\Init$ and
% Finalize, plus some optional procedures P1; : : : ; Pn. All procedures are given using pseudo-code, where
% initially all variables are undefined. An adversary A is executed in game G if it first calls Initialize,
% obtaining its output. Next, it may make arbitrary queries to Pi (according to their specification), again
% obtaining their output. Finally, it makes one single call to Finalize(·) and stops. We define GA as the
% output of A’s call to Finalize.}

\subsection{Mathematical Tools and Hardness Assumptions}
\subsubsection{Bilinear Pairing.}
Let $\Go,\Gt,\GT$ are groups of order $p$ which is a large prime number. 
We consider an efficiently computable and non-degenerate bilinear pairing $e:\Go\times \Gt\rightarrow \GT$ where $\forall \go\in\Go, \gt\in\Gt$, $e({\go^a}, {\gt^b})=e(\go, \gt)^{ab}$ for any $a,b\in\Zp$. 
We use type-3 bilinear pairing in this work where $\Go$ and $\Gt$ have no known isomorphism.
We define $\ABSGen$ as the group generator that generates the type-3 bilinear pairing description.
Precisely, $\abG=(p,\Go,\Gt,\GT,\go,\gt,e)\sample \ABSGen(1^\secpp)$ such that $\Go,\Gt,\GT$ are cyclic groups of prime order $p$ where $\Go=\langle\go \rangle$, $\Gt=\langle\gt \rangle$ and $\GT=\langle \gT\rangle$ s.t. $\gT=e(\go,\gt)$.
We denote $g_s^a$ by $\Gp{a}_s$ \cite{JC:EHKRV17} for any $a\in\Zp$ and $s\in\{1,2,\textrm{T}\}$.
For $\A=\left(a_{ij}\right)\in\Zp^{\ell\times k}$, for $s\in\{1,2,\textrm{T}\}$ we denote {\small
    \[\Gp{\A}_s=
    \left(\begin{matrix}
            g_s^{a_{11}}     & \ldots & g_s^{a_{1k}}     \\
            \vdots           & \ddots & \vdots                \\
            g_s^{a_{\ell 1}} & \ldots & g_s^{a_{\ell k}} \\
        \end{matrix}\right)\in\mathbb{G}_s^{\ell\times k}.\]
}
\begin{definition}\label{eq:Dk}
    Let $\ell,k\in \mathbb{N}$, such that $\ell>k$. 
    We call $\Uk[\ell,k]$ a matrix distribution if it outputs a matrix $\M\in\Zp^{\ell\times k}$ of full rank $k$ in polynomial time (w.l.o.g. we assume the first $k$ rows of $\M\leftarrow \Uk[\ell,k]$ form an invertible matrix). 
    We write $\Uk= \Uk[k+1,k]$.
\end{definition}

\begin{definition}\label{eq:Uk}
    Let $\ell,k\in \mathbb{N}$, such that $\ell>k$. 
    We denote by $\Dk[\ell,k]$ the uniform distribution over all full-rank $\ell\times k$ matrices over $\Zp$. 
    We write $\Dk= \Dk[k+1,k]$.
\end{definition}

\subsubsection{Matrix Diffie-Hellman Assumption.} \label{eq:matDH}
For all adversary $\AA$, its advantage against the matrix Diffie-Hellman problem is defined as 
\[
\Adv{\AA,\mathbb{G}_s}{\matDH[\ell,k]}=
|\Pr[\AA(\Gp{\U}_s, \Gp{\U\matx{x}}_s)=1]-\Pr[\AA(\Gp{\U}_s, \Gp{\matx{z}}_s)=1]|
\]
where $s\in\{1,2\}$, $\U\leftarrow \Dk[\ell,k]$, $\matx{x}\sample\Zp^{k}$ and $\matx{z}\sample\Zp^{\ell}$.
The $\matDH[\ell,k]$ assumption states that $\mathsf{Adv}_{\AA,\mathbb{G}_s}^{\matDH[\ell,k]}$ is $\neglgbl$ for all $\ppt$ adversary $\AA$.
% Due to \cite{JC:EHKRV17}, we know that $\Adv{\AA,\mathbb{G}_s}{m\mhyf\matDH[\ell,k]}\leq (\ell-k)\Adv{\AA,\mathbb{G}_s}{\matDH[\ell,k]}+\frac{1}{p-1}$.
$m\mhyf$fold $\matDH[\ell,k]$ problem contains $m$ independent instances of $\matDH[\ell,k]$ problem.
Due to Random Self-Reducibility \cite{JC:EHKRV17}, we know that \[\Adv{\AA,\mathbb{G}_s}{m\mhyf\matDH[\ell,k]}\leq (\ell-k)\cdot \Adv{\AA,\mathbb{G}_s}{\matDH[k]}+\frac{1}{p-1}.\]

\subsubsection{Lateral Matrix Diffie-Hellman Assumption.} \label{eq:lmatDH}
For all adversary $\AA$, its advantage against the lateral matrix Diffie-Hellman problem is defined as \[\Adv{\AA,\mathbb{G}_s}{\lmatDH[\ell,k]}=|\Pr[\AA(\Gp{\U}_s, \Gp{\U\x}_s,\Gp{\U}_{3-s})=1]-\Pr[\AA(\Gp{\U}_s, \Gp{\z}_s,\Gp{\U}_{3-s})=1]|\] where $\U\leftarrow \Dk[\ell,k]$, $\matx{x}\sample\Zp^{k}$ and $\matx{z}\sample\Zp^{\ell}$.
The $\lmatDH[\ell,k]$ assumption states that $\mathsf{Adv}_{\AA,\abG}^{\lmatDH[\ell,k]}$ is $\neglgbl$ for all $\ppt$ adversary $\AA$.
We can define $m\mhyf$fold $\lmatDH[\ell,k]$ problem contains $m$ independent instances of $\lmatDH[\ell,k]$ problem.
% Following the random self-reducibility arguement by \cite{JC:EHKRV17}, we can easily see that $\Adv{\AA,\mathbb{G}_s}{m\mhyf\lmatDH[\ell,k]}\leq (\ell-k)\Adv{\AA,\mathbb{G}_s}{\lmatDH[k]}+\frac{1}{p-1}$.

\subsubsection{Bilateral Matrix Diffie-Hellman Assumption.} \label{eq:bilmatDH}
For all adversary $\AA$, its advantage against the bilateral matrix Diffie-Hellman problem is defined as \[\Adv{\AA,\abG}{\bmatDH[\ell,k]}=|\Pr[\AA(\On{\U}, \On{\U\matx{x}},\Tw{\U}, \Tw{\U\matx{x}})=1]-\Pr[\AA(\On{\U}, \On{\matx{z}},\Tw{\U}, \Tw{\matx{z}})=1]|\]
where $\U\leftarrow \Dk[\ell,k]$, $\matx{x}\sample\Zp^{k}$ and $\matx{z}\sample\Zp^{\ell}$.
The $\bmatDH[\ell,k]$ assumption states that $\mathsf{Adv}_{\AA,\abG}^{\bmatDH[\ell,k]}$ is negligible in $\secpp$ for all $\ppt$ adversary $\AA$.
We can define $m\mhyf$fold $\bmatDH[\ell,k]$ problem contains $m$ independent instances of $\bmatDH[\ell,k]$ problem.
% Following the random self-reducibility arguement by \cite{JC:EHKRV17}, we can easily see that $\Adv{\AA,\abG}{m\mhyf\bmatDH[\ell,k]}\leq (\ell-k)\Adv{\AA,\abG}{\bmatDH[k]}+\frac{1}{p-1}$.
Over the years, $\bmatDH[\ell,k]$ assumption has become pretty standard \cite{CCS:AgrCha17,TCC:Wee20}.

\begin{fact} 
We state a few results next. 
    \begin{enumerate}
        \item[$(\Fact1)$]{\footnotesize$\matDH[k]\Leftrightarrow \matDH[\ell,k]$ for $\ell>k$. (See \cite[Lemma~1]{EC:GHKW16})}
        \item[$(\Fact2)$]{\footnotesize$\bmatDH[k]\leq \lmatDH[k]\leq \matDH[k]$}
        \item[$(\Fact3)$]{\footnotesize$\Adv{\AA,\abG}{m\mhyf X\mhyf\matDH[\ell,k]}\leq (\ell-k)\cdot \Adv{\AA,\abG}{X\mhyf\matDH[k]}+\frac{1}{p-1}$ for $X\in\{lat,bil\}$. (Due to random self-reducibility~\cite[Lemma~1]{JC:EHKRV17})} 
        \item[$(\Fact4)$]{\footnotesize$\bmatDH[k],\ \lmatDH[k]$ are hard for $k>1$}
    \end{enumerate}
\end{fact}

\section{Broadcast Authenticated Encryption with Keyword Search}
\label{sec:Definition}
%!TEX spellcheck = en_US
%!TEX root = ../main.tex
A BAEKS scheme $\baeks$ is defined by a tuple of five $\ppt$ algorithms $(\Setup,\Kgen,$ $\SEnc,\Tgen,\Test)$ as following.
\subsection{Definition}\label{sec:BaEKS-Def}
\begin{itemize}
    \item $\Setup(1^\secpp,\ulen,\klen)$: It takes the security parameter $1^\secpp$, %$\ell=\poly$ to be the maximum number of users a ciphertext is aimed for and 
    $\ulen=\poly$ to define a userspace $\USet=\{0,1\}^{\ulen}$ and $\klen=\poly$ to define a keyword space $\KWd=\{0,1\}^{\klen}$ and
    publishes params $\pp$. We assume $\pp$ implicitly defines a ciphertext space $\CT$ and a trapdoor space $\TRP$ too. %computes a master secret key $\msk$ and
    \item $\Kgen(\pp)$: It takes the public parameter $\pp$ as input, and outputs public-private key pair $(\pk,\sk)$. It keeps $\sk$ secret and publishes $\pk$.
    \item $\SEnc(\pp,\sk_{\Sndr},\kw,\pk_{\RSet})$: It takes $\pp$, a sender secret key $\sk_{\Sndr}$, a keyword $\kw$ and public-keys $\pk_{\RSet} = \{\pk_{\Rcvr_1},\ldots, \pk_{\Rcvr_\ell}\}$ of a set of receivers $\RSet =\{\Rcvr_1,\ldots,\Rcvr_\ell\}$. %{\color{\new}where $\ell=\poly$}. 
    It outputs a ciphertext $\Ct\in\CT$.
    \item $\Tgen(\pp,\pk_{\Sndrr},\kwd,\sk_{\Rcvrr})$: It takes $\pp$, a sender public key $\pk_{\Sndrr}$, a keyword $\kwd$ and a receiver secret key $\sk_{\Rcvrr}$. It outputs a trapdoor $\Trp\in\TRP$.
    \item $\Test(\pp,\Trp,\Ct)$: It takes $\Trp$ and $\Ct$ as input, and outputs 0/1.
\end{itemize}
From now on, we drop explicit mention of $\pp$ in each algorithm description and assume all the above algorithms use $\pp$ implicitly.

\paragraph{Correctness.} 
For any security parameter $1^\secpp$, any sender $\Sndr\in\USet$, any receiver set $\RSet\subset\USet$ such that $\pk_{\RSet}=(\pk_{\Usr})_{\Usr\in\RSet}$, any keyword $\kw\in\KWd$ if $\Rcvrr\in\RSet$,
$\Pr[\Test(\Tgen(\pk_{\Sndr},\kw,\sk_{\Rcvrr}),\SEnc(\sk_{\Sndr},\kw,\pk_{\RSet}))=1]=1-\neglgbl$
where the probability is taken over $\pp\leftarrow \Setup(1^\secpp,\klen,\ulen)$ and $(\sk_{\Usr},\pk_{\Usr})\leftarrow\Kgen(\pp)$ for all the users $\Usr\in\RSet\sqcup\{\Sndr\}$.

Informally speaking, if $(\Sndrr=\Sndr)\wedge (\kwd=\kw)\wedge (\Rcvrr\in\RSet)$, we say the key-attributes $({\Sndrr},\kwd,{\Rcvrr})$ \emph{match} the ciphertext-attributes $({\Sndr},\kw,{\RSet})$. If $\Test(\Trp,\Ct)\rightarrow 1$, we say the trapdoor $\Trp$ \emph{matches} the ciphertext $\Ct$.

% \begin{center}
\begin{figure*}[h]%\hspace{-0.4cm}
    % \hspace{1cm}
    % \scriptsize
%    \begin{minipage}{1\textwidth}
        \begin{tabular}{|l|}
            \hline
                $\pp\leftarrow\Setup(1^\secpp,\klen,\ulen)$\\
                $(\pk_i,\sk_i)\leftarrow \Kgen(\pp,i)$ for $i\in \RSet_0\cup\RSet_1\cup \{\Sndr_0,\Sndr_1,\Rcvr_0,\Rcvr_1\}$ \\
                $(\kw_0,\kw_1,i,j)\leftarrow \AA\left((\pk_t)_{t\in \RSet_0\cup\RSet_1\cup \{\Sndr_0,\Sndr_1,\Rcvr_0,\Rcvr_1\}}\right)$s.t. $\kw_0,\kw_1\in\KWd\ \ $ \\
                $\qquad\qquad\qquad\qquad\qquad\wedge$ distinct $i,j\in\{0,1\} \wedge ((\Sndr_i,\kw_i,\Rcvr_i)$ \emph{does not match} $(\Sndr_j,\kw_j,\RSet_j))$ \\
                $\Ct\leftarrow \SEnc(\sk_{\Sndr_i},\kw_i,\pk_{\RSet_i})$\\
                $\Trp\leftarrow \Tgen(\pk_{\Sndr_j},\kw_j,\sk_{\Rcvr_j})$\\
                If $\Test(\Trp,\Ct)=1$, Output $1$. \\
                Otherwise, Output $0$.\\
            \hline
        \end{tabular}
%    \end{minipage}
\caption{Consistency}
\label{fig:Consistency}
\end{figure*}
% \end{center}
% \vspace{-0.6cm}

\paragraph{Consistency.}
For a BAEKS scheme $\baeks=(\Setup,\Kgen,\SEnc,\Tgen,$ $\Test)$, the consistency experiment is presented in \Cref{fig:Consistency}.
This definition of consistency of BAEKS is borrowed from \cite{ACISP:Mukherjee23}.

\subsection{Security} 
\subsubsection{Confidentiality of Ciphertexts and Trapdoors.} \label{sec:def-fullcpa}
A BAEKS scheme $\baeks$ satisfies full hiding security $(\fullcpa)$ if for all security parameter $\secpp$, for all $\ppt$ adversary $\AA$, there exists a negligible function $\neglgbl[\cdot]$ such that the advantage of $\Adv{\AA,\baeks}{\fullcpa}\leq \neglgbl$ for 
\begin{figure*}
% \hspace{-0.6cm}
\begin{equation*}
\boxed{
\Adv{\AA,\baeks}{\fullcpa}\!=\!
\left| 
\Pr\left[\bee'\!=\!\bee\!:\!
\begin{array}{l}\!
\,\pp\leftarrow\Setup(1^\secpp,\klen,\ulen), \\
\,\bee\sample \{0,1\}, \Qct,\Qtrp,\Qsk \leftarrow \emptyset, \\
\,\bee'\leftarrow\AA^{\Oct(\bee,\cdot,\cdot),\Otrp(\bee,\cdot,\cdot),\Osk(\cdot),\Opk(\cdot)}(\pp)
\end{array}
\right]
-\frac{1}{2}
\right|    
}
\end{equation*}
\caption{Full Hiding Security of BAEKS (and thereby of PAEKS)}
\label{fig:fullcpa}
\end{figure*}

\noindent where $\AA$ is given access to following oracles with a natural restriction that $({\Sndrr^{\xptb}},\kwd^{\xptb},{\Rcvr^{\xptb}})$ does not \emph{match} $({\Sndr^{\xptb}},\kw^{\xptb},{\RSet^{\xptb}})$ for $b\in\{0,1\}$ (i.e. $\Qct\cap\Qtrp=\emptyset$) and $\AA$ has not queried $\sk_{\Usr}$ for $\Usr\in\{\RSet^{\xpz}\cup\RSet^{\xpo}\cup \{\Sndr^{\xpz},\Sndr^{\xpo},\Rcvr^{\xpz},\Rcvr^{\xpo}\}\}$ (i.e. $\forall\ (u,\dontcare,\dontcare)\in\Qsk, (\pk_u,\dontcare,\dontcare),(\dontcare,\dontcare,\pk_u)\notin {\color{\confcolor}\Qtrp\cup}\Qct$). Moreover, we put the restriction that $\Otrp$ {\color{\confcolor}(and $\Osk$)} queries will have no repetition. Note that, we do not put any such restriction on $\Oct$ queries.

\begin{itemize}
    \item $\Oct$ is an oracle that on input $(\bee,(\pk_{\Sndr^{\xpz}},\kw^{\xpz},\pk_{\RSet^{\xpz}}),(\pk_{\Sndr^{\xpo}},\kw^{\xpo},\pk_{\RSet^{\xpo}}))$ outputs $\SEnc(\pp,$ $\sk_{\Sndr^{\xpb}},\kw^{\xpb},\pk_{\RSet^{\xpb}})$. For ${b\in\{0,1\}}$, $\Qct$ is updated with $\{(\pk_{\Sndr^{\xptb}},\kw^{\xptb},\pk_{\Rcvr})\}_{\Rcvr\in\RSet^{\xptb}}$. $\AA$ makes $|\Qct|$ many $\Oct$ queries.

    \item $\Otrp$ is an oracle that on input $(\bee,(\pk_{\Sndrr^{\xpz}},\kwd^{\xpz},\pk_{\Rcvrr^{\xpz}}),(\pk_{\Sndrr^{\xpo}},\kwd^{\xpo},\pk_{\Rcvrr^{\xpo}}))$ outputs $\Tgen(\pp,$ $\pk_{\Sndrr^{\xpb}},\kwd^{\xpb},\sk_{\Rcvrr^{\xpb}})$. For ${b\in\{0,1\}}$, $\Qtrp$ is updated with $\{(\pk_{\Sndrr^{\xptb}},\kwd^{\xptb},\pk_{\Rcvrr^{\xptb}})\}$. $\AA$ makes $|\Qtrp|$ many $\Otrp$ queries.

    \item {\color{\confcolor}$\Opk$ is an oracle that on new $j$ outputs $\pk_j$ consistently where $(\sk_j,\pk_j)\leftarrow\Kgen(\pp,j)$ after storing $(j,\sk_{j},\pk_{j})$ in $\Qpk$. $\AA$ makes $|\Qpk|$ many $\Opk$ queries.}

    \item {\color{\confcolor}$\Osk$ is an oracle that on new $j$ outputs $(\sk_j,\pk_j)$ consistently which it retrieve from $\Opk(j)$ after storing $(j,\sk_{j},\pk_{j})$ in $\Qsk$. 
    $\AA$ makes $|\Qsk|$ many $\Osk$ queries.}
\end{itemize} 

{We assume %$\Osk(j)$ has been queried before making 
before making a $\Otrp$ or a $\Oct$ query involving user $j$, $\AA$ will query $\Opk(j)$.} % or a $\Osk$ query
%{\color{\confcolor}Observe that, $\Qsk\subset \Qpk$ where all $j$ involved in $\Otrp$ or a $\Oct$ query, $(j,\dontcare,\dontcare)\in \Qpk\setminus \Qsk$.}
% We provide discussions on existing security definitions for confidentiality in \Cref{sec:Different-Security}. %due to space constraints.
{In this work, we also explore the question of integrity of Ciphertexts and Trapdoors and further show that $\fullcpa$ already implies such integrity. %Due to space constraints, 
We defer this discussion to \Cref{sec:Different-Security_Integrity}.
}

\subsection{Differences in Security Defintions}\label{sec:Different-Security_Notions}
Next we discuss different security notions of BAEKS (and PAEKS) that have been defined over the years.

\subsubsection{Privacy of Ciphertexts and Trapdoors.}\label{sec:Other-Privacy_Notions}

\noindent Cheng and Meng \cite{ESORICS:CheMen22} introduced \emph{Fully CI Security} and \emph{Fully TI Security}. 
\paragraph{Fully CI Security.} The challenger $\AC$ setup the keys for the target sender $(\Sndr)$ and target receiver $(\Rcvr)$ at the initialization phase of the system. 
$\AA$ receives corresponding public keys.
Then $\AA$ is allowed to query for ciphertext oracle $(\Oct)$ and trapdoor oracle $(\Otrp)$ in an interleaved manner. 
$\AA$ makes ciphertext query on $(\pk_{\Sndr},\kw,\pk_{\Rcvrr})$ to get $\Ct$ and makes trapdoor query on $(\pk_{\Sndrr},\kw,\pk_{\Rcvr})$ to get $\Trp$. 
In the challenge phase, $\AA$ chooses $\kw_0$ and $\kw_1$ with the restriction that $(\pk_{\Sndr},\kw_0,\dontcare),(\pk_{\Sndr},\kw_1,\dontcare)\notin \Qtrp$.
$\AA$ gets a challenge ciphertext $\Ct\leftarrow \SEnc(\sk_{\Sndr},\kw_{\bee},\pk_{\Rcvr})$. 
$\AA$ can make further oracle queries subject to the restriction.
Finally $\AA$ guesses $\bee'\in\{0,1\}$ and wins if $\bee'=\bee$.

\paragraph{Fully TI Security.} The challenger $\AC$ setup the keys for the target sender $(\Sndr)$ and target receiver $(\Rcvr)$ at the initialization phase of the system. 
$\AA$ receives corresponding public keys.
Then $\AA$ is allowed to query for ciphertext oracle $(\Oct)$ and trapdoor oracle $(\Otrp)$ in an interleaved manner. 
$\AA$ makes ciphertext query on $(\pk_{\Sndr},\kw,\pk_{\Rcvrr})$ to get $\Ct$ and makes trapdoor query on $(\pk_{\Sndrr},\kw,\pk_{\Rcvr})$ to get $\Trp$. 
In the challenge phase, $\AA$ chooses $\kw_0$ and $\kw_1$ with the restriction that $(\dontcare,\kw_0,\pk_{\Rcvr}),(\dontcare,\kw_1,\pk_{\Rcvr})\notin \Qct$.
$\AA$ gets a challenge trapdoor $\Trp\leftarrow \Tgen(\pk_{\Sndr},\kw_{\bee},\sk_{\Rcvr})$. 
$\AA$ can make further oracle queries subject to the restriction.
Finally $\AA$ guesses $\bee'\in\{0,1\}$ and wins if $\bee'=\bee$.

They have further defined MCI/MTI-security respectively where in the challenge phase, $\AA$ chooses two sequences of keywords $(\kw_{0,i})_{i\in I}$ and $(\kw_{1,i})_{i\in I}$ non-adaptively.
Moreover, their definitions support multiple users but considers only \emph{CT-confidentiality} and \emph{Trap-confidentiality} i.e. hiding of keyword(s). 
Interestingly, Cheng and Meng \cite{ESORICS:CheMen22} do not provide a security proof for CI-Security $\Rightarrow$ MCI-Security and TI-Security $\Rightarrow$ MTI-Security. They claim to have followed \cite{PROVSEC:QCZZ21} which also do not provide a proof for this claim.
Observe that neither of MCI and MTI security allows multiple challenges on target sender only, that too in a non-adaptive manner. 

\noindent Emura \cite{APKC:Emura22,IEICE:Emura24} in his works on PAEKS, resorted to IND-CKA and IND-IKGA Security.
\paragraph{IND-CKA Security.} The security notion is fairly similar to fully CI-security except that the challenger simulates multiple senders (of $\AC$'s choice) and the adversary $\AA$ gets public keys of all such senders. Similar to CI-security, the target receiver $(\Rcvr)$ is fixed at the start. $\AA$ can adaptively choose two keywords $\kw_0$ and $\kw_1$ and a target sender from the list provided by $\AC$.
$\AA$ can make ciphertext and trapdoor query on fixed target receiver $\Rcvr$. %Emura \cite{APKC:Emura22,IEICE:Emura24} that IND-CKA security implies 

\paragraph{IND-IKGA Security.} The security notion is exactly same as IND-CKA except that the challenge is a trapdoor. 

Emura \cite{APKC:Emura22,IEICE:Emura24} claimed that IND-CKA also implies MCI-security due to \cite{PROVSEC:QCZZ21} and IND-IKGA in the form they considered do not imply MTI-security directly.

The IND-CKA security of \cite{IET:Emura23}, in addition to the IND-CKA security of \cite{APKC:Emura22,IEICE:Emura24}, also include hiding the receiver identity and allow user corruptions. The IND-IKGA security of \cite{IET:Emura23} is exactly same as the IND-CKA security of \cite{IET:Emura23} except that the challenge is a trapdoor. 

\noindent Mukherjee \cite{ACISP:Mukherjee23} proposed $\anonctcpa$ and $\anontrapcpa$ which are quite similar to IND-CKA and IND-IKGA of Emura \cite{IET:Emura23} except that both the notions additionally capture sender anonymity but do not consider user corruptions.

\paragraph{\textbf{Full-Hiding Security} ($\fullcpa$).}
In this work, we introduce $\fullcpa$ security that allows interleaved challenge queries to ciphertext, trapdoor and secret key oracles. 
Recall that both $\SEnc$ and $\Tgen$ use secret keys and therefore the security captures multi-ciphertext and multi-trapdoor security simultaneously. 
Moreover, this security model considers keyword hiding along with sender and receiver anonymities in both the ciphertext and in the trapdoors simultaneously.
This indeed matches the security definition of \cite{TCC:SheShiWat09} therefore achieves the most coveted security.

This concludes our claim of stronger security model and we next describe the problem in the reduction essential in the proof of consistency of Emura's BAEKS \cite{IET:Emura23} and Emura's PAEKS \cite{APKC:Emura22}. Following that, we describe the system model we consider here for BAEKS.

\subsection{Analysis of Consistency of Emura's PAEKS}
\label{sec:Emura-Correctness}
%!TEX spellcheck = en_US
%!TEX root = ../main.tex

% {\color{red}COMMENT ON \cite{IEICE:Emura24}}

Emura \cite{APKC:Emura22} introduced a notion for \emph{computational consistency} for PAEKS. We reproduce the definition in \Cref{fig:Emura-Consistency} and informally describe their proof technique. 

\setlength{\tabcolsep}{2pt}
\begin{figure}
    \small
    \centering
    %    \begin{minipage}{1\textwidth}
            \begin{tabular}{|l|}
                \hline
                    $\Exp{\paeks,\AA}{\mathsf{consist}}$:\\
                    \hspace*{0.3cm}$(\pk_{\Rcvr},\sk_{\Rcvr})\leftarrow\Kgen_{\Rcvr}(1^\secpp)$\\
                    \hspace*{0.3cm}$(\pk_{\Sndr[0]},\sk_{\Sndr[0]})\leftarrow\Kgen_{\Sndr[0]}(\pk_{\Rcvr})$\\
                    \hspace*{0.3cm}$(\pk_{\Sndr[1]},\sk_{\Sndr[1]})\leftarrow\Kgen_{\Sndr[1]}(\pk_{\Rcvr})$\\
                    \hspace*{0.3cm}$(\kw,\kwd,i,j)\leftarrow \AA(\pk_{\Rcvr},\pk_{\Sndr[0]},\pk_{\Sndr[1]})$ \\
                    \hspace{0.6cm}s.t. $\kw,\kwd\in\KWd\ \ \wedge$ distinct $i,j\in\{0,1\}\wedge ((\kw,i)\neq (\kwd,j))$ \\
                    \hspace*{0.3cm}$\Ct\leftarrow \SEnc(\sk_{\Sndr[i]},\kw,\pk_{\Rcvr})$\\
                    \hspace*{0.3cm}$\Trp\leftarrow \Tgen(\pk_{\Sndr[j]},\kwd,\sk_{\Rcvr})$\\
                    \hspace*{0.3cm}If $\Test(\Trp,\Ct)=1$, Output $1$ or $0$ Otherwise.\\
                \hline
            \end{tabular}
    %    \end{minipage}
        \caption{Computational Consistency of \cite{APKC:Emura22}}
        \label{fig:Emura-Consistency}
    \end{figure}
Emura \cite{APKC:Emura22} called a PAEKS scheme $\paeks$ is computationally consistent if $\Pr[\Exp{\paeks,\AA}{\mathsf{consist}}=1]\leq \neglgbl$ for all $\ppt$ adversary $\AA$. This definition, informally speaking, captures the situation that no $\ppt$ adversary can produce two unmatching keywords  (i.e. $(\kw,\Sndr[i])\neq (\kwd,\Sndr[j])$ for distinct $i,j\in\{0,1\}$) but the corrsponding $\Ct$ and $\Trp$ matches. First we notice that some important situations are missing in the definition. In particular, the adversary does not choose the sender or the receivers. Moreover, it considered same receiver in both $\SEnc$ and $\Trp$ thereby restricting the model. However, these are only some observations about the consistency definition in \cite{APKC:Emura22} and are not the biggest issue in their consistency proof. Next we describe the proof technique informally.

Firstly, Emura \cite{APKC:Emura22} proposed a generic construction of PAEKS from a PKE, a Hash-Proof system and a PEKS. Their idea was to reduce the computational consistency of PAEKS to the computational consistency of underlying PEKS. To argue computational consistency of the construction generically, Emura proposed a hybrid argument. The argument starts from $\game{0}$ which basically is the real security game. In $\game{1}$, the projective hash function in $\SEnc$ is replaced by hash function of SPHF. In $\game{2}$, they modified key-pair computation of $\Sndr[i]$. In $\game{3}$, they sampled the sender public key in $\SEnc$ from outside the language. Then, in $\game{4}$, they again modified key-pair computation of $\Sndr[i]$. 
Then, in $\game{5}$, they again modified key-pair computation of $\Sndr[j]$. 
In $\game{6}$, they sampled the sender public key in $\Tgen$ from outside the language. In $\game{7}$ and in $\game{8}$, they replace the PRF with a random function in $\SEnc$ and $\Tgen$ respectively. Finally in $\game{9}$, they reduce the consistency of PAEKS to the consistency of underlying PEKS.

We note down some of the problems with the above argument:
\begin{enumerate}
    \item The reductions are not formally proven. We take an example where $\game{0}$ and $\game{1}$ are proven indistinguishable. Note that, they state that $|\Pr[\win{0}] - \Pr[\win{1}]|$ is negligible due to the correctness of WI-SPHF. However, this is not formally argued via forming reductions.
    
    \item In fact, none of the reductions are argued properly. Typically provable security argues that, if a given adversary breaks a particular security notion, we can define a reduction that breaks another security notions. Except $\game{1}\approx \game{2}$ and $\game{7}\approx \game{8}$, reduction arguments are missing despite being stated.
    
    \item We look at the reduction that claims to have proven $\game{1}\approx \game{2}$. To properly understand the given reduction, we look at their security model reproduced in \Cref{fig:Emura-Consistency}. Observe that, the adversary in this security model is considered successful if it produces $(\kw,\kwd,i,j)$ for $(\kw,\Sndr[i])\neq (\kwd, \Sndr[j])$ but $\Ct$ \emph{matches} $\Trp$ where  $\Ct\leftarrow \SEnc(\sk_{\Sndr[i]},\kw,\pk_{\Rcvr})$ and $\Trp\leftarrow \Tgen(\pk_{\Sndr[j]},\kwd,\sk_{\Rcvr})$. Thus, to argue $\game{1}\approx \game{2}$, the reduction $\AB$ should run $\AA$ to get such $(\kw,\kwd,i,j)$. Now, the paper claims that the reduction will modify $\pk_{\Sndr[i]}$ keeping corresponding $\sk_{\Sndr[i]}$ unchanged. In particular, till $\game{1}$, $\pk_{\Sndr[i]}$ was an encryption of $0$ and in $\game{2}$, they make it an encryption of $1$. The issue as we understand is $\pk_{\Sndr[i]}$ is used in both $\SEnc$ and in $\Trp$ and does not really behave differently based on the value that is encrypted. To summarise, once $\AA$ produces the adversarial challenge $(\kw,\kwd,i,j)$ according to \Cref{fig:Emura-Consistency}, it is not really clear how the reduction $\AB$ is going to use this challenge tuple $(\kw,\kwd,i,j)$ to break IND-CPA of the encryption scheme.
    
    \item The bigger issue we think is with the reduction to argue $\game{8}\approx \game{9}$. While arguing this indistinguishability, when $\AA$ outputs $(\kw,\kwd,i,j)$, $\AB$ randomly chooses $der\text{-}\kw,der\text{-}\kwd\sample \mathcal{KS}$ and sends $(der\text{-}\kw,der\text{-}\kwd)$ to the challenger. The rest of the argument works on $der\text{-}\kw,der\text{-}\kwd$. In this reduction we think $\AA$ is not utilized at all. In particular, $\AB$ could anyway sample $der\text{-}\kw,der\text{-}\kwd\sample \mathcal{KS}$ without any input from $\AA$ and complete the game with the challenger. This as a result shows that $\AB$ can always distinguish between $\game{8}$ and $\game{9}$. This, in our opinion, contradicts with the notion of polynomial reduction that is used to relate hardness/security of two problems/protocols.
\end{enumerate}

Recently, an expanded version of this paper was published \cite{IEICE:Emura24}. 
Although this version gives sufficient details on the proof of security, we unfortunately were not able to find any additional details over \cite{APKC:Emura22} on the proof of consistency.
Therefore, observations we have made above still persists in our opinion.

\subsection{System Model} \label{sec:BaEKS-SysDef}
The system model we consider in this work is similar to that of \cite{ACISP:Mukherjee23}. 
% {\color{red}Add more details on $|\Users|$.} 
To set a system up where the user space $\Users=\{0,1\}^{\ulen}$ for $\ulen=\poly$ and keyword space $\KWd=\{0,1\}^{\klen}$ for $\klen=\poly$, the system administrator runs $\Setup$ to generate public parameter $\pp$ and then can go offline.
Any user $\Usr\in\Users$ can generate its own public-private key pairs $(\pk_{\Usr},\sk_{\Usr})$ taking $\pp$ into account, to join the system. 
% We will assume a trusted third party generates the system params $\pp$ only one-time at the start of the setting. 
% To join this system, a user will generate its own key-pairs with regards to this system params $\pp$.
The requirement from a BAEKS scheme is that a user could encrypt a keyword for a (set) of user(s) such that any privileged user would be able to search.
To compute a searchable ciphertext of a keyword $\kw$ for a set of users (called receiver-set $\RSet=\{\Rcvr_1,\ldots,\Rcvr_\ell\}$), a user (called sender $\Sndr$) first gathers the receivers' public keys $\pk_{\RSet}=\{\pk_{\Rcvr_1},\ldots, \pk_{\Rcvr_\ell}\}$. 
The sender $\Sndr$ then encrypts a keyword $\kw$ to $\Ct$ running $\SEnc$ on its own secret key $\sk_{\Sndr}$ and on the set of privileged receivers $\pk_{\RSet}$. 
The sender $\Sndr$ stores $\Ct$ in the cloud.
Any user $\Rcvr$ can compute a trapdoor $\Trp$ using $\Tgen$ for a keyword $\kwd$ of interest with its own secret key $\sk_{\Rcvr}$ and a sender's public key $\pk_{\Sndrr}$, and send $\Trp$ to the cloud server as a search query.
The cloud server can run $\Test$ on $\Trp$ and $\Ct$ to check match without knowing neither the sender's identity nor the receiver's identity. 
A document will be returned if all the followings hold: the underlying keywords are the same $(\kw=\kwd)$, the trapdoor $\Trp$ is for querying the content from the sender $\Sndr$ (i.e. $\Sndr=\Sndrr$) rather than other senders, and the receiver $\Rcvr$ is one of the target receivers (i.e. $\Rcvr\in\RSet$) of the searchable ciphertext $\Ct$.

\section{Tight Broadcast Authenticated Encryption with Keyword Search}
\label{sec:Construction}
% \vspace{-0.25cm}
%!TEX spellcheck = en_US
%!TEX root = ../main.tex

\begin{figure*}[!h]
    % \vspace{-2cm}
    \scriptsize
    \hspace{-0cm}
    \centering
    \fbox{
        \begin{minipage}{0.49\textwidth}
            \underline{$\Setup(1^\secpp,\klen,\ulen)$}
            \begin{algorithmic}[1]
                \State $\abG=(p,\Go,\Gt,\GT,\go,\gt,e)\sample \ABSGen(1^\secpp)$
                \State $\A\sample\Dk[\nRA,\nCA],\B\sample\Dk[\nRB,\nCB]$
                \State For $(\ii,b)\in[\klen]\times\{0,1\}:$ 
                \item[]\hspace{0.5cm}$\Q[\ii,b],\N[\ii,b]\sample \Zp^{(\nRAd)\times (\nRB)}$
                \State $\K[\ii,b]=\iCol{\Q[\ii,b]}{\N[\ii,b]}$
                \State $\pp=(\On{\A},\Tw{\A},\On{\B},\Tw{\B},$
                \item[]\hfill$\left\{\On{\Kt[\ii,b]\A},\Tw{\K[\ii,b]\B}\right\}_{\substack{{\ii\in[\klen]}\\{b\in\{0,1\}}}})$
                % \item[] $
                % \item[]\hfill$
                % \item[]
            \end{algorithmic}
            \underline{$\Kgen(\pp,j)$}
            \begin{algorithmic}[1]
                \State $\W[j],\Y[j]\sample \Zp^{(\nRAu)\times (\nRB)}$
                \State {\color{\confcolorr}For $(\ii,b)\in[2\ulen]\times\{0,1\}$: $\J[j,\ii,b]\sample \Zp^{(\nRAu)\times (\nRB)}$} 
                \State {\color{\confcolorr}$\X[j]=\SumiDT{1}{\ulen}\J[j,i,{\msgii[j]{i}}],\Z[j]=\SumiDT{\ulen+1}{2\ulen}\J[j,i,{\msgii[j]{i}}]$}
                \State Set $\U[j]=\iCol{\W[j]}{\X[j]}$ and $\V[j]=\iCol{\Y[j]}{\Z[j]}$ 
                \State $\sk_j=\left(\On{\Ut[j]\A},\Tw{\V[j]\B}\right)$
                \State $\pk_j=\left(\Tw{\U[j]\B},\On{\Vt[j]\A}\right)$
                % \item[]
            \end{algorithmic}
        \end{minipage}
        \vline \hspace{1pt}
        \begin{minipage}{0.50\textwidth}
            \underline{$\SEnc(\pp,\sk_{\Sndr}, \kw, \pk_{\RSet}=(\pk_{\Rcvr_1},\ldots,\pk_{\Rcvr_{\ell}}))$}
            \begin{algorithmic}[1]
                % \State Let $\pk_{\RSet})$
                \State $\s\sample\Zp^{\nCA}$, $\perm\sample \Perm{[\ell]}$ 
                \State $\Ct=((\On{\ct[{\ij,0}]},\On{\ct[{\ij,1}]})_{\ij\in[\ell]})$ s.t. 
                \item[]$\ct[{\ij,0}]={\A\s[\ij]}$ %and  
                \item[]$\ct[{\ij,1}]={(\Sum{\ii}{[\klen]}\Kt[\ii,{\msgii[\kw]{\ii}}]+\Ut[\Sndr]+\Vt[\Rcvr_{\perm(\ij)}])\A\s[\ij]}$
            \end{algorithmic}        
            \underline{$\Tgen(\pp,\pk_{\Sndrr},\kwd,\sk_{\Rcvrr})$}   
            \begin{algorithmic}[1]
                \State $\r\sample\Zp^{\nCB}$
                \State $\Trp=(\Tw{\trp[{0}]},\Tw{\trp[{1}]})$ s.t. 
                \item[]$\trp[{0}]={\B\r}$ %and 
                \item[]$\trp[{1}]=(\Sum{\ii}{[\klen]}\K[\ii,{\msgii[\kwd]{\ii}}]+\U[\Sndrr]+\V[\Rcvrr]){\B\r}$ 
            \end{algorithmic}
            \underline{$\Test(\Trp,\Ct)$}
            \begin{algorithmic}[1]
                \item Let $\Ct=((\On{\ct[{\ij,0}]},\On{\ct[{\ij,1}]})_{\ij\in[\ell]})$ %and  
                \item Let $\Trp=(\Tw{\trp[{0}]},\Tw{\trp[{1}]})$
                \item Output $1$ if $\exists \ij\in[\ell]$ s.t. 
                \item[]\hfill$e(\On{\ct[{\ij,0}]}, \Tw{\trp[{1}]})=e(\On{\ct[{\ij,1}]},\Tw{\trp[{0}]})$
                \item Output $0$ otherwise
            \end{algorithmic}     
        \end{minipage}
    }
    \caption{Our BAEKS Construction: $\baeks$}
    \label{fig:Construction-BAEKS}
\end{figure*}

% \afterpage{\clearpage}
% \newpage\newpage

\subsection{Our Construction} \label{sec:BaEKS-Cons}
We give $\baeks$ in \Cref{fig:Construction-BAEKS}. This construction assumes a trusted third party running a one-time $\Setup$. As per the system model described in \Cref{sec:BaEKS-SysDef}, any user $j$ can run $\Kgen(j)$ to generate its key pair. %The construction at its base utlizes Naor-Reingold based affine mac \cite{C:BlaKilPan14,AC:HofJiaPan18,PKC:LanPan19,PKC:LanPan20}.
Our construction is proven correct and satistically consistent next. 
% in \Cref{sec:Correct-Consistent}. 

%!TEX spellcheck = en_US
%!TEX root = ../main.tex

\subsubsection{Correctness.} \label{sec:BAEKS-Correct}
% \paragraph{Correctness.}
Suppose, ${\Sndrr}=\Sndr$, $\kwd=\kw$ and $\Rcvrr\in\RSet$ for %=\{\Rcvr_1,\ldots,\Rcvr_\ell\}
\begin{itemize}
    \item $\Ct=(\On{\ct[{1,0}]},\On{\ct[{1,1}]},\ldots,\On{\ct[{\ell,0}]},\On{\ct[{\ell,1}]})\leftarrow \SEnc(\pp,\sk_{\Sndr},\kw,\pk_{\RSet})$
    \item $\Trp=(\Tw{\trp[{0}]},\Tw{\trp[{1}]})\leftarrow \Tgen(\pp,\pk_{{\Sndrr}},\kwd,\sk_{\Rcvrr})$,
%    \item[] \hspace{-0.5cm}such that ${\Sndrr}=\Sndr$, $\kwd=\kw$ and $\Rcvr\in\RSet=\{\Rcvr_1,\ldots,\Rcvr_\ell\}$. %where $\exists j\in[\ell]$, $\Rcvr=\Rcvr^{[j]}$. 
\end{itemize}

{%\scriptsize
\begin{equation*}
    \begin{aligned}
        A   &= e(\On{\ct[{\ij,0}]}, \Tw{\trp[{1}]})= e(\On{\A\s[\ij]},\Tw{(\Sum{\ii}{[\klen]}\K[\ii,{\msgii[\kwd]{\ii}}]+\U[\Sndrr]+\V[\Rcvrr]){\B\r}})\\
            &= e(\On{(\Sum{\ii}{[\klen]}\Kt[\ii,{\msgii[\kwd]{\ii}}]+\Ut[\Sndrr]+\Vt[\Rcvrr])\A\s[\ij]}, \Tw{\B\r})\\
    %     \end{aligned}
    % \end{equation*}
    % \begin{equation*}
    %     \begin{aligned}
            B   &= e(\On{\ct[{\ij,1}]},\Tw{\trp[{0}]})= e(\On{(\Sum{\ii}{[\klen]}\Kt[\ii,{\msgii[\kw]{\ii}}]+\Ut[\Sndr]+\Vt[{\Rcvr_{\perm(\ij)}}])\A\s[\ij]}, \Tw{\B\r})\\
    \frac{A}{B} &= \Tt{\rt\Bt (\Sum{\ii}{[\klen]}(\Kt[\ii,{\msgii[\kwd]{\ii}}]-\Kt[\ii,{\msgii[\kw]{\ii}}])+\Ut[\Sndr]-\Ut[\Sndrr]+\Vt[\Rcvr_{\perm(\ij)}]-\Vt[\Rcvrr])\A\s[\ij]}\\
    \end{aligned}
\end{equation*}
}
\noindent Since, ${\Sndrr}=\Sndr$, $\kwd=\kw$ and $\Rcvrr\in\RSet$ ($\exists j\in[\ell]$, $\Rcvrr=\Rcvr_{\tau(j)}$), $\Pr[A=B]=1$.

\subsubsection{Consistency.} \label{sec:BAEKS-Consistent}
Informally describing \Cref{fig:Consistency}, we need to show that for distinct $i,j\in\{0,1\}$, if ${\Sndr^{[j]}}\neq{\Sndr^{[i]}}$ or if $\kw^{[j]}\neq \kw^{[i]}$ or if $\Rcvr^{[j]}\notin\RSet^{[i]}$ for adversarially chosen keywords $\kw^{[j]},\kw^{[i]}\in\KWd$, a trapdoor $\Trp_j\leftarrow\Tgen(\pk_{\Sndr^{[j]}},\kw^{[j]},\sk_{\Rcvr^{[j]}})$ \emph{matches} a ciphertext $\Ct_i\leftarrow\SEnc(\sk_{\Sndr^{[i]}},\kw^{[i]},\pk_{\RSet^{[i]}})$ with only negligible probability. We emphasize that this holds even against statistical adversary that is given public keys of the senders and the receivers.

\begin{proofsketch}
    Observe that adversary $\AA$ gets params $\pp=(\On{\A},\Tw{\A},$ $\On{\B},\Tw{\B},$ $\left\{\On{\Kt[\ii,b]\A},\Tw{\K[\ii,b]\B}\right\}_{\substack{{\ii\in[\klen]}\\{b\in\{0,1\}}}})$. 
    $\AA$ also gets the public keys for $\left((\pk_t)_{t\in \RSet^{[0]}\cup\RSet^{[1]}\cup \{\Sndr^{[0]},\Sndr^{[1]},\Rcvr^{[0]},\Rcvr^{[1]}\}}\right)$.
    It therefore has access to $\pk_t = (\Tw{\U[t]\B},\On{\Vt[t]\A})$ where $\U[t],\V[t]\in\Zp^{\nRA\times\nRB}$ are defined as per \Cref{fig:Construction-BAEKS} by the experiment for all $t\in \RSet^{[0]}\cup\RSet^{[1]}\cup \{\Sndr^{[0]},\Sndr^{[1]},\Rcvr^{[0]},\Rcvr^{[1]}\}$.
    We show that, given all these information, there is enough entropy in the corresponding secret keys $\left((\sk_t)_{t\in \RSet^{[0]}\cup\RSet^{[1]}\cup \{\Sndr^{[0]},\Sndr^{[1]},\Rcvr^{[0]},\Rcvr^{[1]}\}}\right)$ to show that the best $\AA$ can do is to guess despite $\AA$ being an unbounded adversary.
\end{proofsketch}

To formulate our argument, we first check $\Kgen$ where $\U[t]$ and $\V[t]$ are sampled in a particular way for an input $t$. Observe that, $\W[t],\Y[t],\J[t,\ii,b]\in\Zp^{(\nRAu)\times (\nRB)}$ are all sampled uniformly at random for $\ii\in [\ulen]$ and $b\in\{0,1\}$. Therefore, $\U[t],\V[t]\in\Zp^{(\nRA)\times (\nRB)}$ are distributed uniformly at random.

We sample $\K[\ii,b]=\tK[\ii,b]+\tk[\ii,b]\T$ for all $\ii\in [\klen]$ and $b\in\{0,1\}$. We further sample $\U[t]=\tU[t]+\tu[t]\T$ and $\V[t]=\tV[t]+\tv[t]\T$ for $\tU[t],\tV[t],\T\sample\Zp^{\nRA\times\nRB}$, $\tu[t],\tv[t]\sample\Zp$ for all $t\in \RSet^{[0]}\cup\RSet^{[1]}\cup \{\Sndr^{[0]},\Sndr^{[1]},\Rcvr^{[0]},\Rcvr^{[1]}\}$. %, $\ii\in L$. %and $b\in\{0,1\}$. %and $\mu=\Apt\Bp\neq 0$.
Since, $\tU[t]$ and $\tV[t]$ are uniformly random, $\U[t]$ and $\V[t]$ are uniformly random for all $t\in \RSet^{[0]}\cup\RSet^{[1]}\cup \{\Sndr^{[0]},\Sndr^{[1]},\Rcvr^{[0]},\Rcvr^{[1]}\}$. %, $\ii\in L$ and $b\in\{0,1\}$.
Thus, for all $t\in \RSet^{[0]}\cup\RSet^{[1]}\cup \{\Sndr^{[0]},\Sndr^{[1]},\Rcvr^{[0]},\Rcvr^{[1]}\}$, $\pk_t$ are properly distributed and do not leak any information about the corresponding $\tu[t],\tv[t]$.
We sample $\r\sample\Zp^k$. The trapdoor $\Trp_j=(\trp[{0}]^{[j]},\trp[{1}]^{[j]})$ is then,

{\scriptsize
\begin{equation*}
    \begin{aligned}
        \trp[{0}]^{[j]}&=\Tw{\B\r} \\
        \trp[{1}]^{[j]}&=\Tw{(\Sum{\ii}{[\klen]}\K[\ii,{\msgii[\kw]{\ii}^{[j]}}]+\U[\Sndr^{[j]}]+\V[\Rcvr^{[j]}])\B\r+\widetilde{k_1}^{[j]}\T\B\r}
        %\\
        %&\qquad\qquad\qquad\qquad\quad
        \text{for }\widetilde{k_1}^{[j]}=\boxed{\Sum{\ii}{[\klen]}\tk[\ii,{\msgii[\kw]{\ii}^{[j]}}]+\tu[\Sndr^{[j]}]+\tv[\Rcvr^{[j]}]}\\
    \end{aligned}.
\end{equation*}
}
%  \[\Trp_j=(\trp[{0}]^{[j]}=\Tw{\B\r},\trp[{1}]^{[j]}=\Tw{\Sum{\ii}{[L]}(\U[\Sndr^{[j]},\ii,\kw_{\ii}^{[j]}]+\V[\Rcvr^{[j]},\ii,\kw_{\ii}^{[j]}])\B\r+\boxed{\Sum{\ii}{[L]}(\tu[\Sndr^{[j]},\ii,\kw_{\ii}^{[j]}]+\tv[\Rcvr^{[j]},\ii,\kw_{\ii}^{[j]}])}\T\B\r}).\]

On the other hand, $\Ct_i=((\ct[{\ij,0}]^{[i]},\ct[{\ij,1}]^{[i]})_{\ij\in[\ell]})$ for $\RSet^{[i]}=\{\Rcvr_1^{[i]},\ldots,\Rcvr_{\ell}^{[i]}\}$, a random permutation $\tau:[\ell]\rightarrow[\ell]$ and $\s[\ij]\sample\Zp^{k}$ such that,

{\scriptsize
\begin{equation*} 
    \begin{aligned}
        \ct[{\ij,0}]^{[i]}&=\On{\A\s[\ij]} \\
        \ct[{\ij,1}]^{[i]}&=\On{(\Sum{\ii}{[\klen]}\Kt[\ii,{\msgii[\kw]{\ii}^{[i]}}]+\Ut[\Sndr^{[i]}]+\Vt[\Rcvr_{\perm(\ij)}^{[i]}])\A\s[\ij]+\widetilde{c_{\ij,1}}^{[i]}\T^\top\A\s[\ij]}
        % \\
        % &\qquad\qquad\qquad\qquad\quad
        \text{for }\widetilde{c_{\ij,1}}^{[i]}=\boxed{\Sum{\ii}{[\klen]}\tk[\ii,{\msgii[\kw]{\ii}^{[i]}}]+\tu[\Sndr^{[i]}]+\tv[\Rcvr_{\perm(\ij)}^{[i]}]} \\
    \end{aligned}.
\end{equation*}
}
We now focus at the boxed part of the above equations. Let $\sigma_l=\widetilde{c_{\ij,1}}^{[i]}= \Sum{\ii}{[\klen]}\tk[\ii,{\msgii[\kw]{\ii}^{[i]}}]+\tu[\Sndr^{[i]}]+\tv[\Rcvr_{\perm(\ij)}^{[i]}]$ for all $l\in[\ell]$ and $\sigma_{\ell+1}=\widetilde{k_1}^{[j]}=\Sum{\ii}{[\klen]}\tk[\ii,{\msgii[\kw]{\ii}^{[j]}}]+\tu[\Sndr^{[j]}]+\tv[\Rcvr^{[j]}]$.

\begin{description}
    \item[\textbf{Case-1}] Consider the case $\Sndr^{[j]}\neq \Sndr^{[i]}$. 
    Notice that, $\tu[\Sndr^{[j]}]$ and $\tu[\Sndr^{[i]}]$ are sampled uniformly at random, $\sigma_{\ell+1}$ is independent to all $\{\sigma_l\}_{l\in[\ell]}$. Moreover, $\dforall a,b\in[\ell]$, $\tv[\Rcvr_{\perm(a)}^{[i]}]$ and $\tv[\Rcvr_{\perm(b)}^{[i]}]$ are sampled uniformly at random, $\sigma_{a}$ is independent to all $\sigma_b$.
    Therefore, $\sigma_1,\ldots,\sigma_{\ell+1}$ are independent and uniformly random.

\item[\textbf{Case-2}] Consider the case $(\Sndr=\Sndr^{[j]}=\Sndr^{[i]})\wedge (\kw^{[j]} \neq \kw^{[i]})$. Since, $\kw^{[j]}\neq \kw^{[j]}$, $\exists d\in[\klen]$ s.t. $\kw_d^{[j]}\neq \kw_d^{[i]}$. Notice that, $\tk[d,\kw_d^{[j]}]$ and $\tk[d,\kw_d^{[i]}]$ are sampled uniformly at random, $\sigma_{\ell+1}$ is independent to all $\{\sigma_l\}_{l\in[\ell]}$. Moreover, $\dforall a,b\in[\ell]$, $\tv[\Rcvr_{\perm(a)}^{[i]}]$ and $\tv[\Rcvr_{\perm(b)}^{[i]}]$ are sampled uniformly at random, $\sigma_{a}$ is independent to all $\sigma_b$.
Therefore, $\sigma_1,\ldots,\sigma_{\ell+1}$ are independent and uniformly random.

\item[\textbf{Case-3}] Consider the case $(\Sndr=\Sndr^{[j]}=\Sndr^{[i]})\wedge (\kw=\kw^{[j]}=\kw^{[i]}) \wedge (\Rcvr^{[j]} \notin \RSet^{[i]})$. Since, $\Rcvr^{[j]} \notin \RSet^{[i]}$, $\forall d\in[\ell]$, $\Rcvr^{[j]}\neq \Rcvr_{\perm(d)}^{[i]}$. 

{\scriptsize
\vspace{-0.5cm}  
    \[%\hspace{-2.0cm}
    \underset{{(\ell+1)\times 1}}{\underbrace{
        \begin{bmatrix}
            \sigma_1        \\
            \sigma_2        \\
            \vdots          \\
            \sigma_{\ell+1} \\
        \end{bmatrix}
    }}
        =
        \underset{{(\ell+1)\times (\ell+2)}}{\underbrace{
        \begin{bmatrix}
            1 &   1  &  0   &\cdots&  0   &  0\\
            1 &   0  &  1   &\cdots&  0   &  0\\
            \vdots&\vdots&\vdots&\ddots&\vdots&\vdots\\
            1 &   0  &  0   &\cdots&  1   &  0 \\
            1 &   0  &  0   &\cdots&  0   &  1 \\
        \end{bmatrix}}}
        \times
        \underset{{(\ell+2)\times 1}}{\underbrace{
        \begin{bmatrix}
            \Sum{\ii}{[\klen]}\tk[\ii,{\msgii[\kw]{\ii}^{[i]}}]+\tu[\Sndr]\\
            \tv[\Rcvr_{\perm(1)}^{[i]}]\\
            \tv[\Rcvr_{\perm(2)}^{[i]}]\\
            \vdots \\
            \tv[\Rcvr_{\perm(\ell)}^{[i]}]\\
            \tv[\Rcvr^{[j]}]\\
        \end{bmatrix}}}%_{(\ell+2)\times 1}
    \]
}
   
The above relation between $\sigma_1,\ldots,\sigma_{\ell+1}$ is represented as a system of the linear equation where the linear transformation has rank $(\ell+1)$.
Since, $\tk[\ii,{b^{[i]}}],\tu[\Sndr]$ are all chosen uniformly at random $\ii\in [\klen]$ and $b\in\{0,1\}$. Thus, $\Sum{\ii}{[\klen]}\tk[\ii,{\msgii[\kw]{\ii}^{[i]}}]+\tu[\Sndr]$ is also distributed unformly at random.
Observe that, $\Sum{\ii}{[\klen]}\tk[\ii,{\msgii[\kw]{\ii}^{[i]}}]+\tu[\Sndr],$ $\tv[\Rcvr_{\perm(1)}^{[i]}],$ $\ldots,\tv[\Rcvr_{\perm(\ell)}^{[i]}],\tv[\Rcvr^{[j]}]$ are independent and uniformly random, $\sigma_1,\ldots,\sigma_{\ell+1}$ are too independent and uniformly random.
\end{description}

\noindent In other words, $\sigma_1,\ldots,\sigma_{\ell+1}$ are completely uncorelated and as a result $\Ct_i$ and $\Trp_j$ considered in the consistency game \emph{does not match} except with negligible probability.

\subsection{Security}\label{sec:BaEKS-SecThm} 
Next, we argue that our construction $\baeks$ is adaptively $\fullcpa$ secure.
\begin{theorem}\label{thm:full-security}
    If the $\lmatDH$ Assumption holds in $\abG$, then $\baeks$ is a $\fullcpa$-secure BAEKS scheme. Precisely, for any $\ppt$ adversary $\AA$ that breaks $\baeks$ in the $\fullcpa$ model making at most $q$ challenge trapdoor queries, at most $\qq$ challenge encryption requests and at most $\qqq$ key requests, there exists $\ppt$ algorithms $\AB_1,\AB_2$ such that,
    \begin{equation*}
        \begin{aligned}
            \Adv{\AA,\baeks}{\fullcpa}&\leq \Adv{\AB,\Go}{|\Qct|\ell\mhyf\lmatDH[\nRA,\nCA]}+ \Adv{\AB,\abG}{\Core}\\
            &\qquad +\Adv{\AB,\Gt}{|\Qtrp|\mhyf\lmatDH[\nRB,\nCB]} +\frac{4}{p-1}\\
            &\leq \nRCA\Adv{\AB,\Go}{\lmatDH[\nCA]}+ \Adv{\AB,\abG}{\Core}\\
            & +\nRCB\Adv{\AB,\Gt}{\lmatDH[\nCB]} +\frac{4}{p-1}\quad (\text{due to }\Fact3)\\
        \end{aligned}
    \end{equation*}
\end{theorem}

We use devise a hybrid argument to argue the $\fullcpa$ security of $\baeks$. We define $\game{0},\game{1},\ldots, \game{6}$ in \Cref{fig:Full-Hybrid} to argue the coveted $\fullcpa$ security. We describe the games and provide the security argument next. %in \Cref{sec:Full-Proofs} due to space limitation.
% Looking ahead, $\Adv{\AB,\abG}{\Core}$ is tightly secure in \Cref{thm:Core}. Since $\Adv{\AB,\Gt}{\lmatDH[\nCB]}$ is tightly secure and $k$ is prefixed security constant,   
Thus, $\Adv{\AA,\baeks}{\fullcpa}$ above is also tight.
%!TEX spellcheck = en_US
%!TEX root = ../main.tex

\subsection{Full Security Argument}\label{sec:BaEKS-FullSec}
% \begin{proofsketch}[Proof of \Cref{thm:full-security}]

% \end{proofsketch}

\begin{figure*}
    \scriptsize
    %\centering
    \hspace{-0.7cm} %{-3cm}
    \fbox{
        %\parbox[c][150mm][c]{1.10\textwidth}{
        %\centering\hspace{-2.5cm}$\game{1}$,\lfbox[dotted]{$\game{2}$,\lfbox[background-color=lightgray!60!white,border-color=lightgray!60!white]{$\game{3}$},\lfbox[background-color=lightgray!60!white,border-color=black]{$\game{4}$},\lfbox{$\game{5}$}}\\ \ \\
        \begin{minipage}{0.44\textwidth} %{0.53\textwidth}
            \underline{$\Init(1^\secpp,\klen,\ulen)$}\ // \lfbox[background-color=lightgray!60!white,border-color=lightgray!60!white]{$\game{3}\mhyf\game{6}$}
            \begin{algorithmic}[1]
                \State $\abG=(p,\Go,\Gt,\GT,\go,\gt,e)\sample \ABSGen(1^\secpp)$
                \State $\A\sample\Dk[\nRA,\nCA]$ \lfbox[background-color=lightgray!60!white,border-color=lightgray!60!white]{s.t. $\Au\in\Zp^{\nRAu\times \nCB}$ is invertible}
                \State $\B\sample\Dk[\nRB,\nCB]$
                \State {\color{\confcolorr}For $(\ii,b)\in[\klen]\times\{0,1\}:$}
                \item[]\hspace{0.2cm}$\Q[\ii,b],\N[\ii,b]\sample \Zp^{(\nRAd)\times (\nRB)}$ %$\qquad\bslash \hame{1}-\hame{4}$
                \item[]\hspace{0.2cm}$\K[\ii,b]=\iCol{\Q[\ii,b]}{\N[\ii,b]}$
                \item[]\lfbox[rounded]{$\Kt[\ii,b]\A\sample \Zp^{\nRB\times \nCA}$} \hfill  $\bslash \game{4}\mhyf\game{6}$
                \item[]\lfbox[background-color=lightgray!60!white,border-color=lightgray!60!white]{$\K[\ii,b]\B=\iCol{\Au^{-\top}((\Kt[\ii,b]\A)^\top\B-\Adt\N[\ii,b]\B)}{\N[\ii,b]\B}$} 
                \item[] 
                \State $\pp=(\On{\A},\Tw{\A},\On{\B},\Tw{\B},$
                \item[]\hfill$\left\{\On{\Kt[\ii,b]\A},\Tw{\K[\ii,b]\B}\right\}_{\substack{{\ii\in[\klen]}\\{b\in\{0,1\}}}})$
                \item[]
                \item[]
                \item[]
            \end{algorithmic}
            \underline{$\OKgen(j)$}\ // \lfbox[background-color=lightgray!60!white,border-color=lightgray!60!white]{$\game{3}\mhyf\game{6}$}
            \begin{algorithmic}[1]
                \State $\W[j],\Y[j]\sample \Zp^{(\nRAu)\times (\nRB)}$  \hfill $\bslash \game{0}\mhyf\game{2}$
                \State {\color{\confcolorr}For $(\ii,b)\in[2\ulen]\times\{0,1\}$: 
                \item[]\hspace{0.5cm}$\J[j,\ii,b]\sample \Zp^{(\nRAu)\times (\nRB)}$}  \hfill $\bslash \game{0}\mhyf\game{2}$ 
                \State {\color{\confcolorr}$\X[j]=\SumiDT{1}{\ulen}\J[j,i,{\msgii[j]{i}}],\Z[j]=\SumiDT{\ulen+1}{2\ulen}\J[j,i,{\msgii[j]{i}}]$} \hfill $\bslash \game{0}\mhyf\game{2}$ %\in \Zp^{(\nRAd)\times (\nRB)}$

                \State $\U[j]=\iCol{\W[j]}{\X[j]},\V[j]=\iCol{\Y[j]}{\Z[j]}$ \hfill  $\bslash \game{0}\mhyf\game{2}$   
                \item[]\lfbox[rounded]{$\Ut[j]\A,\Vt[j]\A\sample \Zp^{\nRB\times \nCA}$} \hfill  $\bslash \game{4}\mhyf\game{6}$
                \item[]\lfbox[background-color=lightgray!60!white,border-color=lightgray!60!white]{$\U[j]\B=\iCol{\Au^{-\top}((\Ut[j]\A)^\top\B-\Adt\X[j]\B)}{\X[j]\B}$} 
                \item[]\lfbox[background-color=lightgray!60!white,border-color=lightgray!60!white]{$\V[j]\B=\iCol{\Au^{-\top}((\Vt[j]\A)^{\top}\B-\Adt\Z[j]\B)}{\Z[j]\B}$}  
                \State $\sk_j=(\left(\On{\Ut[j]\A},\Tw{\V[j]\B}\right))$
                % \State {$\sk_j=(\left(\On{\Ut[j,\ii,b]\A},\Tw{\iCol{\Au^{-\top}(\V[j,\ii,b]^{\top}\B-\Adt\Z[j,\ii,b])}{\Z[j,\ii,b]\B}}\right)_{\substack{{\ii\in[\klen]}\\{b\in\{0,1\}}}})$}
                \State $\pk_j=(\left(\Tw{\U[j]\B},\On{\Vt[j]\A}\right))$
                % \State \lfbox[background-color=lightgray!60!white,border-color=lightgray!60!white]{$\pk_j=(\left(\Tw{\U[j,\ii,b]\B},\On{\Vt[j,\ii,b]\A}\right)_{\substack{{\ii\in[\klen]}\\{b\in\{0,1\}}}})$}
                \item[]
                \item[]
                \item[]
                \item[]
            \end{algorithmic}
            \underline{$\Finalize(\bee\in\{0,1\})$}
            \begin{algorithmic}[1]
                \State Return $\bee$
                \item[]\hspace{0.2cm}$\wedge\ (\Qtrp\cap\Qct=\emptyset)$
                \item[]\hspace{0.2cm}$\wedge\ (\forall\ (u,\dontcare,\dontcare)\in\Qsk,$ 
                \item[]\hfill $(\pk_u,\dontcare,\dontcare),(\dontcare,\dontcare,\pk_u)\notin \Qtrp\cup\Qct)$
            \end{algorithmic}
            % \underline{$\Test(\Trp,\Ct)$}
            % \begin{algorithmic}[1]
            %     \item Let $\Ct=((\On{\ct[{\ij,0}]},\On{\ct[{\ij,1}]})_{\ij\in[\ell]})$
            %     \item Let $\Trp=(\Tw{\trp[{0}]},\Tw{\trp[{1}]})$
            %     \item Output $1$ if $\exists \ij\in[\ell]$ s.t. 
            %     \item[]\hspace{0.35cm}$e(\On{\ct[{\ij,0}]}, \Tw{\trp[{1}]})=e(\On{\ct[{\ij,1}]},\Tw{\trp[{0}]})$
            %     \item Output $0$ otherwise
            %     \item[]
            %     \item[]
            %     \item[]
            % \end{algorithmic}     
        \end{minipage}
        \vline \hspace{1pt}
        \begin{minipage}{0.9\textwidth}
            \underline{$\OSEnc(\sk_{\Sndr}, \kw, \pk_{\RSet})$}\ //$\game{1}$,\lfbox[dotted]{$\game{2}$,\lfbox[background-color=lightgray!60!white,border-color=lightgray!60!white]{$\game{3}$},\lfbox[rounded]{$\game{4}$},\lfbox[background-color=lightgray!60!white,border-color=black]{$\game{5}$},\lfbox{$\game{6}$}}
            \begin{algorithmic}[1]
                \State Let $\pk_{\RSet}=(\pk_{\Rcvr_1},\ldots,\pk_{\Rcvr_{\ell}})$
                \State $\perm\sample \Perm{[\ell]}$ 
                \State $\Ct=((\On{\ct[{\ij,0}]},\On{\ct[{\ij,1}]})_{\ij\in[\ell]})$ s.t.
                \item[]$\ct[{\ij,0}]={\A\s[\ij]}$ for $\s[\ij]\sample\Zp^{\nCA}$
                \item[]\lfbox[dotted]{$\ct[{\ij,0}]\sample \Zp^{\nRA}$}
                \item[]\lfbox[background-color=lightgray!60!white,border-color=lightgray!60!white]{\lfbox[rounded]{$\ctu_{\ij,0}\sample\Zp^{\nRAu},\a[\ij,0]\sample\Zp^{\nRAd}, \ctd_{\ij,0}=\a[\ij,0]+\Ad\Au^{-1}\ctu_{\ij,0}$}}, \lfbox[background-color=lightgray!60!white,border-color=black]{$\a[\Sndr,\kw,\Rcvr_{\perm(\ij)}]\sample \Zp^{\nRB}$}
                \item[]$\ct[{\ij,1}]={(\Sum{\ii}{[\klen]}\Kt[\ii,{\msgii[\kw]{\ii}}]+\Ut[\Sndr]+\Vt[\Rcvr_{\perm(\ij)}])\ct[{\ij,0}]}$ 
                \item[]\lfbox[background-color=lightgray!60!white,border-color=lightgray!60!white]{\lfbox[rounded]{$\ct[{\ij,1}]=(\Sum{\ii}{[\klen]}\Kt[\ii,{\msgii[\kw]{\ii}}]+\Ut[\Sndr]+\Vt[\Rcvr_{\perm(\ij)}]){\A}\Au^{-1}\ctu_{\ij,0}+(\Sum{\ii}{[\klen]}\Nt[\ii,{\msgii[\kw]{\ii}}]+\Xt[\Sndr]+\Zt[\Rcvr_{\perm(\ij)}])\a[\ij,0]$}}
                \item[]\lfbox[background-color=lightgray!60!white,border-color=black]{$\ct[{\ij,1}]=(\Sum{\ii}{[\klen]}\Kt[\ii,{\msgii[\kw]{\ii}}]+\Ut[\Sndr]+\Vt[\Rcvr_{\perm(\ij)}]){\A}\Au^{-1}\ctu_{\ij,0}+\a[\Sndr,\kw,\Rcvr_{\perm(\ij)}]$}
                \item[]\lfbox{$\ct[{\ij,1}]\sample\Zp^{\nRB}$}
                \item[] 
            \end{algorithmic}        
            \underline{$\OTGen(\pk_{\Sndrr},\kwd,\sk_{\Rcvrr})$}\ //$\game{1}$,\lfbox[dotted]{$\game{2}$,\lfbox[background-color=lightgray!60!white,border-color=lightgray!60!white]{$\game{3}$},\lfbox[background-color=lightgray!60!white,border-color=black]{$\game{4}$},\lfbox{$\game{5}$}}   
            \begin{algorithmic}[1]
                \State $\r\sample\Zp^{\nCB}$
                \State $\Trp=(\Tw{\trp[{0}]},\Tw{\trp[{1}]})$ where
                \item[]$\trp[{0}]={\B\r}$
                \item[]\lfbox[dotted]{$\trp[{0}]\sample \Zp^{\nRB}$}, \lfbox[background-color=lightgray!60!white,border-color=black]{$\d[\Sndrr,\kwd,\Rcvrr]\sample \Zp^{\nRAd}$}
                \item[]$\trp[{1}]={(\Sum{\ii}{[\klen]}\K[\ii,{\msgii[\kwd]{\ii}}]+\U[\Sndrr]+\V[\Rcvrr])\trp[{0}]}$
                \item[]\hspace{0.5cm}{$=\iCol{\trpu_1}{\trpd_1}=\iCol{(\Sum{\ii}{[\klen]}\Q[\ii,{\msgii[\kwd]{\ii}}]+\W[\Sndrr]+\Y[\Rcvrr])\trp[{0}]}{(\Sum{\ii}{[\klen]}\N[\ii,{\msgii[\kwd]{\ii}}]+\X[\Sndrr]+\Z[\Rcvrr])\trp[{0}]}$}
                \item[]\lfbox[background-color=lightgray!60!white,border-color=lightgray!60!white]{\lfbox[rounded]{
                    $\trp[{1}]=\iCol{\trpu_1}{\trpd_1}=\iCol{\Au^{-\top}\At(\Sum{\ii}{[\klen]}\K[\ii,{\msgii[\kwd]{\ii}}]+\U[\Sndrr]+\V[\Rcvrr])\trp[{0}]-\Au^{-\top}\Adt\trpd_1}{(\Sum{\ii}{[\klen]}\N[\ii,{\msgii[\kwd]{\ii}}]+\X[\Sndrr]+\Z[\Rcvrr])\trp[{0}]}$
                }}
                \item[]\lfbox[background-color=lightgray!60!white,border-color=black]{
                    $\trp[{1}]=\iCol{\trpu_1}{\trpd_1}=\iCol{\Au^{-\top}\At(\Sum{\ii}{[\klen]}\K[\ii,{\msgii[\kwd]{\ii}}]+\U[\Sndrr]+\V[\Rcvrr])\trp[{0}]-\Au^{-\top}\Adt\trpd_1}{\d[\Sndrr,\kwd,\Rcvrr]}$}
                \item[]\lfbox{$\trp[{1}]\sample\Zp^{\nRA}$}    
            \end{algorithmic}
        \end{minipage}
        %}
    }
    \caption{Hybrids for the transition from $\game{1}$ to $\game{6}$.}
    \label{fig:Full-Hybrid}
\end{figure*}

\begin{proof}
    We argue the proof of \Cref{thm:full-security} via a hybrid argument. We sequence the games first and follow it by lemmas that acts as a reduction to so-called Core Lemma \Cref{sec:Core}.
    \begin{itemize}
        \item $\game{0}$: The real situation.
        \item $\game{1}$: For $\ij\in[\ell]$, use $\ct[{\ij,0}]$ to compute $\ct[{\ij,1}]$ {\color{\confcolor}and $\trp[{0}]$ to compute $\trp[{1}]$}.
        \item $\game{2}$: Sample $\ct[{\ij,0}]$ for $\ij\in[\ell]$ {\color{\confcolor}and $\trp[{0}]$} uniformly at random.
        \item $\game{3}$: Replace usage of $\Q[\ii,b],\W[j]$ and $\Y[j]$, by $(\Q[\ii,b],\N[\ii,b]),(\U[j],\X[j])$ and $(\V[j],\Z[j])$ respectively.
        \item $\game{4}$: Sample $\Qt[\ii,b]\A,\Ut[j]\A$ and $\Vt[j]\A$ uniformly at random.
        \item $\game{5}$: We remove effect of $\N[\ii,b],\X[j]$ and $\Z[j]$ in both $\OSEnc$ and $\OTGen$ responses.
        \item $\game{6}$: $\{\ct[{1,1}],\ldots,\ct[{\ell,1}]\}$ and $\trp[{1}]$ are all random quantities of appropriate dimensions.
    \end{itemize}
\end{proof}

\subsection{Lemmas for Full Security of BAEKS}\label{sec:BaEKS-FullSec-Lemmas}

\begin{lemma}$(\game{0}$ to $\game{1})$\label{lem:G{0}-G{1}} For any adversary $\AA$, $\adv{\AA}{\game{0}}=\adv{\AA}{\game{1}}.$
\end{lemma}

\begin{lemma}$(\game{1}$ to $\game{2})$\label{lem:G{1}-G{2}}
    For any $\ppt$ adversary $\AA$ in the $\fullcpa$ security model, there exists an $\ppt$ adversary $\AB$ such that 
    
    $|\adv{\AA}{\game{1}}-\adv{\AA}{\game{2}}|\leq \Adv{\AB,\Go}{|\Qct|\ell\mhyf\lmatDH[\nRA,\nCA]}+\Adv{\AB,\Gt}{|\Qtrp|\mhyf\lmatDH[\nRB,\nCB]}+\frac{2}{p-1}.$
\end{lemma}

\begin{lemma}$(\game{2}$ to $\game{3})$\label{lem:G{2}-G{3}}
    For any adversary $\AA$, $|\adv{\AA}{\game{2}}-\adv{\AA}{\game{3}}|\leq \frac{1}{p-1}.$
\end{lemma}

\begin{lemma}$(\game{3}$ to $\game{4})$\label{lem:G{3}-G{4}}
    For any adversary $\AA$,
    $|\adv{\AA}{\game{3}}-\adv{\AA}{\game{4}}|\leq \frac{1}{p-1}.$
\end{lemma}

\begin{lemma}$(\game{4}$ to $\game{5})$\label{lem:G{4}-G{5}}
    For any $\ppt$ adversary $\AA$ in the $\fullcpa$ security model, there exists an $\ppt$ adversary $\AB$ such that
    $|\adv{\AA}{\game{4}}-\adv{\AA}{\game{5}}|\leq \Adv{\AB,\abG}{\Core}.$
\end{lemma}

\begin{lemma}$(\game{5}$ to $\game{6})$\label{lem:G{5}-G{6}}
    For any adversary $\AA$,
    $\adv{\AA}{\game{5}}=\adv{\AA}{\game{6}}.$
\end{lemma}

%!TEX spellcheck = en_US
%!TEX root = ../main.tex

\subsection{Proof of Lemmas for Full Security of BAEKS}\label{sec:Full-Lemma-Proofs}
\begin{proof}[Proof of \Cref{lem:G{0}-G{1}}]
    This is a syntactic change and therefore $\adv{\AA}{\game{0}}=\adv{\AA}{\game{1}}$.
\end{proof}

\begin{proof}[Proof of \Cref{lem:G{1}-G{2}}]
    To argue this indistinguishability, we introduce an intermediate game $\game{1.5}$ which is exactly same as $\game{1}$ except that $\ct[{\ij,0}]$ are sampled from $\Zp^{\nRA}$ for all $\ij\in[\ell]$ for all challenge queries. We then argue that, $\game{1}$ and $\game{1.5}$ are computationally indistinguishable under $|\Qct|\ell\mhyf$fold $\lmatDH[\nRA,\nCA]$ assumption. The same argument can also be applied to show $\game{1.5}$ and $\game{2}$ are computationally indistinguishable under $|\Qtrp|\mhyf$fold $\lmatDH[\nRB,\nCB]$ assumption.

    We next construct the reduction $\AB$ that takes $|\Qct|\ell\mhyf$fold $\lmatDH[\nRA,\nCA]$ instance $(\On{\M},\Tw{\M},\On{\f[1]},\ldots,$ $\On{\f[|\Qct|\ell]})$. It samples $\B\sample \Zp^{\nRB\times\nCB}$ and simulates $\pp$ using $(\On{\M},\Tw{\M},\On{\B},\Tw{\B})$. It uses $\f[(c-1)\ell+1]$ to $\f[c\ell]$ to simulate $c^{th}$ challenge $\ct[{1,1}]$ to $\ct[{\ell,1}]$. Rest of the simulation is done exactly the same as in $\game{1}$. Thus, $|\adv{\AA}{\game{1}}-\adv{\AA}{\game{1.5}}|\leq \Adv{\AB,\Go}{|\Qct|\ell\mhyf\lmatDH[\nRA,\nCA]}+\frac{1}{p-1}$.

    Similarly, we could argue that, $|\adv{\AA}{\game{1.5}}-\adv{\AA}{\game{2}}|\leq \Adv{\AB,\Gt}{|\Qtrp|\mhyf\lmatDH[\nRB,\nCB]}+\frac{1}{p-1}$.
\end{proof}

\begin{proof}[Proof of \Cref{lem:G{2}-G{3}}]
    $\game{3}$ does not use $\Q[\ii,b]$, $\W[j]$ and $\Y[j]$ anymore. 
    In particular, by $\Kt[j,\ii,b]\A=\Qt[j,\ii,b]\Au+\Nt[j,\ii,b]\Ad$, $\Ut[j]\A=\Wt[j]\Au+\Xt[j]\Ad$ and $\Vt[j]\A=\Yt[j]\Au+\Zt[j]\Ad$, we get
$\Qt[j,\ii,b]=(\Kt[j,\ii,b]\A-\Nt[j,\ii,b]\Ad)\Au^{-1}$, $\Wt[j]=(\Ut[j]\A-\Xt[j]\Ad)\Au^{-1}$ and $\Yt[j]=(\Vt[j]\A-\Zt[j]\Ad)\Au^{-1}$. We substitute $\Q[\ii,b],\W[j]$ and $\Y[j]$ in all $\Otrp$ responses:
{\scriptsize    
    \begin{equation*}
        \begin{aligned}
            \trpu_1^{\top}  &= \trp[{0}]^{\top}(\Sum{\ii}{[\klen]}\Qt[\ii,{\msgii[\kwd]{\ii}}]+\Wt[\Sndrr]+\Yt[\Rcvrr])\\
                            &= \trp[{0}]^{\top}\Sum{\ii}{[\klen]}(\Kt[\ii,{\msgii[\kwd]{\ii}}]\A-\Nt[\ii,{\msgii[\kwd]{\ii}}]\Ad)\Au^{-1}+\trp[{0}]^{\top}(\Ut[\Sndrr]\A-\Xt[\Sndrr]\Ad)\Au^{-1}+\trp[{0}]^{\top}\Vt[\Rcvrr]\A\Au^{-1}\\
                            &\qquad -\trp[{0}]^{\top}\Zt[\Rcvrr]\Ad\Au^{-1}\\
                            &= \trp[{0}]^{\top}(\Sum{\ii}{[\klen]}\Kt[\ii,{\msgii[\kwd]{\ii}}]\A+\Ut[\Sndrr]\A+\Vt[\Rcvrr]\A)
                            -\trp[{0}]^{\top}(\Sum{\ii}{[\klen]}\Nt[\ii,{\msgii[\kwd]{\ii}}]\Ad+\Xt[\Sndrr]\Ad+\Zt[\Rcvrr]\Ad)\Au^{-1}\\
                            &= \trp[{0}]^{\top}(\Sum{\ii}{[\klen]}\Kt[\ii,{\msgii[\kwd]{\ii}}]+\Ut[\Sndrr]+\Vt[\Rcvrr])\A
                            -\trp[{0}]^{\top}(\Sum{\ii}{[\klen]}\Nt[\ii,{\msgii[\kwd]{\ii}}]+\Xt[\Sndrr]+\Zt[\Rcvrr])\Ad\Au^{-1}\\
                        &= \trp[{0}]^{\top}(\Sum{\ii}{[\klen]}\Kt[\ii,{\msgii[\kwd]{\ii}}]+\Ut[\Sndrr]+\Vt[\Rcvrr])\A
                            -\trp[{0}]^{\top}(\Sum{\ii}{[\klen]}\Nt[\ii,{\msgii[\kwd]{\ii}}]+\Xt[\Sndrr]+\Zt[\Rcvrr])\Ad\Au^{-1}\\
                        &= \trp[{0}]^{\top}(\Sum{\ii}{[\klen]}\Kt[\ii,{\msgii[\kwd]{\ii}}]+\Ut[\Sndrr]+\Vt[\Rcvrr])\A-\trpd_1^{\top}\Ad\Au^{-1}\\        
        \end{aligned}
    \end{equation*}
}    As for the distribution of $\Oct$ responses, it is easy to see that $\ct[{\ij,0}]$ are uniformly random, as in $\game{2}$. 
    We substitute $\Q[\ii,b],\W[j]$ and $\Y[j]$ in all $\Oct$ responses:
{\scriptsize   
    \begin{align*}\hspace{-1.1cm} 
        % \begin{aligned}
            \ct[{\ij,0}]
            &= \iCol{\ctu_{\ij,0}}{\ctd_{\ij,0}}\text{ for }\ctd_{\ij,0}=\a[\ij,0]+\Ad\Au^{-1}\ctu_{\ij,0}\\
            \ct[{\ij,1}]
            &= (\Sum{\ii}{[\klen]}\Kt[\ii,{\msgii[\kw]{\ii}}]+\Ut[\Sndr]+\Vt[\Rcvr_{\perm(\ij)}]){\A}\Au^{-1}\ctu_{\ij,0}
            +(\Sum{\ii}{[\klen]}\Nt[\ii,{\msgii[\kw]{\ii}}]+\Xt[\Sndr]+\Zt[\Rcvr_{\perm(\ij)}])\a[\ij,0]\\
            &= (\Sum{\ii}{[\klen]}(\Qt[\ii,{\msgii[\kw]{\ii}}]\Au+\Nt[\ii,{\msgii[\kw]{\ii}}]\Ad)+\Wt[\Sndr]\Au+\Xt[\Sndr]\Ad+\Vt[\Rcvr_{\perm(\ij)}]\Au+\Zt[\Rcvr_{\perm(\ij)}]\Ad)\Au^{-1}\ctu_{\ij,0}\\
            &\quad+(\Sum{\ii}{[\klen]}\Nt[\ii,{\msgii[\kw]{\ii}}]+\Xt[\Sndr]+\Zt[\Rcvr_{\perm(\ij)}])(\ctd_{\ij,0}-\Ad\Au^{-1}\ctu_{\ij,0}) \\
            &= \!(\!\Sum{\ii}{[\klen]}\!\Qt[\ii,{\msgii[\kw]{\ii}}]\!+\!\Wt[\Sndr]\!+\!\Vt[\Rcvr_{\perm(\ij)}]\!)\!\Au\Au^{-1}\ctu_{\ij,0}
            +(\Sum{\ii}{[\klen]}\Nt[\ii,{\msgii[\kw]{\ii}}]+\Xt[\Sndr]\!+\!\Zt[\Rcvr_{\perm(\ij)}]\!)\!\Ad\Au^{-1}\ctu_{\ij,0}\\
            &\qquad+(\Sum{\ii}{[\klen]}\Nt[\ii,{\msgii[\kw]{\ii}}]+\Xt[\Sndr]+\Zt[\Rcvr_{\perm(\ij)}])\ctd_{\ij,0}-(\Sum{\ii}{[\klen]}\Nt[\ii,{\msgii[\kw]{\ii}}]+\Xt[\Sndr]+\Zt[\Rcvr_{\perm(\ij)}])\Ad\Au^{-1}\ctu_{\ij,0}\\
            &= (\Sum{\ii}{[\klen]}\Qt[\ii,{\msgii[\kw]{\ii}}]+\Wt[\Sndr]+\Vt[\Rcvr_{\perm(\ij)}])\ctu_{\ij,0}
                        +(\Sum{\ii}{[\klen]}\Nt[\ii,{\msgii[\kw]{\ii}}]+\Xt[\Sndr]+\Zt[\Rcvr_{\perm(\ij)}])\ctd_{\ij,0}\\
                        &= {(\Sum{\ii}{[\klen]} \iRow{\Qt[\ii,{\msgii[\kw]{\ii}}]}{\Nt[\ii,{\msgii[\kw]{\ii}}]}
                            +\iRow{\Wt[\Sndr]}{\Xt[\Sndr]}
                            +\iRow{\Yt[\Rcvr_{\perm(\ij)}]}{\Zt[\Rcvr_{\perm(\ij)}]})\iCol{\ctu_{\ij,0}}{\ctd_{\ij,0}}} \displaybreak[1]\\
                        &= {(\Sum{\ii}{[\klen]}\iCol{\Q[\ii,\kw_{\ii}]}{\N[\ii,\kw_{\ii}]}^{\top}+\iCol{\W[\Sndr]}{\X[\Sndr]}^{\top}+\iCol{\Y[\Rcvr_{\perm(\ij)}]}{\Z[\Rcvr_{\perm(\ij)}]}^{\top})\ct[{\ij,0}]}\\
                        &= {(\Sum{\ii}{[\klen]}\Kt[\ii,{\msgii[\kw]{\ii}}] +\Ut[\Sndr]+\Vt[\Rcvr_{\perm(\ij)}])\ct[{\ij,0}]}\\
        % \end{aligned}
    \end{align*}
}
    This is a syntactic change provided $\Au$ is invertible and therefore $|\adv{\AA}{\game{2}}-\adv{\AA}{\game{3}}|\leq \frac{1}{p-1}$.
\end{proof}

\begin{proof}[Proof of \Cref{lem:G{3}-G{4}}]
    We sample $\Kt[\ii,b]\A,\Ut[j]\A,\Vt[j]\A\sample \Zp^{\nRB\times \nCA}$. Then we compute $\K[\ii,b]\B,\U[j]\B$ and $\V[j]\B$ following $\game{3}$. Since $\game{3}$ sampled $\K[\ii,b],\J[j,i,{\msgii[j]{i}}],\W[j],\Y[j]$ uniformly at random, $\Kt[\ii,b]\A,$ $\Ut[j]\A,\Vt[j]\A$ were also uniformly random quantities.
    This is therefore a conceptual change provided $\Au$ is invertible and therefore $|\adv{\AA}{\game{3}}=\adv{\AA}{\game{4}}|\leq \frac{1}{p-1}$.
\end{proof}

\begin{proof}[Proof of \Cref{lem:G{4}-G{5}}]
    This lemma is proved assuming $\Exp{\AB}{\Core}{\bee}$ is tightly secure for all $\ppt$ adversary which we have proved in the next section. Given a $\ppt$ adversary $\AA$ that distinguishes $\game{4}$ and $\game{5}$, we construct a reduction $\AB$ that distinguishes $\Exp{\AB}{\Core}{0}$ from $\Exp{\AB}{\Core}{1}$. Let $\AC$ be the $\Exp{\AB}{\Core}{\bee}$ challenger who sampled $\bee\sample\{0,1\}$ uniformly at random at the start of the game. 

    \begin{itemize}
        \item $\AC$ gives out $\ppC$ to $\AB$.
        \begin{itemize}
            \item $\AB$ samples $\dee\sample\{0,1\}$.
            \item $\AB$ samples $\A\sample \Dk[\nRA,\nCA]$.
            \item $\AB$ samples {$\Kt[\ii,b]\A\sample \Zp^{\nRB\times \nCA}$} for all $(\ii,b)\in[\klen]\times\{0,1\}$.
            \item $\AB$ defines $\pp=(\On{\A},\Tw{\A},\On{\Kt[\ii,b]\A},\ppC)$.
            \item $\AB$ gives $\pp$ to $\AA$.
        \end{itemize}
        \item When $\AA$ makes $\Osk$ queries on a new $j$:
        \begin{itemize}
            \item $\AB$ first makes a $\OKey$ query on $j$.
            \item $\AB$ then makes a $\OCor$ query on $j$.
            \item $\AC$ responds with $\sky_j=\left(\Tw{\X[j]\B},\Tw{\Z[j]\B}\right)$.
            \item $\AB$ samples {$\Ut[j]\A,\Vt[j]\A\sample \Zp^{\nRB\times \nCA}$}.
            \item $\AB$ computes {\scriptsize$\U[j]\B=\iCol{\Au^{-\top}((\Ut[j]\A)^\top\B-\Adt\X[j]\B)}{\X[j]\B}$} and {\scriptsize$\V[j]\B=\iCol{\Au^{-\top}((\Vt[j]\A)^{\top}\B-\Adt\Z[j]\B)}{\Z[j]\B}.$}
            \item $\AB$ maintains a list $\Qsk$ of $(j,\sk_j,\pk_j)$.  
            \item It returns {\scriptsize$(\sk_j,\pk_j)$} to $\AA$ where 
            {\scriptsize $\sk_j=\left(\On{\Ut[j]\A},\Tw{\V[j]\B}\right)$} and 
            {\scriptsize $\pk_j=\left(\Tw{\U[j]\B},\On{\Vt[j]\A}\right).$}
        \end{itemize}
        \item {When $\AA$ makes $\Opk$ queries on a new $j$:}
        \begin{itemize}
            \item $\AB$ first makes a $\OKey$ query on $j$.
            % \item {\color{\confcolor}$\AB$ then makes a $\OCor$ query on $j$.}
            \item {\color{\confcolor}$\AC$ responds with $(j,{\color{\confcolor}\Tw{\X[j]\B}})$.}
            \item $\AB$ samples {\scriptsize$\Ut[j]\A,\Vt[j]\A\sample \Zp^{\nRB\times \nCA}$}. 
            % and {\scriptsize\color{\confcolor}$\hX[j,\ii,b],\hZ[j,\ii,b]\sample \Zp^{(\nRAd)\times (\nRB)}$ for all $(\ii,b)\in[L]\times\{0,1\}$.}
            \item $\AB$ computes {\scriptsize$\U[j]\B=\iCol{\Au^{-\top}((\Ut[j]\A)^\top\B-\Adt\X[j]\B)}{\X[j]\B}$}. 
            %and {\scriptsize$\V[j,\ii,b]\B=\iCol{\Au^{-\top}((\Vt[j,\ii,b]\A)^{\top}\B-\Adt\hZ[j,\ii,b]\B)}{\hZ[j,\ii,b]\B}.$}
            % \item $\AB$ maintains a list $\Qpk$ of $(j,(\hsk_j,${\scriptsize$(\hX[j,\ii,b],\hZ[j,\ii,b])_{\substack{{\ii\in[L]}\\{b\in\{0,1\}}}}$}$),\hpk_j)$ where 
            % {\scriptsize $\hsk_j=(\left(\On{\Ut[j,\ii,b]\A},\Tw{\V[j,\ii,b]\B}\right)_{\substack{{\ii\in[L]}\\{b\in\{0,1\}}}})$} and 
            \item $\AB$ sends
            {\scriptsize $\pk_j=\left(\Tw{\U[j]\B},\On{\Vt[j]\A}\right)$} back to $\AA$.
            % \item It returns {$\hpk_j$} to $\AA$.
        \end{itemize}

        \item When $\AA$ makes $\Otrp$ queries on new $((\pk_{\Sndrr^{\xpz}},\kwd^{\xpz},\pk_{\Rcvrr^{\xpz}}),(\pk_{\Sndrr^{\xpo}},\kwd^{\xpo},\pk_{\Rcvrr^{\xpo}}))$:
        \begin{itemize}
            \item $\AB$ finds $\Sndrr^{\xpd}$ and $\Rcvrr^{\xpd}$ from $\Qpk$ it maintained.
            \item $\AB$ makes a $\QTg$ query on $(\Sndrr^{\xpd},\kwd^{\xpd},\Rcvrr^{\xpd})$.
            \item $\AC$ replies with $\Tg=(\tg[0],\tg[1])\in \Gt^{\nRB}\times \Gt^{\nRAd}$.
            \item $\AB$ uses $\Tg$ to compute $\Trp=(\trp[0],\trp[1])$ where it sets $\trp[0]=\tg[0]$ and 
{\scriptsize
\[\trp[1]=\iCol{\Au^{-\top}{(\Sum{\ii}{[\klen]}\Kt[\ii,{\msgii[\kwd]{\ii}}]\A+\Ut[\Sndrr]\A+\Vt[\Rcvrr]\A)}^{\top}\tg[{0}]-\Au^{-\top}\Adt\tg[1]}{\tg[1]}\]
}
            \item $\AB$ forwards $\Trp$ to $\AA$ as the challenge. 
        \end{itemize}
        \item When $\AA$ makes $\Oct$ queries on new $((\pk_{\Sndr^{\xpz}},\kw^{\xpz},\pk_{\RSet^{\xpz}}),(\pk_{\Sndr^{\xpo}},\kw^{\xpo},\pk_{\RSet^{\xpo}}))$:
        \begin{itemize}
            \item $\AB$ finds $\Sndr^{\xpd}$ and $\RSet^{\xpd}=\{\Rcvr_1,\ldots,\Rcvr_{\ell}\}$ from $\Qpk$ it maintained.
            \item $\AB$ makes a $\QChl$ query on $(\Sndr^{\xpd},\kw^{\xpd},\RSet^{\xpd})$.
            \item $\AC$ replies with $\Chl=((\chl[\ij,0],\chl[\ij,1])_{\ij\in[\ell]})\in (\Go^{\nRAd}\times \Go^{\nRB})^{\ell}$
            \item $\AB$ uses $\Chl$ to compute $\Ct=((\ct[\ij,0],\ct[\ij,1])_{\ij\in[\ell]})$ where it sets {\scriptsize$\ct[\ij,0]=\iCol{\hct_{\ij,0}}{\chl[\ij,0]+\Ad\Au^{-1}\hct_{\ij,0}}$} and 
{\scriptsize
$\ct[{\ij,1}]=\On{(\Sum{\ii}{[\klen]}\Kt[\ii,{\msgii[\kwd]{\ii}}]\A+\Ut[\Sndrr]\A+\Vt[\Rcvrr]\A)\Au^{-1}\hct_{\ij,0}}+\chl[\ij,1]$
}
for {\scriptsize$\hct_{\ij,0}\sample\Zp^{\nRAu}$}
            \item $\AB$ forwards $\Ct$ to $\AA$ as the challenge. 
        \end{itemize}
        \item $\AA$ finally outputs $\dee'\in\{0,1\}$.
    \end{itemize}

{\color{\confcolor}    
Observe that when $\AA$ makes $\Osk$ query on any arbitrary $j$, $\AB$ queries $\OKey$ and $\OCor$ to get $\sky_j$. $\AB$ then uses $\sky_j=(\left(\Tw{\X[j,\ii,b]\B},\Tw{\Z[j,\ii,b]\B}\right)_{\substack{{\ii\in[L]}\\{b\in\{0,1\}}}})$, to compute $\pk_j$ and $\sk_j$. Thus, $\AA$ receives the keys of $j$.

Observe that when $\AA$ makes $\Opk$ query on any arbitrary $j$, $\AB$ queries $\OKey$ on $j$. Notice that, $\AB$ gets $(j,\Tw{\X[j]\B})$ in this case and cannot compute $\pk_j$ and $\sk_j$ completely. $\AB$ samples $\Ut[j]\A,\Vt[j]\A\sample \Zp^{\nRB\times \nCA}$ uniformly at random, and use it simulate $\pk_j$ response to $\AA$. 
% Observe that, irrespective of the value of $\bee$, $\AA$ receives the same $\pp$ and $(\sk_j,\pk_j)$ on it's query.
    If $\bee=0$, $\AA$ gets properly simulated challenges for both $\Trp$ and {$\Ct$}. $\AA$ can guess $\dee$ in this case with probability $1/2+|\adv{\AA}{\game{4}}-\adv{\AA}{\game{5}}|$.
    If $\bee=1$, $\AA$ gets simulated random challenges for both $\Trp$ and {$\Ct$}.
    $\AA$ can guess $\dee$ in this case with probability $1/2$.
    Thus, $|\adv{\AA}{\game{4}}-\adv{\AA}{\game{5}}|\leq \Adv{\AB,\abG}{\Core}$.    
{The restrictions in $\fullcpa$ security and the Core-Lemma \Cref{fig:Core-Lemma} are same and therefore translate naturally.}
}
\end{proof}

\begin{proof}[Proof of \Cref{lem:G{5}-G{6}}]
    Recall that, $\Qct\cap\Qtrp=\emptyset$.
    Since $\a[\Sndr,\kw,\Rcvr_{\perm(\ij)}]\sample \Zp^{\nRB}$ and $\d[\Sndrr,\kwd,\Rcvrr]\sample \Zp^{\nRAd}$ are sampled afresh, these make $\ct[{\ij,1}]$ and $\trp[{1}]$ uniformly random and independent.
\end{proof}

\section{Core Lemma} 
\label{sec:Core}
%!TEX spellcheck = en_US
%!TEX root = ../main.tex

We now present the core lemma in terms of an experiment in \Cref{fig:Core-Lemma} that forms the basis in the proofs of our BAEKS construction. 
Informally speaking, the core lemma experiment checks if $\AA$ has guessed the experiment's choice $\bee'$ right after making polynomial many queries to $\OKey$, $\OCor$, $\OTag$, $\OChal$ oracles in an interleaved manner subject to the condition that same queries are not made on $\OTag$ and $\OChal$ oracles and only users, that are not a part of any $\OTag$ or $\OChal$ queries, can be corrupted. 
Observe that, these restrictions match with the restrictions of $\fullcpa$ security and are only natural restrictions.
% {\color{red}Add Informal description.}
Interestingly, $\OKey$ outputs $\Tw{\X[j]\B}$, a part of the secret key. However, we note that the indistinguishability argument still goes through. Intuitively, the primary reason for that would be for every user $j$, $\X[j]=\SumiDT{1}{\ulen}\J[j,i,{\msgii[j]{i}}]$, a summation of random quantity. Thus, even after the projection $\Tw{\X[j]\B}$ is published, $\X[j]$ keeps sufficient entropy that can be utilized to inject randomness into $\OTag$ and $\OChal$ responses. This small observation has allowed us to prove the Core Lemma and subsequently use Core Lemma to prove \Cref{lem:G{4}-G{5}}. 
% It extends the proof of Affine-MAC from \cite{C:BlaKilPan14,AC:HofJiaPan18,PKC:LanPan19,PKC:LanPan20} except that 
For any adversary $\AA$, the advantage is defined as 
\[
    \Adv{\AA,\abG}{\Core}=|\Pr[\Exp{\AA}{\Core}{0}=1]-\Exp{\AA}{\Core}{1}=1|
\]
\begin{figure*}
    \scriptsize
    %\centering
    \hspace{-0.5cm}
    \fbox{
        % \parbox[c][73mm][c]{1.0\textwidth}{
        % \centering$\came{0}$, \lfbox{$\came{1}$}\\ \ \\
        %\hspace{-0.6cm}
        \begin{minipage}{0.52\textwidth}
            \underline{$\Exp{\AA}{\Core}{\bee}$}
            \begin{algorithmic}[1]
                \State $\abG\sample \ABSGen(1^\secpp)$
                \State $\B\sample\Dk[\nRB,\nCB]$
                \State {\color{\confcolorr}For $(\ii,b)\in[\klen]\times\{0,1\}: \N[\ii,b]\sample \Zp^{(\nRAd)\times (\nRB)}$}
                \State $\ppC=(\On{\B},\Tw{\B},{\color{\confcolorrr}\left\{\Tw{\N[\ii,b]\B}\right\}_{\!\substack{\!\ii\in [\klen] \\ \!b\in\{0,1\}}}})$ %\On{\N[\ii,b]\B},
                \State $\bee'\leftarrow\AA^{\OKey(\cdot),\OTag(\cdot,\cdot,\cdot),\OChal(\cdot,\cdot,\cdot)}(\ppC)$
                \State Return $(\bee\iseq \bee')\ \wedge\ (\QTg\cap\QChl=\emptyset)$
                \item[]\hspace{0.0cm}${\color{\confcolor}\wedge\ (\forall\ (u,\dontcare)\in\QCr, (u,\dontcare,\dontcare),(\dontcare,\dontcare,u)\notin {\color{\confcolor}\QTg\cup}\QChl)}$
                \item[]
            \end{algorithmic}
            \underline{$\OKey(j\in\USet)$}
            \begin{algorithmic}[1]
                \State For $(\ii,b)\in[2\ulen]\times\{0,1\}$: $\J[j,\ii,b]\sample \Zp^{(\nRAd)\times (\nRB)}$
                \State ${\color{\confcolorr}\X[j] = \SumiDT{1}{\ulen}\J[j,i,{\msgii[j]{i}}]}$, ${\color{\confcolorr}\Z[j] = \SumiDT{\ulen+1}{2\ulen}\J[j,i,{\msgii[j]{i}}]}$ %$\quad\bslash \hame{1}-\hame{4}$          
                \State $\sky_j=(\left(\Tw{\X[j]\B},\Tw{\Z[j]\B}\right))$
                \State $\QSk=\QSk\sqcup \{(j,\sky_j)\}$
                % \State $\SSk=\SSk\cup \{(j,\sky_j)\}$
                \State Return $(j,{\color{\confcolor}\Tw{\X[j]\B}})$
                \item[]
                % \item[]
            \end{algorithmic}
            % \underline{$\Finalize(\bee\in\{0,1\})$}
            % \begin{algorithmic}[1]
            %     \State Return $\bee \wedge (\QTg\cap\QChl=\emptyset)$
            %     \item[]\hspace{1cm}${\color{\confcolor}\wedge\ (\dforall\ u,v\in\QSk, (u,\dontcare,v)\notin {\color{\confcolor}\QTg\cup}\QChl)}$
            %     \item[]
            %     \item[]
            %     \item[]
            %     % \item[]
            %     % \item[]
            %     % \item[]
            % \end{algorithmic}  
        \end{minipage}
        \vline \hspace{1pt}
        \begin{minipage}{0.57\textwidth}
            \underline{$\OCor(j\in \USet)$}
            \begin{algorithmic}[1]
                \State $\QCr=\QCr\sqcup \{(j,\sky_j)\}$ from $\QSk$
                \State Return $\sky_j$
                \item[]
            \end{algorithmic}
            
            \underline{$\OTag({\Sndrr},\kwd,{\Rcvrr})$}   
            \begin{algorithmic}[1]
                \State $\QTg=\QTg\sqcup \{({\Sndrr},\kwd,{\Rcvrr})\}$
                \State $\Tg=(\Tw{\tg[{0}]},\Tw{\tg[{1}]})$ where
                \item[]$\tg[{0}]\sample\Zp^{\nRB}$, $\tg'\sample\Zp^{\nRAd}$ %={\B\r}$
                \item[]$\tg[{1}]=(1-\bee)(\SumiDT[\ii]{1}{\klen}\N[\ii,{\msgii[\kwd]{\ii}}]+\X[\Sndrr]+\Z[\Rcvrr])\tg[{0}]+\bee\cdot \tg'$ 
                \State Return $\Tg$ 
                \item[]
            \end{algorithmic}
            \underline{$\OChal({\Sndr}, \kw, {\RSet})$}
            \begin{algorithmic}[1]
                \State Let $\RSet=\{\Rcvr_1,\ldots,\Rcvr_{\ell}\}$
                \State $\QChl=\QChl\cup \{({\Sndr},\kw,{\Rcvr_{\ij}})_{\ij\in[\ell]}\}$ and $\perm\sample \Perm{[\ell]}$ 
                \State $\Chl=((\On{\chl[\ij,0]},\On{\chl[\ij,1]})_{\ij\in[\ell]})$ s.t.
                \item[]$\chl[\ij,0]\sample\Zp^{\nRAd}$, $\chl'\sample \Zp^{\nRB}$
                \item[]$\chl[\ij,1]=(1-\bee){
                (\SumiDT[\ii]{1}{\klen}\Nt[\ii,{\msgii[\kw]{\ii}}]+\Xt[\Sndr]+\Zt[\Rcvr_{\perm(\ij)}])}\chl[\ij,0]$
                %\lfbox[background-color=lightgray!60!white,border-color=lightgray!60!white]{
                    $+\bee\cdot \chl'$
                % \item[]\hspace{4in} where $\chl'\sample \Zp^{\nRB}$
                %} 
                %\item[]\lfbox{$\chl[\ij,1]^{\xpo}\sample \Zp^{\nRB}$}
                \State Return $\Chl$
                % \item[]
            \end{algorithmic}        
        \end{minipage}
    % }
    }
    \caption{Experiment for the Core Lemma.} %Here, $\Ffnc:\USet\times\KWd\times\USet\rightarrow \Zp^{\nRAd\times\nRB}$ is a random function. The difference between $\Exp{0,\AA}{\Core}$ and $\Exp{1,\AA}{\Core}$ is highlighted in gray.}
    \label{fig:Core-Lemma}
\end{figure*}

\begin{theorem}[Core Lemma]\label{thm:Core}
    Consider the experiments  $\Exp{\bee,\AA}{\Core}$ be defined as~\Cref{fig:Core-Lemma} where $\AA$ is given access to $\OKgen$, $\OCor$, $\OTag$ and $\OChal$ oracles for polynomial number of queries (i.e. $|\QSk|,|\QTg|,|\QChl|=\poly$).
    We claim that under the $\matDH[\nRB,\nCB]$ assumption for $\Go$ and $\matDH[\nRB,\nCB]$ assumption for $\Gt$, $\Adv{\AA,\abG}{\Core}=\neglgbl$ for all $\ppt$ adversary.
    For all $\ppt$ adversary $\AA$ having advantage $\Adv{\AA,\abG}{\Core}$,
    there exists $\ppt$ adversaries {$\AB,\AB_1,\AB_2,\AB_3$} s.t. 
{%\scriptsize    
    \begin{equation*}
        \begin{aligned}
            \Adv{\AA,\abG}{\Core} &\leq 2(2\ulen+\klen)\Adv{\AB_1,\Gt}{|\QTg|\mhyf\matDH[\nRB,\nCB]}\\
            &\qquad +2(2\ulen+\klen)\Adv{\AB_2,\Go}{\prll|\QChl|\mhyf\matDH[\nRA,\nCA]}\\
            &\qquad +(2\ulen+\klen)\Adv{\AB_3,\Go}{\prll|\QChl|\mhyf\matDH[\nRB,\nCB]}\\
            &\qquad +\frac{(2\ulen+\klen)(4\nCB+3|\QChl|+6)}{p-1}+\frac{(2\ulen+\klen)\nCB}{p} \\
            &\leq 4(2\ulen+\klen)\nCB\Adv{\AB,\Gt}{\matDH[\nCB]}\\
            &\qquad +4(2\ulen+\klen)\nCA\Adv{\AB,\Go}{\matDH[\nCA]}\\
            &\qquad +\frac{(2\ulen+\klen)(4\nCB+3|\QChl|+6)}{p-1}\\
            &\qquad +\frac{(2\ulen+\klen)\nCB}{p}\qquad (\text{due to }\Fact3)\\
            % & \qquad\qquad\qquad\qquad\qquad\qquad\qquad\qquad\qquad\qquad\qquad\qquad\qquad\qquad\qquad(\text{due to }\Fact3)\\
        \end{aligned}
    \end{equation*}
}
\end{theorem}

We use devise a hybrid argument to argue the indistinguishability of $\Exp{\AA}{\Core}{0}$ and $\Exp{\AA}{\Core}{1}$. We define $\hame{1},\hame{2,0,1},\ldots,\hame{2,{\klen-1},5},\hame{2,{\klen},1},\hame{3,{\klen},1},\ldots,\hame{3,{\klen+\ulen-1},5},$ $\hame{3,{\klen+\ulen},1},\hame{4,{\klen+\ulen},1},$ $\ldots,\hame{4,{\klen+2\ulen-1},5}$, $\hame{4,{\klen+2\ulen},5}$, $\hame{5}$ and $\hame{6}$ in \Cref{fig:Core-Hybrid}, \Cref{fig:CoreS-Hybrid} and \Cref{fig:CoreR-Hybrid} to argue the said indistinguishability. We describe the games next and provide the security argument in \Cref{sec:Core-Security} due to space limitation.
%!TEX spellcheck = en_US
%!TEX root = ../main.tex

\subsection{Overview of the Proof of Core-Lemma}\label{sec:Core-FullSec}
% \begin{proofsketch}[Proof of \Cref{thm:Core}]

% \end{proofsketch}

%background-color=lightgray!60!white,border-color=lightgray!60!white

%background-color=lightgray!60!white,border-color=black

\begin{proof}
    We use devise a hybrid argument to argue the indistinguishability of $\Exp{\AA}{\Core}{0}$ and $\Exp{\AA}{\Core}{1}$. We define $\hame{1},\hame{2,0,1},\ldots,\hame{2,{\klen-1},5},\hame{2,{\klen},1},\hame{3,{\klen},1},\ldots,$ $\hame{3,{\klen+\ulen-1},5},$ $\hame{3,{\klen+\ulen},1},\hame{4,{\klen+\ulen},1},\ldots,\hame{4,{\klen+2\ulen-1},5},,\hame{4,{\klen+2\ulen},5}$, $\hame{5}$ and $\hame{6}$ in \Cref{fig:Core-Hybrid}, \Cref{fig:CoreS-Hybrid} and \Cref{fig:CoreR-Hybrid} to argue the said indistinguishability. 
    %Note that, we defined $\hame{3,{\klen},1}=\hame{3,{\klen},1}$ and $\hame{4,{\klen+\ulen},1}=\hame{3,{\klen+\ulen},1}$. 
    In \Cref{sec:Core-Lemma-Lemmas}, we give the lemmas needed to prove these indistinguishabilities. 
    % Due to space constraint, we defer the proofs of these lemmas to \Cref{sec:Core-Lemma-Proofs}.
    \begin{itemize}
        \item $\hame{1}$: Same as $\Exp{\AA}{\Core}{0}$.
        \item $\hame{2,i,1}$: Injected randomness via $\RFdef{i}:\{0,1\}^{i}\rightarrow \Zp^{\nRAd\times \nRCB}$ defined on-the-fly. %,\ZRFdef[j]{i}
        \item $\hame{2,i,2}$: Sample each $\tg[{0}]$ based on concerned $\msgii[\kw]{i+1}\in \{0,1\}$.
        \item $\hame{2,i,3}$: Express $\RFdef{i}:\{0,1\}^{i}\rightarrow \Zp^{\nRAd\times \nRCB}$ in the span $(\Bzs||\Bos)$ in terms of $\ZFdef{i},\OFdef{i}:\{0,1\}^{i}\rightarrow \Zp^{\nRAd\times \nCB}$ defined on-the-fly. 
        % Similarly, we express $\ZRFdef[j]{i}:\{0,1\}^{i}\rightarrow \Zp^{\nRAd\times \nRCB}$ in the span $(\Bzs||\Bos)$ in terms of $\ZZFdef[j]{i},\ZOFdef[j]{i}:\{0,1\}^{i}\rightarrow \Zp^{\nRAd\times \nCB}$.
        \item $\hame{2,i,4}$: Inject randomness in $\ZFdef{i}:\{0,1\}^{i}\rightarrow \Zp^{\nRAd\times \nCB}$ to define $\ZFdef{i+1}:\{0,1\}^{i+1}\rightarrow \Zp^{\nRAd\times \nCB}$  based on concerned $\msgii[\kw]{i+1}\in \{0,1\}$. %,\XOFdef[j]{i}, \XOFdef[j]{i+1}
        \item $\hame{2,i,5}$: Inject randomness in $\OFdef{i}:\{0,1\}^{i}\rightarrow \Zp^{\nRAd\times \nCB}$ to define $\OFdef{i+1}:\{0,1\}^{i+1}\rightarrow \Zp^{\nRAd\times \nCB}$  based on concerned $\msgii[\kw]{i+1}\in \{0,1\}$. % \ZZFdef[j]{i}, $\ZZFdef[j]{i+1},$ 
        \item $\hame{3,\klen+i,1}$: Injected randomness via $\SRFdef{\klen+i}:\{0,1\}^{\klen+i}\rightarrow \Zp^{\nRAd\times \nRCB}$ defined on-the-fly.
        \item $\hame{3,\klen+i,2}$: Sample each $\tg[{0}]$ based on concerned $\msgii[\Sndr]{i+1}\in \{0,1\}$.
        \item $\hame{3,\klen+i,3}$: Express $\SRFdef{\klen+i}:\{0,1\}^{\klen+i}\rightarrow \Zp^{\nRAd\times \nRCB}$ in the span $(\Bzs||\Bos)$ in terms of $\SZFdef{\klen+i},\SOFdef{\klen+i}:\{0,1\}^{\klen+i}\rightarrow \Zp^{\nRAd\times \nCB}$ defined on-the-fly. 
        \item $\hame{3,\klen+i,4}$: Inject randomness in $\SZFdef{\klen+i}:\{0,1\}^{\klen+i}\rightarrow \Zp^{\nRAd\times \nCB}$ to define $\SZFdef{\klen+i+1}:\{0,1\}^{\klen+i+1}\rightarrow \Zp^{\nRAd\times \nCB}$ depending on $\msgii[\Sndr]{i+1}\in \{0,1\}$. %,\XOFdef[j]{i}, \XOFdef[j]{i+1}
        \item $\hame{3,\klen+i,5}$: Inject randomness in $\SOFdef{\klen+i}:\{0,1\}^{\klen+i}\rightarrow \Zp^{\nRAd\times \nCB}$ to define $\SOFdef{\klen+i+1}:\{0,1\}^{\klen+i+1}\rightarrow \Zp^{\nRAd\times \nCB}$ depending on $\msgii[\Sndr]{i+1}\in \{0,1\}$.
        \item $\hame{4,\klen+\ulen+i,1}$: Injected randomness via $\RRFdef{\klen+\ulen+i}:\{0,1\}^{\klen+\ulen+i}\rightarrow \Zp^{\nRAd\times \nRCB}$ defined on-the-fly.
        \item $\hame{4,\klen+\ulen+i,2}$: Sample each $\tg[{0}]$ based on concerned $\msgii[\Rcvr]{i+1}\in \{0,1\}$.
        \item $\hame{4,\klen+\ulen+i,3}$: Express $\RRFdef{\klen+\ulen+i}:\{0,1\}^{\klen+\ulen+i}\rightarrow \Zp^{\nRAd\times \nRCB}$ in the span $(\Bzs||\Bos)$ in terms of $\RZFdef{\klen+\ulen+i},\ROFdef{\klen+\ulen+i}:\{0,1\}^{i}\rightarrow \Zp^{\nRAd\times \nCB}$ defined on-the-fly. 
        \item $\hame{4,\klen+\ulen+i,4}$: Inject randomness in $\RZFdef{\klen+\ulen+i}:\{0,1\}^{\klen+\ulen+i}\rightarrow \Zp^{\nRAd\times \nCB}$ to define $\RZFdef{\klen+\ulen+i+1}:\{0,1\}^{\klen+\ulen+i+1}\rightarrow \Zp^{\nRAd\times \nCB}$ depending on $\msgii[\Rcvr]{i+1}\in \{0,1\}$. %,\XOFdef[j]{i}, \XOFdef[j]{i+1}
        \item $\hame{4,\klen+\ulen+i,5}$: Inject randomness in $\ROFdef{\klen+\ulen+i}:\{0,1\}^{\klen+\ulen+i}\rightarrow \Zp^{\nRAd\times \nCB}$ to define $\ROFdef{\klen+\ulen+i+1}:\{0,1\}^{\klen+\ulen+i+1}\rightarrow \Zp^{\nRAd\times \nCB}$ depending on $\msgii[\Rcvr]{i+1}\in \{0,1\}$.

        \item $\hame{5}$: All $\tg[{1}]$ are sampled uniformly at random. %All $L$-bits of keywords are consumed to define a random function $\mathsf{RF}:\USet\times\KWd\times\USet\rightarrow \Zp^{\nRAd\times \nRCB}$. %$\XRFdef[j]{L},\ZRFdef[j]{L}:\{0,1\}^{L}\rightarrow \Zp^{\nRAd\times \nRCB}$.
        \item $\hame{6}$: Same as $\Exp{\AA}{\Core}{1}$.
        % {\color{\confcolor}For all challenges and tags, sample $\{\chl[j,1]\}_{j\in [L]},\tg[{1}]$ uniformly at random. }
    \end{itemize}
\end{proof}

\subsection{Lemmas for Core-Lemma}\label{sec:Core-Lemma-Lemmas}
% \begin{lemma}$(\hame{0}$ to $\hame{1})$\label{lem:H{0}-H{1}} 
%     For any adversary $\AA$, $\adv{\AA}{\hame{0}}=\adv{\AA}{\hame{1}}.$
% \end{lemma}

\begin{lemma}$(\hame{1}$ to $\hame{2,0,1})$\label{lem:H{1}-H{2,0,1}} 
    For any adversary $\AA$, $\adv{\AA}{\hame{1}}=\adv{\AA}{\hame{2,0,1}}.$
\end{lemma}

\begin{lemma}$(\hame{2,i,1}$ to $\hame{2,i,2})$\label{lem:H{2,i,1}-H{2,i,2}}
    For any $\ppt$ adversary $\AA$ in the $\fullcpa$ security model, there exists a $\ppt$ adversary $\AB$ such that
{  $|\adv{\AA}{\hame{2,i,1}}-\adv{\AA}{\hame{2,i,2}}|\leq {2}\Adv{\AB,\Gt}{|\QTg|\mhyf\matDH[\nRB,\nCB]}+{\frac{2\times2\nCB}{p-1}}.$}
\end{lemma}

\begin{lemma}$(\hame{2,i,2}$ to $\hame{2,i,3})$\label{lem:H{2,i,2}-H{2,i,3}}
    For any adversary $\AA$,
    $|\adv{\AA}{\hame{2,i,2}}-\adv{\AA}{\hame{2,i,3}}|\leq \frac{\nCB}{p}.$
\end{lemma}

\begin{lemma}$(\hame{2,i,3}$ to $\hame{2,i,4})$\label{lem:H{2,i,3}-H{2,i,4}}
    For any $\ppt$ adversary $\AA$ in the $\fullcpa$ security model, there exists a $\ppt$ adversary $\AB$ such that
    {$|\adv{\AA}{\hame{2,i,3}}-\adv{\AA}{\hame{2,i,4}}|\leq \Adv{\AB,\Go}{\prll|\QChl|\mhyf\matDH[\nRA,\nCA]}+ {\frac{|\QChl|+2}{p-1}}.$}
\end{lemma}

\begin{lemma}$(\hame{2,i,4}$ to $\hame{2,i,5})$\label{lem:H{2,i,4}-H{2,i,5}}
    For any $\ppt$ adversary $\AA$ in the $\fullcpa$ security model, there exists a $\ppt$ adversary $\AB$ such that
{$|\adv{\AA}{\hame{2,i,4}}-\adv{\AA}{\hame{2,i,5}}|\leq \Adv{\AB,\Go}{\prll|\QChl|\mhyf\matDH[\nRA,\nCA]}+ {\frac{|\QChl|+2}{p-1}}.$}
\end{lemma}

\begin{lemma}$(\hame{2,\klen,1}$ to $\hame{3,\klen,1})$\label{lem:H{2,k,1}-H{3,k,1}}
    For any adversary $\AA$,
    $\adv{\AA}{\hame{2,\klen,1}}=\adv{\AA}{\hame{3,\klen,1}}.$
\end{lemma}

\begin{lemma}$(\hame{3,\klen+i,1}$ to $\hame{3,\klen+i,2})$\label{lem:H{3,i,1}-H{3,i,2}}
    For any $\ppt$ adversary $\AA$ in the $\fullcpa$ security model, there exists a $\ppt$ adversary $\AB$ such that
{  $|\adv{\AA}{\hame{3,\klen+i,1}}-\adv{\AA}{\hame{3,\klen+i,2}}|\leq {2}\Adv{\AB,\Gt}{|\QTg|\mhyf\matDH[\nRB,\nCB]}+{\frac{2\times2\nCB}{p-1}}.$}
\end{lemma}

\begin{lemma}$(\hame{3,\klen+i,2}$ to $\hame{3,\klen+i,3})$\label{lem:H{3,i,2}-H{3,i,3}}
    For any adversary $\AA$,
    $|\adv{\AA}{\hame{3,\klen+i,2}}-\adv{\AA}{\hame{3,\klen+i,3}}|\leq \frac{\nCB}{p}.$
\end{lemma}

\begin{lemma}$(\hame{3,\klen+i,3}$ to $\hame{3,\klen+i,4})$\label{lem:H{3,i,3}-H{3,i,4}}
    For any $\ppt$ adversary $\AA$ in the $\fullcpa$ security model, there exists a $\ppt$ adversary $\AB$ such that
    {$|\adv{\AA}{\hame{3,\klen+i,3}}-\adv{\AA}{\hame{3,\klen+i,4}}|\leq \Adv{\AB,\Go}{\prll|\QChl|\mhyf\matDH[\nRA,\nCA]}+ {\frac{|\QChl|+2}{p-1}}.$}
\end{lemma}

\begin{lemma}$(\hame{3,\klen+i,4}$ to $\hame{3,\klen+i,5})$\label{lem:H{3,i,4}-H{3,i,5}}
    For any $\ppt$ adversary $\AA$ in the $\fullcpa$ security model, there exists a $\ppt$ adversary $\AB$ such that
{$|\adv{\AA}{\hame{3,\klen+i,4}}-\adv{\AA}{\hame{3,\klen+i,5}}|\leq \Adv{\AB,\Go}{\prll|\QChl|\mhyf\matDH[\nRA,\nCA]}+ {\frac{|\QChl|+2}{p-1}}.$}
\end{lemma}

\begin{lemma}$(\hame{3,\klen+\ulen,1}$ to $\hame{4,\klen+\ulen,1})$\label{lem:H{3,ku,1}-H{4,ku,1}}
    For any adversary $\AA$,
    $\adv{\AA}{\hame{3,\klen+\ulen,1}}=\adv{\AA}{\hame{4,\klen+\ulen,1}}.$
\end{lemma}

\begin{lemma}$(\hame{4,\klen+\ulen+i,1}$ to $\hame{4,\klen+\ulen+i,2})$\label{lem:H{4,i,1}-H{4,i,2}}
    For any $\ppt$ adversary $\AA$ in the $\fullcpa$ security model, there exists a $\ppt$ adversary $\AB$ such that
{  $|\adv{\AA}{\hame{4,\klen+\ulen+i,1}}-\adv{\AA}{\hame{4,\klen+\ulen+i,2}}|\leq {2}\Adv{\AB,\Gt}{|\QTg|\mhyf\matDH[\nRB,\nCB]}+{\frac{2\times2\nCB}{p-1}}.$}
\end{lemma}

\begin{lemma}$(\hame{4,\klen+\ulen+i,2}$ to $\hame{4,\klen+\ulen+i,3})$\label{lem:H{4,i,2}-H{4,i,3}}
    For any adversary $\AA$,
    $|\adv{\AA}{\hame{4,\klen+\ulen+i,2}}-\adv{\AA}{\hame{4,\klen+\ulen+i,3}}|\leq \frac{\nCB}{p}.$
\end{lemma}

\begin{lemma}$(\hame{4,\klen+\ulen+i,3}$ to $\hame{4,\klen+\ulen+i,4})$\label{lem:H{4,i,3}-H{4,i,4}}
    For any $\ppt$ adversary $\AA$ in the $\fullcpa$ security model, there exists a $\ppt$ adversary $\AB$ such that
    {$|\adv{\AA}{\hame{4,\klen+\ulen+i,3}}-\adv{\AA}{\hame{4,\klen+\ulen+i,4}}|\leq \Adv{\AB,\Go}{\prll|\QChl|\mhyf\matDH[\nRA,\nCA]}+ {\frac{|\QChl|+2}{p-1}}.$}
\end{lemma}

\begin{lemma}$(\hame{4,\klen+\ulen+i,4}$ to $\hame{4,\klen+\ulen+i,5})$\label{lem:H{4,i,4}-H{4,i,5}}
    For any $\ppt$ adversary $\AA$ in the $\fullcpa$ security model, there exists a $\ppt$ adversary $\AB$ such that
{$|\adv{\AA}{\hame{4,\klen+\ulen+i,4}}-\adv{\AA}{\hame{4,\klen+\ulen+i,5}}|\leq \Adv{\AB,\Go}{\prll|\QChl|\mhyf\matDH[\nRA,\nCA]}+ {\frac{|\QChl|+2}{p-1}}.$}
\end{lemma}

\begin{lemma}$(\hame{4,L,1}$ to $\hame{5})$\label{lem:H{4,L,1}-H{5}}
    For any adversary $\AA$, $\adv{\AA}{\hame{4,L,1}}=\adv{\AA}{\hame{5}}$
    %\[|\adv{\AA}{\hame{2,L-1,5}}-\adv{\AA}{\hame{3}}|{\color{red}\leq {2\times\nRCB}\Adv{\AB,\Gt}{\matDH[\nRB,\nCB]}+{\frac{2\times2\nCB}{p-1}}}.\]
\end{lemma}

\begin{lemma}$(\hame{5}$ to $\hame{6})$\label{lem:H{5}-H{6}}
    For any $\ppt$ adversary $\AA$ in the $\fullcpa$ security model, there exists a $\ppt$ adversary $\AB$ such that
{$|\adv{\AA}{\hame{5}}-\adv{\AA}{\hame{6}}|{\leq \Adv{\AB,\Go}{\prll|\QChl|\mhyf\matDH[\nRB,\nCB]}+\frac{|\QChl|+2}{p-1}}.$}
\end{lemma}

%!TEX spellcheck = en_US
%!TEX root = ../main.tex

\begin{figure*}[!h]
    \scriptsize
    %\centering
    \hspace{-1.0cm}
    \fbox{\parbox[c][117mm][c]{1.12\textwidth}{
        \centering\hspace{-1.0cm}{$\hame{1}$},\lfbox[rounded]{$\hame{2,i,1}$},\lfbox[dotted]{$\hame{2,i,2}$,\lfbox[border-style=dashed]{$\hame{2,i,3}$},\lfbox[background-color=lightgray!60!white,border-color=lightgray!60!white]{$\hame{2,i,4}$},\lfbox{$\hame{2,i,5}$}}\\ \ \\ 
        % , \lfbox[background-color=lightgray!60!white,border-color=black]{$\hame{3}$},  %
        % , \lfbox[border-style=double,border-width=2pt]{$\hame{4}$}\\ \ \\
        % \hspace{-2cm}
        \begin{minipage}{0.54\textwidth}
            \underline{$\Init(1^\secpp)$}
            \begin{algorithmic}[1]
                \State $\abG\sample \ABSGen(1^\secpp)$ %=(p,\Go,\Gt,\GT,\go,\gt,e)
                \State $\B$,\lfbox[dotted]{$\Bz,\Bo$}, \lfbox[border-style=dashed]{\lfbox{$\Bzs,\Bos$}} $\sample\Dk[\nRB,\nCB]$ s.t.
                \item[$-$] $(\B,\Bz,\Bo)$ generates $\Zp^{\nRB}$
                \item[$-$] $(\Bzs,\Bos)$ generates $\Bp$
                \item[$-$] $\Bt\Bzs=\Zro$, $\Bot\Bzs=\Zro$, $\Bt\Bos=\Zro$, $\Bzt\Bos=\Zro$ 
                \State {\color{\confcolorr}For $(\ii,b)\in[\klen]\times\{0,1\}: \N[\ii,b]\sample \Zp^{(\nRAd)\times (\nRB)}$} %\hfill $\bslash \hame{1}-\hame{4}$

                \State $\ppC=(\On{\B},\Tw{\B},{\color{\confcolorrr}\left\{\Tw{\N[\ii,b]\B}\right\}_{\!\substack{\!\ii\in [\klen] \\ \!b\in\{0,1\}}}})$
                \item[]
                % \item[]
                % \item[]
                % \item[]
            \end{algorithmic}
            \underline{$\OChal({\Sndr}, \kw, {\RSet})$}
            \begin{algorithmic}[1]
                \State Let $\RSet=\{\Rcvr_1,\ldots,\Rcvr_{\ell}\}$
                \State $\perm\sample \Perm{[\ell]}$ 
                \State $\QChl=\QChl\cup \{({\Sndr},\kw,{\Rcvr_i})_{i\in[\ell]}\}$
                \State $\Chl=((\On{\chl[\ij,0]},\On{\chl[\ij,1]})_{\ij\in[\ell]})$ s.t.
                % \State Return $\Chl=((\On{\chl[j,0]},\On{\chl[j,1]})_{j\in[\ell]})$
                \item[]$\chl[j,0]\sample\Zp^{\nRAd}$ 
                \item[]$\chl[j,1]={(\Sum{\ii}{[\klen]}\Nt[\ii,{\msgii[\kw]{\ii}}]+\Xt[\Sndr]+\Zt[\Rcvr_{\perm(j)}])\chl[j,0]}$
                \item[]\hspace{0.3cm}\lfbox[rounded]{$+\Bp(\RF[{\Sndr,\Rcvr_{\perm(j)}}]{i}{\msgi[\kw]{i}}^{\top})\chl[j,0]$}  % +\ZRF[\Rcvr_{\perm(j)}]{i}{\msgi[\kw]{i}}^{\top}
                \item[]\hspace{0.3cm}\lfbox[border-style=dashed]{$+(\Bzs\ZF[{\Sndr,\Rcvr_{\perm(j)}}]{i}{\msgi[\kw]{i}}^{\top}+\Bos\OF[{\Sndr,\Rcvr_{\perm(j)}}]{i}{\msgi[\kw]{i}}^{\top})\chl[j,0]$} %+\Bos(\XOF[\Sndr]{i}{\msgi[\kw]{i}}^{\top}+\ZOF[\Rcvr_{\perm(j)}]{i}{\msgi[\kw]{i}}^{\top})  
                \item[]\hspace{0.3cm}\lfbox[background-color=lightgray!60!white,border-color=lightgray!60!white]{$+(\Bzs\ZF[{\Sndr,\Rcvr_{\perm(j)}}]{i+1}{\msgi[\kw]{i+1}}^{\top}+\Bos\OF[{\Sndr,\Rcvr_{\perm(j)}}]{i}{\msgi[\kw]{i}}^{\top})\chl[j,0]$}  %+(\Bzs\ZZF[\Rcvr_{\perm(j)}]{i}{\msgi[\kw]{i}}^{\top}+\Bos\ZOF[\Rcvr_{\perm(j)}]{i}{\msgi[\kw]{i}}^{\top})
                \item[]\hspace{0.3cm}\lfbox{$+(\Bzs\ZF[{\Sndr,\Rcvr_{\perm(j)}}]{i+1}{\msgi[\kw]{i+1}}^{\top}+\Bos\OF[{\Sndr,\Rcvr_{\perm(j)}}]{i+1}{\msgi[\kw]{i+1}}^{\top})\chl[j,0]$}  % +(\Bzs\ZZF[\Rcvr_{\perm(j)}]{i+1}{\msgi[\kw]{i+1}}^{\top}+\Bos\ZOF[\Rcvr_{\perm(j)}]{i+1}{\msgi[\kw]{i+1}}^{\top})
                % \item[]\hspace{0.3cm}\lfbox[background-color=lightgray!60!white,border-color=black]{$+\Bp\RF{\Sndr}{\kw}{\Rcvr_{\perm(j)}}^{\top}\chl[j,0]$}  
%%%                \item[]\lfbox[border-style=double,border-width=2pt]{$\chl[j,1]\sample \Zp^{\nRB}$}
                \State Return $\Chl$ %=((\On{\chl[j,0]},\On{\chl[j,1]})_{j\in[\ell]})$
                \item[]
            \end{algorithmic}        
        \underline{$\OCor(j)$}
            \begin{algorithmic}[1]
                \State $\QCr=\QCr\sqcup \{(j,\sky_j)\}$ from $\QSk$
                \State Return $\sky_j$
                % \item[]
            \end{algorithmic}    
        \end{minipage}
        \vline \hspace{1pt}
        \begin{minipage}{0.55\textwidth}
            \underline{$\OKey(j)$}
            \begin{algorithmic}[1]
                \State For $(\ii,b)\in[2\ulen]\times\{0,1\}$: $\J[j,\ii,b]\sample\Zp^{(\nRAd)\times (\nRB)}$ %\hfill $\bslash \hame{0}$
                \State ${\color{\confcolorr}\X[j] =  \SumiDT{1}{\ulen}\J[j,i,{\msgii[j]{i}}],\Z[j] = \SumiDT{\ulen+1}{2\ulen}\J[j,i,{\msgii[j]{i}}]}$ %\hfill $\bslash \hame{1}-\hame{4}$   
                \State $\sky_j=(\left(\Tw{\X[j]\B},\Tw{\Z[j]\B}\right))$
                \State $\QSk=\QSk\sqcup \{(j,\sky_j)\}$
                \State Return $(j,{\color{\confcolor}\Tw{\X[j]\B}})$
                \item[]
                % \item[]
            \end{algorithmic}
            \underline{$\OTag({\Sndrr},\kwd,{\Rcvrr})$}   
            \begin{algorithmic}[1]
                \State $\QTg=\QTg\sqcup \{({\Sndrr},\kwd,{\Rcvrr})\}$
                %\State $\r\sample\Zp^{\nCB}$
                \State $\Tg=(\Tw{\tg[{0}]},\Tw{\tg[{1}]})$ where
                \item[]$\tg[{0}]\sample\Zp^{\nRB}$ %={\B\r}$
                \item[]\lfbox[dotted]{$\tg[{0}] = \B\r+\B_{\kwd_{i+1}}\r_{\kwd_{i+1}}$ for $\r,\r_{\kwd_{i+1}}\sample\Zp^{\nCB}$}   
                \item[]$\tg[{1}]={(\Sum{\ii}{[\klen]}\N[\ii,{\msgii[\kwd]{\ii}}]+\X[\Sndrr]+\Z[\Rcvrr])\tg[{0}]}$
                \item[]\hspace{0.3cm}\lfbox[rounded]{$+(\RF[{\Sndrr,\Rcvrr}]{i}{\msgi[\kwd]{i}})\Bpt\tg[{0}]$}  
                \item[]\hspace{0.3cm}\lfbox[border-style=dashed]{$+((\ZF[{\Sndrr,\Rcvrr}]{i}{\msgi[\kwd]{i}}\Bzst+\OF[{\Sndrr,\Rcvrr}]{i}{\msgi[\kwd]{i}}\Bost))\tg[{0}]$}   
                %{$+((\XZF[\Sndrr]{i}{\msgi[\kwd]{i}}\Bzst+\XOF[\Sndrr]{i}{\msgi[\kwd]{i}}\Bost)+(\ZZF[\Rcvrr]{i}{\msgi[\kwd]{i}}\Bzst+\ZOF[\Rcvrr]{i}{\msgi[\kwd]{i}})\Bost)\tg[{0}]$}   
                \item[]\hspace{0.3cm}\lfbox[background-color=lightgray!60!white,border-color=lightgray!60!white]{$+(\ZF[{\Sndrr,\Rcvrr}]{i+1}{\msgi[\kwd]{i+1}}\Bzst+\OF[{\Sndrr,\Rcvrr}]{i}{\msgi[\kwd]{i}}\Bost)\tg[{0}]$} 
                %{$+((\XZF[\Sndrr]{i+1}{\msgi[\kwd]{i+1}}\Bzst+\XOF[\Sndrr]{i+1}{\msgi[\kwd]{i+1}}\Bost)+(\ZZF[\Rcvrr]{i}{\msgi[\kwd]{i}}\Bzst+\ZOF[\Rcvrr]{i}{\msgi[\kwd]{i}})\Bost)\tg[{0}]$}
                \item[]\hspace{0.2cm}
                \lfbox{$+(\ZF[{\Sndrr,\Rcvrr}]{i+1}{\msgi[\kwd]{i+1}}\Bzst+\OF[{\Sndrr,\Rcvrr}]{i+1}{\msgi[\kwd]{i+1}}\Bost)\tg[{0}]$} 
                %\lfbox{$+((\XZF[\Sndrr]{i+1}{\msgi[\kwd]{i+1}}\Bzst+\XOF[\Sndrr]{i+1}{\msgi[\kwd]{i+1}}\Bost)+(\ZZF[\Rcvrr]{i+1}{\msgi[\kwd]{i+1}}\Bzst+\ZOF[\Rcvrr]{i+1}{\msgi[\kwd]{i+1}})\Bost)\tg[{0}]$}
                %\lfbox[background-color=lightgray!60!white,border-color=black]{$+\RF{\Sndrr}{\kwd}{\Rcvrr}\Bpt\tg[{0}]$}  
%%%                \item[]\hspace{0.0cm}\lfbox[background-color=lightgray!60!white,border-color=black]{$\tg[{1}]\sample \Zp^{\nRAd}$}
                % \item[]\lfbox[border-style=double,border-width=2pt]{$\tg[{1}]\sample \Zp^{\nRAd}$}
                \State Return $\Tg$ 
                \item[]
            \end{algorithmic}
            \underline{$\Finalize(\bee\in\{0,1\})$}
            \begin{algorithmic}[1]
                \State Return $\bee\wedge\ (\QTg\cap\QChl=\emptyset)$
                % \item[]\hspace{0.2cm}${\color{\confcolor}\wedge\ (\QSk\subset \SSk)}$
                \item[]\hspace{0.2cm}${\color{\confcolor}\wedge\ (\forall\ (u,\dontcare)\in\QCr, (u,\dontcare,\dontcare),(\dontcare,\dontcare,u)\notin {\color{\confcolor}\QTg\cup}\QChl)}$
                \item[]
                % \item[]
            \end{algorithmic}  
        \end{minipage}
    }
    }
    \caption{Hybrids for the transition from {\color{\confcolor}$\hame{1}$ to $\hame{2,\klen-1,5}$}.}
    \label{fig:Core-Hybrid}
\end{figure*}

\begin{figure*}[!h]
    \scriptsize
    %\centering
    \hspace{-1.17cm}
    \fbox{\parbox[c][117mm][c]{1.39\textwidth}{
        \centering\hspace{-1.0cm}{}\lfbox[rounded]{$\hame{3,\klen+i,1}$},\lfbox[dotted]{$\hame{3,\klen+i,2}$,\lfbox[border-style=dashed]{$\hame{3,\klen+i,3}$},\lfbox[background-color=lightgray!60!white,border-color=lightgray!60!white]{$\hame{3,\klen+i,4}$},\lfbox{$\hame{3,\klen+i,5}$}}\\ \ \\ 
        % ,\lfbox[background-color=lightgray!60!white,border-color=black]{$\hame{3}$} %$\hame{0},\hame{1}$
        % ,\lfbox[border-style=double,border-width=2pt]{$\hame{4}$}\\ \ \\
        \hspace{-0.4cm}
        \begin{minipage}{0.67\textwidth}
            \underline{$\Init(1^\secpp)$}
            \begin{algorithmic}[1]
                \State $\abG\sample \ABSGen(1^\secpp)$ %=(p,\Go,\Gt,\GT,\go,\gt,e)
                \State $\B$,\lfbox[dotted]{$\Bz,\Bo$}, \lfbox[border-style=dashed]{\lfbox{$\Bzs,\Bos$}} $\sample\Dk[\nRB,\nCB]$ s.t.
                \item[$-$] $(\B,\Bz,\Bo)$ generates $\Zp^{\nRB}$
                \item[$-$] $(\Bzs,\Bos)$ generates $\Bp$
                \item[$-$] $\Bt\Bzs=\Zro$, $\Bot\Bzs=\Zro$, $\Bt\Bos=\Zro$, $\Bzt\Bos=\Zro$ 
                \State {\color{\confcolorr}For $(\ii,b)\in[\klen]\times\{0,1\}: \N[\ii,b]\sample \Zp^{(\nRAd)\times (\nRB)}$} %$\qquad\bslash \hame{1}-\hame{4}$
                \State $\ppC=(\On{\B},\Tw{\B},{\color{\confcolorrr}\left\{\Tw{\N[\ii,b]\B}\right\}_{\!\substack{\!\ii\in [\klen] \\ \!b\in\{0,1\}}}})$
                \item[]
                % \item[]
                % \item[]
                % \item[]
            \end{algorithmic}
            \underline{$\OChal({\Sndr}, \kw, {\RSet})$}
            \begin{algorithmic}[1]
                \State Let $\RSet=\{\Rcvr_1,\ldots,\Rcvr_{\ell}\}$
                \State $\perm\sample \Perm{[\ell]}$ 
                \State $\QChl=\QChl\cup \{({\Sndr},\kw,{\Rcvr_{\ij}})_{\ij\in[\ell]}\}$
                \State $\Chl=((\On{\chl[\ij,0]},\On{\chl[\ij,1]})_{\ij\in[\ell]})$ s.t.
                % \State Return $\Chl=((\On{\chl[j,0]},\On{\chl[j,1]})_{j\in[\ell]})$
                \item[]$\chl[j,0]\sample\Zp^{\nRAd}$ 
                \item[]$\chl[j,1]={(\Sum{\ii}{[\klen]}\Nt[\ii,{\msgii[\kw]{\ii}}]+\Xt[\Sndr]+\Zt[\Rcvr_{\perm(j)}])\chl[j,0]}$
                \item[]\hspace{0.3cm}\lfbox[rounded]{$+\Bp(\SRF[{\Sndr,\Rcvr_{\perm(j)}}]{\klen+i}{\kw\!\parallel\!\msgi[\Sndr]{i}}^{\top})\chl[j,0]$}  % +\ZRF[\Rcvr_{\perm(j)}]{i}{\msgi[\kw]{i}}^{\top}
                \item[]\hspace{0.3cm}\lfbox[border-style=dashed]{$+(\Bzs\SZF[{\Sndr,\Rcvr_{\perm(j)}}]{\klen+i}{\kw\!\parallel\!\msgi[\Sndr]{i}}^{\top}+\Bos\SOF[{\Sndr,\Rcvr_{\perm(j)}}]{\klen+i}{\kw\!\parallel\!\msgi[\Sndr]{i}}^{\top})\chl[j,0]$} %+\Bos(\XOF[\Sndr]{i}{\msgi[\kw]{i}}^{\top}+\ZOF[\Rcvr_{\perm(j)}]{i}{\msgi[\kw]{i}}^{\top})  
                \item[]\hspace{0.3cm}\lfbox[background-color=lightgray!60!white,border-color=lightgray!60!white]{$+(\Bzs\SZF[{\Sndr,\Rcvr_{\perm(j)}}]{\klen+i+1}{\kw\!\parallel\!\msgi[\Sndr]{i+1}}^{\top}+\Bos\SOF[{\Sndr,\Rcvr_{\perm(j)}}]{\klen+i}{\kw\!\parallel\!\msgi[\Sndr]{i}}^{\top})\chl[j,0]$}  %+(\Bzs\ZZF[\Rcvr_{\perm(j)}]{i}{\msgi[\kw]{i}}^{\top}+\Bos\ZOF[\Rcvr_{\perm(j)}]{i}{\msgi[\kw]{i}}^{\top})
                \item[]\hspace{0.3cm}\lfbox{$+(\Bzs\SZF[{\Sndr,\Rcvr_{\perm(j)}}]{\klen+i+1}{\kw\!\parallel\!\msgi[\Sndr]{i+1}}^{\top}+\Bos\SOF[{\Sndr,\Rcvr_{\perm(j)}}]{\klen+i+1}{\kw\!\parallel\!\msgi[\Sndr]{i+1}}^{\top})\chl[j,0]$}  % +(\Bzs\ZZF[\Rcvr_{\perm(j)}]{i+1}{\msgi[\kw]{i+1}}^{\top}+\Bos\ZOF[\Rcvr_{\perm(j)}]{i+1}{\msgi[\kw]{i+1}}^{\top})
                % \item[]\hspace{0.3cm}\lfbox[background-color=lightgray!60!white,border-color=black]{$+\Bp\RF{\Sndr}{\kw}{\Rcvr_{\perm(j)}}^{\top}\chl[j,0]$}  
%%%                \item[]\lfbox[border-style=double,border-width=2pt]{$\chl[j,1]\sample \Zp^{\nRB}$}
                \State Return $\Chl$ %=((\On{\chl[j,0]},\On{\chl[j,1]})_{j\in[\ell]})$
                \item[]
            \end{algorithmic}        
        \underline{$\OCor(j)$}
            \begin{algorithmic}[1]
                \State $\QCr=\QCr\sqcup \{(j,\sky_j)\}$ from $\QSk$
                \State Return $\sky_j$
                % \item[]
            \end{algorithmic}    
        \end{minipage}
        \vline \hspace{1pt}
        \begin{minipage}{0.68\textwidth}
            \underline{$\OKey(j)$}
            \begin{algorithmic}[1]
                \State For $(\ii,b)\in[\klen]\times\{0,1\}$: $\J[j,\ii,b]\sample \Zp^{(\nRAd)\times (\nRB)}$.
                \State{$\color{\confcolorr}\X[j] = \SumiDT{1}{\ulen}\J[j,i,{\msgii[j]{i}}],\Z[j] =\SumiDT{\ulen+1}{2\ulen}\J[j,i,{\msgii[j]{i}}]$} 
                \State $\sky_j=\left(\Tw{\X[j]\B},\Tw{\Z[j]\B}\right)$
                \State $\QSk=\QSk\sqcup \{(j,\sky_j)\}$
                \State Return $(j,{\color{\confcolor}\Tw{\X[j]\B}})$
                \item[]
                % \item[]
            \end{algorithmic}
            \underline{$\OTag({\Sndrr},\kwd,{\Rcvrr})$}   
            \begin{algorithmic}[1]
                \State $\QTg=\QTg\sqcup \{({\Sndrr},\kwd,{\Rcvrr})\}$
                %\State $\r\sample\Zp^{\nCB}$
                \State $\Tg=(\Tw{\tg[{0}]},\Tw{\tg[{1}]})$ where
                \item[]$\tg[{0}]\sample\Zp^{\nRB}$ %={\B\r}$
                \item[]\lfbox[dotted]{$\tg[{0}] = \B\r+\B_{\Sndrr_{i+1}}\r_{\Sndrr_{i+1}}$ for $\r,\r_{\Sndrr_{i+1}}\sample\Zp^{\nCB}$}   
                \item[]$\tg[{1}]={(\Sum{\ii}{[\klen]}\N[\ii,{\msgii[\kwd]{\ii}}]+\X[\Sndrr]+\Z[\Rcvrr])\tg[{0}]}$
                \item[]\hspace{0.3cm}\lfbox[rounded]{$+(\SRF[{\Sndrr,\Rcvrr}]{\klen+i}{\kwd\!\parallel\!\msgi[\Sndrr]{i}})\Bpt\tg[{0}]$}  
                \item[]\hspace{0.3cm}\lfbox[border-style=dashed]{$+((\SZF[{\Sndrr,\Rcvrr}]{\klen+i}{\kwd\!\parallel\!\msgi[\Sndrr]{i}}\Bzst+\SOF[{\Sndrr,\Rcvrr}]{\klen+i}{\kwd\!\parallel\!\msgi[\Sndrr]{i}}\Bost))\tg[{0}]$}   
                %{$+((\XZF[\Sndrr]{i}{\msgi[\kwd]{i}}\Bzst+\XOF[\Sndrr]{i}{\msgi[\kwd]{i}}\Bost)+(\ZZF[\Rcvrr]{i}{\msgi[\kwd]{i}}\Bzst+\ZOF[\Rcvrr]{i}{\msgi[\kwd]{i}})\Bost)\tg[{0}]$}   
                \item[]\hspace{0.3cm}\lfbox[background-color=lightgray!60!white,border-color=lightgray!60!white]{$+(\SZF[{\Sndrr,\Rcvrr}]{\klen+i+1}{\kwd\!\parallel\!\msgi[\Sndrr]{i+1}}\Bzst+\SOF[{\Sndrr,\Rcvrr}]{\klen+i}{\kwd\!\parallel\!\msgi[\Sndrr]{i}}\Bost)\tg[{0}]$} 
                %{$+((\XZF[\Sndrr]{i+1}{\msgi[\kwd]{i+1}}\Bzst+\XOF[\Sndrr]{i+1}{\msgi[\kwd]{i+1}}\Bost)+(\ZZF[\Rcvrr]{i}{\msgi[\kwd]{i}}\Bzst+\ZOF[\Rcvrr]{i}{\msgi[\kwd]{i}})\Bost)\tg[{0}]$}
                \item[]\hspace{0.2cm}
                \lfbox{$+(\ZF[{\Sndrr,\Rcvrr}]{\klen+i+1}{\kwd\!\parallel\!\msgi[\Sndrr]{i+1}}\Bzst+\SOF[{\Sndrr,\Rcvrr}]{\klen+i+1}{\kwd\!\parallel\!\msgi[\Sndrr]{i+1}}\Bost)\tg[{0}]$} 
                %\lfbox{$+((\XZF[\Sndrr]{i+1}{\msgi[\kwd]{i+1}}\Bzst+\XOF[\Sndrr]{i+1}{\msgi[\kwd]{i+1}}\Bost)+(\ZZF[\Rcvrr]{i+1}{\msgi[\kwd]{i+1}}\Bzst+\ZOF[\Rcvrr]{i+1}{\msgi[\kwd]{i+1}})\Bost)\tg[{0}]$}
                %\lfbox[background-color=lightgray!60!white,border-color=black]{$+\RF{\Sndrr}{\kwd}{\Rcvrr}\Bpt\tg[{0}]$}  
%%%                \item[]\hspace{0.0cm}\lfbox[background-color=lightgray!60!white,border-color=black]{$\tg[{1}]\sample \Zp^{\nRAd}$}
                % \item[]\lfbox[border-style=double,border-width=2pt]{$\tg[{1}]\sample \Zp^{\nRAd}$}
                \State Return $\Tg$ 
                \item[]
            \end{algorithmic}
            \underline{$\Finalize(\bee\in\{0,1\})$}
            \begin{algorithmic}[1]
                \State Return $\bee\wedge\ (\QTg\cap\QChl=\emptyset)$
                % \item[]\hspace{0.2cm}${\color{\confcolor}\wedge\ (\QSk\subset \SSk)}$
                \item[]\hspace{0.2cm}${\color{\confcolor}\wedge\ (\forall\ (u,\dontcare)\in\QCr, (u,\dontcare,\dontcare),(\dontcare,\dontcare,u)\notin {\color{\confcolor}\QTg\cup}\QChl)}$
                \item[]
                % \item[]
            \end{algorithmic}  
        \end{minipage}
    }
    }
    \caption{Hybrids for the transition from {\color{\confcolor}$\hame{3,\klen,1}$ to $\hame{3,\klen+\ulen-1,5}$}.}
    \label{fig:CoreS-Hybrid}
\end{figure*}

\begin{figure*}[!h]
    \scriptsize
    %\centering
    \hspace{-2.9cm}
    \fbox{\parbox[c][134mm][c]{1.55\textwidth}{
        \centering\hspace{-1.0cm}{}\lfbox[rounded]{$\hame{4,\klen+\ulen+i,1}$},\lfbox[dotted]{$\hame{4,\klen+\ulen+i,2}$,\lfbox[border-style=dashed]{$\hame{4,\klen+\ulen+i,3}$},\lfbox[background-color=lightgray!60!white,border-color=lightgray!60!white]{$\hame{4,\klen+\ulen+i,4}$},\lfbox{$\hame{4,\klen+\ulen+i,5}$}}, 
        \lfbox[background-color=lightgray!60!white,border-color=black]{$\hame{5}$}, 
        \lfbox[border-style=double,border-width=2pt]{$\hame{6}$}\\ \ \\
        \hspace{-0.2cm}
        \begin{minipage}{0.60\textwidth}%{0.80\textwidth}
            \underline{$\Init(1^\secpp)$}
            \begin{algorithmic}[1]
                \State $\abG\sample \ABSGen(1^\secpp)$ %=(p,\Go,\Gt,\GT,\go,\gt,e)
                \State $\B$,\lfbox[dotted]{$\Bz,\Bo$}, \lfbox[border-style=dashed]{\lfbox{$\Bzs,\Bos$}} $\sample\Dk[\nRB,\nCB]$ s.t.
                \item[$-$] $(\B,\Bz,\Bo)$ generates $\Zp^{\nRB}$
                \item[$-$] $(\Bzs,\Bos)$ generates $\Bp$
                \item[$-$] $\Bt\Bzs=\Zro$, $\Bot\Bzs=\Zro$, $\Bt\Bos=\Zro$, $\Bzt\Bos=\Zro$ 
                \State {\color{\confcolorr}For $(\ii,b)\in[\klen]\times\{0,1\}: \N[\ii,b]\sample \Zp^{(\nRAd)\times (\nRB)}$} %$\qquad\bslash \hame{1}-\hame{4}$
                \State $\ppC=(\On{\B},\Tw{\B},{\color{\confcolorrr}\left\{\Tw{\N[\ii,b]\B}\right\}_{\!\substack{\!\ii\in [\klen] \\ \!b\in\{0,1\}}}})$
                \item[]
                % \item[]
                % \item[]
                % \item[]
            \end{algorithmic}
            \underline{$\OChal({\Sndr}, \kw, {\RSet})$}
            \begin{algorithmic}[1]
                \State Let $\RSet=\{\Rcvr_1,\ldots,\Rcvr_{\ell}\}$
                \State $\perm\sample \Perm{[\ell]}$ 
                \State $\QChl=\QChl\cup \{({\Sndr},\kw,{\Rcvr_i})_{i\in[\ell]}\}$
                \State $\Chl=((\On{\chl[\ij,0]},\On{\chl[\ij,1]})_{\ij\in[\ell]})$ s.t.
                % \State Return $\Chl=((\On{\chl[j,0]},\On{\chl[j,1]})_{j\in[\ell]})$
                \item[]$\chl[j,0]\sample\Zp^{\nRAd}$ 
                \item[]$\chl[j,1]={(\Sum{\ii}{[\klen]}\Nt[\ii,{\msgii[\kw]{\ii}}]+\Xt[\Sndr]+\Zt[\Rcvr_{\perm(j)}])\chl[j,0]}$
                \item[]\hspace{0.3cm}\lfbox[rounded]{$+\Bp(\RRF[{\Sndr,\Rcvr_{\perm(j)}}]{\klen+\ulen+i}{\kw\!\parallel\!\Sndr\!\parallel\!\msgi[{{\Rcvr_{\perm(j)}}}]{i}}^{\top})\chl[j,0]$}  % +\ZRF[\Rcvr_{\perm(j)}]{i}{\msgi[\kw]{i}}^{\top}
                \item[]\hspace{0.3cm}\lfbox[border-style=dashed]{$+(\Bzs\RZF[{\Sndr,\Rcvr_{\perm(j)}}]{\klen+\ulen+i}{\kw\!\parallel\!\Sndr\!\parallel\!\msgi[{{\Rcvr_{\perm(j)}}}]{i}}^{\top}+\Bos\ROF[{\Sndr,\Rcvr_{\perm(j)}}]{\klen+\ulen+i}{\kw\!\parallel\!\Sndr\!\parallel\!\msgi[{{\Rcvr_{\perm(j)}}}]{i}}^{\top})\chl[j,0]$} %+\Bos(\XOF[\Sndr]{i}{\msgi[\kw]{i}}^{\top}+\ZOF[\Rcvr_{\perm(j)}]{i}{\msgi[\kw]{i}}^{\top})  
                \item[]\hspace{0.3cm}\lfbox[background-color=lightgray!60!white,border-color=lightgray!60!white]{$+(\Bzs\RZF[{\Sndr,\Rcvr_{\perm(j)}}]{\klen+\ulen+i+1}{\kw\!\parallel\!\Sndr\!\parallel\!\msgi[{{\Rcvr_{\perm(j)}}}]{i+1}}^{\top}+\Bos\ROF[{\Sndr,\Rcvr_{\perm(j)}}]{\klen+\ulen+i}{\kw\!\parallel\!\Sndr\!\parallel\!\msgi[{{\Rcvr_{\perm(j)}}}]{i}}^{\top})\chl[j,0]$}  %+(\Bzs\ZZF[\Rcvr_{\perm(j)}]{i}{\msgi[\kw]{i}}^{\top}+\Bos\ZOF[\Rcvr_{\perm(j)}]{i}{\msgi[\kw]{i}}^{\top})
                \item[]\hspace{0.3cm}\lfbox{$+(\Bzs\RZF[{\Sndr,\Rcvr_{\perm(j)}}]{\klen+\ulen+i+1}{\kw\!\parallel\!\Sndr\!\parallel\!\msgi[{{\Rcvr_{\perm(j)}}}]{i+1}}^{\top}\chl[j,0]$}
                \item[]\hspace{3cm}\lfbox{$+\Bos\ROF[{\Sndr,\Rcvr_{\perm(j)}}]{\klen+\ulen+i+1}{\kw\!\parallel\!\Sndr\!\parallel\!\msgi[{{\Rcvr_{\perm(j)}}}]{i+1}}^{\top})\chl[j,0]$}  % +(\Bzs\ZZF[\Rcvr_{\perm(j)}]{i+1}{\msgi[\kw]{i+1}}^{\top}+\Bos\ZOF[\Rcvr_{\perm(j)}]{i+1}{\msgi[\kw]{i+1}}^{\top})
                % \item[]\hspace{0.3cm}\lfbox[background-color=lightgray!60!white,border-color=black]{$+\Bp\RF{\Sndr}{\kw}{\Rcvr_{\perm(j)}}^{\top}\chl[j,0]$}  
                \item[]\lfbox[border-style=double,border-width=2pt]{$\chl[j,1]\sample \Zp^{\nRB}$}
                \State Return $\Chl$ %=((\On{\chl[j,0]},\On{\chl[j,1]})_{j\in[\ell]})$
                \item[]
            \end{algorithmic}        
        \underline{$\OCor(j)$}
            \begin{algorithmic}[1]
                \State $\QCr=\QCr\sqcup \{(j,\sky_j)\}$ from $\QSk$
                \State Return $\sky_j$
                % \item[]
            \end{algorithmic}    
        \end{minipage}
        \vline \hspace{1pt}
        \begin{minipage}{0.72\textwidth}
            \hspace{-0.05cm}
            \underline{$\OKey(j)$}
            \begin{algorithmic}[1]
                \State For $(\ii,b)\in[\klen]\times\{0,1\}$: $\J[j,\ii,{b}]\sample\Zp^{(\nRAd)\times (\nRB)}$
    %\item[]\hspace{0.35cm}$\X[j,\ii,b],\Z[j,\ii,b]\sample\Zp^{(\nRAd)\times (\nRB)}$ $\qquad\qquad\qquad\quad\bslash \hame{0}$
                \State ${\color{\confcolorr}\X[j] = \SumiDT{1}{\ulen}\J[j,i,{\msgii[j]{i}}], \Z[j] =  \SumiDT{\ulen+1}{2\ulen}\J[j,i,{\msgii[j]{i}}]}$ %$\quad\bslash \hame{1}-\hame{4}$   
                \State $\sky_j=(\left(\Tw{\X[j]\B},\Tw{\Z[j]\B}\right)_{\substack{{\ii\in[\klen]}\\{b\in\{0,1\}}}})$
                \State $\QSk=\QSk\sqcup \{(j,\sky_j)\}$
                \State Return $(j,{\color{\confcolor}\Tw{\X[j]\B}})$
                \item[]
                % \item[]
            \end{algorithmic}
            \underline{$\OTag({\Sndrr},\kwd,{\Rcvrr})$}   
            \begin{algorithmic}[1]
                \State $\QTg=\QTg\sqcup \{({\Sndrr},\kwd,{\Rcvrr})\}$
                %\State $\r\sample\Zp^{\nCB}$
                \State $\Tg=(\Tw{\tg[{0}]},\Tw{\tg[{1}]})$ where
                \item[]$\tg[{0}]\sample\Zp^{\nRB}$ %={\B\r}$
                \item[]\lfbox[dotted]{$\tg[{0}] = \B\r+\B_{\Rcvrr_{i+1}}\r_{\Rcvrr_{i+1}}$ for $\r,\r_{\Rcvrr_{i+1}}\sample\Zp^{\nCB}$}   
                \item[]$\tg[{1}]={(\Sum{\ii}{[\klen]}\N[\ii,{\msgii[\kwd]{\ii}}]+\X[\Sndrr]+\Z[\Rcvrr])\tg[{0}]}$
                \item[]\hspace{0.3cm}\lfbox[rounded]{$+(\RF[{\Sndrr,\Rcvrr}]{\klen+\ulen+i}{\kwd\!\parallel\!\Sndrr\!\parallel\!\msgi[\Rcvrr]{i}})\Bpt\tg[{0}]$}  
                \item[]\hspace{0.3cm}\lfbox[border-style=dashed]{$+((\ZF[{\Sndrr,\Rcvrr}]{\klen+\ulen+i}{\kwd\!\parallel\!\Sndrr\!\parallel\!\msgi[\Rcvrr]{i}}\Bzst+\OF[{\Sndrr,\Rcvrr}]{i}{\msgi[\Rcvrr]{i}}\Bost))\tg[{0}]$}   
                %{$+((\XZF[\Sndrr]{i}{\msgi[\kwd]{i}}\Bzst+\XOF[\Sndrr]{i}{\msgi[\kwd]{i}}\Bost)+(\ZZF[\Rcvrr]{i}{\msgi[\kwd]{i}}\Bzst+\ZOF[\Rcvrr]{i}{\msgi[\kwd]{i}})\Bost)\tg[{0}]$}   
                \item[]\hspace{0.3cm}\lfbox[background-color=lightgray!60!white,border-color=lightgray!60!white]{$+(\ZF[{\Sndrr,\Rcvrr}]{\klen+\ulen+i+1}{\kwd\!\parallel\!\Sndrr\!\parallel\!\msgi[\Rcvrr]{i+1}}\Bzst\!+\!\OF[{\Sndrr,\Rcvrr}]{\klen+\ulen+i}{\kwd\!\parallel\!\Sndrr\!\parallel\!\msgi[\Rcvrr]{i}}\Bost)\tg[{0}]$} 
                %{$+((\XZF[\Sndrr]{i+1}{\msgi[\kwd]{i+1}}\Bzst+\XOF[\Sndrr]{i+1}{\msgi[\kwd]{i+1}}\Bost)+(\ZZF[\Rcvrr]{i}{\msgi[\kwd]{i}}\Bzst+\ZOF[\Rcvrr]{i}{\msgi[\kwd]{i}})\Bost)\tg[{0}]$}
                \item[]\hspace{0.2cm}
                \lfbox{$+(\ZF[{\Sndrr,\Rcvrr}]{\klen+\ulen+i+1}{\kwd\!\parallel\!\Sndrr\!\parallel\!\msgi[\Rcvrr]{i+1}}\Bzst\tg[0]$}
                \item[]\hspace{3cm}
                \lfbox{$+\OF[{\Sndrr,\Rcvrr}]{\klen+\ulen+i+1}{\kwd\!\parallel\!\Sndrr\!\parallel\!\msgi[\Rcvrr]{i+1}}\Bost)\tg[{0}]$} 
                %\lfbox{$+((\XZF[\Sndrr]{i+1}{\msgi[\kwd]{i+1}}\Bzst+\XOF[\Sndrr]{i+1}{\msgi[\kwd]{i+1}}\Bost)+(\ZZF[\Rcvrr]{i+1}{\msgi[\kwd]{i+1}}\Bzst+\ZOF[\Rcvrr]{i+1}{\msgi[\kwd]{i+1}})\Bost)\tg[{0}]$}
                %\lfbox[background-color=lightgray!60!white,border-color=black]{$+\RF{\Sndrr}{\kwd}{\Rcvrr}\Bpt\tg[{0}]$}  
                \item[]\hspace{0.0cm}\lfbox[background-color=lightgray!60!white,border-color=black]{$\tg[{1}]\sample \Zp^{\nRAd}$}
                % \item[]\lfbox[border-style=double,border-width=2pt]{$\tg[{1}]\sample \Zp^{\nRAd}$}
                \State Return $\Tg$ 
                \item[]
            \end{algorithmic}
            \underline{$\Finalize(\bee\in\{0,1\})$}
            \begin{algorithmic}[1]
                \State Return $\bee\wedge\ (\QTg\cap\QChl=\emptyset)$
                % \item[]\hspace{0.2cm}${\color{\confcolor}\wedge\ (\QSk\subset \SSk)}$
                \item[]\hspace{0.2cm}${\color{\confcolor}\wedge\ (\forall\ (u,\dontcare)\in\QCr, (u,\dontcare,\dontcare),(\dontcare,\dontcare,u)\notin {\color{\confcolor}\QTg\cup}\QChl)}$
                \item[]
                % \item[]
            \end{algorithmic}  
        \end{minipage}
    }
    }
    \caption{Hybrids for the transition from {\color{\confcolor}$\hame{4,\klen+\ulen,1}$ to $\hame{6}$}.}
    \label{fig:CoreR-Hybrid}
\end{figure*}

\subsection{Proof of Lemmas for Core-Lemma}\label{sec:Core-Lemma-Proofs}
% \begin{proof}[Proof of \Cref{lem:H{0}-H{1}}]
%     As $\J[\ii,b]$ are sampled uniformly at random for all $(\ii,b)\in[L]\times\{0,1\}$, $\X[j,\ii,b],\Z[j,\ii,b]$ are defined uniformly at random.
%     % % For every entity $j$ queried, we replace $\X[j,1,b]$ by $\X[j,1,b]+\XRF[j]{0}{\epsilon}\Bpt$ and $\Z[j,1,b]$ by $\Z[j,1,b]+\ZRF[j]{0}{\epsilon}\Bpt$ for $b\in\{0,1\}$. Observe that, $\sky_j$ stays the same. This modification updates $\tg[{1}]={\Sum{\ii}{[L]}(\X[\Sndrr,\ii,\kwd_{\ii}]+\Z[\Rcvrr,\ii,\kwd_{\ii}])\tg[{0}]}+(\XRF[\Sndrr]{i}{\msgi[\kwd]{i}}+\ZRF[\Rcvrr]{i}{\msgi[\kwd]{i}})\Bpt\tg[{0}]$ and $\chl[\ij,1]={\Sum{\ii}{[L]}(\Xt[\Sndr,\ii,\kw_{\ii}]+\Zt[\Rcvr_{\perm(\ij)},\ii,\kw_{\ii}])\chl[\ij,0]}+\Bp(\XRF[\Sndr]{i}{\msgi[\kw]{i}}^{\top}+\ZRF[\Rcvr_{\perm(\ij)}]{i}{\msgi[\kw]{i}}^{\top})\chl[\ij,0]$ for all $\ij\in[\ell]$.
%     % For all $(\rbtSndr,\rbtkw,\rbtRcvr)\in \QTg\sqcup\QChl$, we have $\RF[{\rbtSndr,\rbtRcvr}]{0}{\msgi[\rbtkw]{0}}=\RF[{\rbtSndr,\rbtRcvr}]{0}{\epsilon}$ and $\X[\rbtSndr,1,b]+\Z[\rbtRcvr,1,b]$ is identically distributed to $\X[\rbtSndr,1,b]+\Z[\rbtRcvr,1,b]+\RF[{\rbtSndr,\rbtRcvr}]{0}{\epsilon}\Bpt$ for $\X[\rbtSndr,1,b],\Z[\rbtRcvr,1,b]\sample \Zp^{(\nRAd)\times (\nRB)}$ where $b\in\{0,1\}$.
% \end{proof}

\begin{proof}[Proof of \Cref{lem:H{1}-H{2,0,1}}]
    Suppose we replace $\N[\ii,b]$ by $\N[\ii,b]+\RF[\rbtSndr]{0}{\msgi[\rbtkw]{0}}\Bpt$ for $b\in\{0,1\}$. Observe that, $\pp$ stays the same. 
    %This modification updates $\tg[{1}]={\Sum{\ii}{[L]}(\X[\Sndrr,\ii,\kwd_{\ii}]+\Z[\Rcvrr,\ii,\kwd_{\ii}])\tg[{0}]}+(\XRF[\Sndrr]{i}{\msgi[\kwd]{i}}+\ZRF[\Rcvrr]{i}{\msgi[\kwd]{i}})\Bpt\tg[{0}]$ and $\chl[\ij,1]={\Sum{\ii}{[L]}(\Xt[\Sndr,\ii,\kw_{\ii}]+\Zt[\Rcvr_{\perm(\ij)},\ii,\kw_{\ii}])\chl[\ij,0]}+\Bp(\XRF[\Sndr]{i}{\msgi[\kw]{i}}^{\top}+\ZRF[\Rcvr_{\perm(\ij)}]{i}{\msgi[\kw]{i}}^{\top})\chl[\ij,0]$ for all $\ij\in[\ell]$.
    For all $(\rbtSndr,\rbtkw,\rbtRcvr)\in \QTg\sqcup\QChl$, we have $\RF[{\rbtSndr,\rbtRcvr}]{0}{\msgi[\rbtkw]{0}}=\RF[{\rbtSndr,\rbtRcvr}]{0}{\epsilon}$ and $\N[\ii,b]$ is identically distributed to $\N[\ii,b]+\RF[{\rbtSndr,\rbtRcvr}]{0}{\epsilon}\Bpt$ for $b\in\{0,1\}$. %$\X[\rbtSndr,1,b],\Z[\rbtRcvr,1,b]\sample \Zp^{(\nRAd)\times (\nRB)}$ where 
\end{proof}

\begin{proof}[Proof of \Cref{lem:H{2,i,1}-H{2,i,2}}]
    To argue this indistinguishability, we introduce an intermediate game $\hame{2,i,1.5}$ where $\tg[{0}]$ is sampled from $\Span(\B||\Bz)$ if $\kwd_{i+1}=0$ and is sampled from $\Zp^{\nRB}$ if $\kwd_{i+1}=1$. We then argue that, $\hame{2,i,1}$ and $\hame{2,i,1.5}$ are computationally indistinguishable under $|\QTg|\mhyf$fold $\matDH[\nRB,\nCB]$ assumption. The same argument can also be applied to show $\hame{2,i,1.5}$ and $\hame{2,i,2}$ are computationally indistinguishable.

    To argue the indistinguishability of $\hame{2,i,1}$ and $\hame{2,i,1.5}$, we devise an $\ppt$ algorithm $\AB$ next. Let $\AB$ receive $|\QTg|\mhyf$fold $\matDH[\nRB,\nCB]$ instance $(\Tw{\M},\Tw{\f[1]},\ldots,\Tw{\f[|\QTg|]})$.
    \begin{description}
        \item[$\Setup$.] $\AB$ samples $\B,\Bo\sample\Zp^{\nRB\times \nCB}$ and $\Bp\sample\Zp^{\nRB\times\nRCB}$ s.t. $\Bt\Bp=\Zro$. $\AB$ also chooses $\N[\ii,b]\sample \Zp^{(\nRAd)\times (\nRB)}$ for all $(\ii,b)\in[\klen]\times\{0,1\}$ uniformly at random. 
        \item[$\OKey$.] On every query on $j$, $\AB$ samples $\J[\ii,b]\sample \Zp^{(\nRAd)\times (\nRB)}$ for all $(\ii,b)\in[2\ulen]\times\{0,1\}$ uniformly at random sets $\X[j] = \SumiDT{1}{\ulen}\J[i,{\msgii[j]{i}}]$ and $\Z[j] = \SumiDT{\ulen+1}{2\ulen}\J[i,{\msgii[j]{i}}]$  for all $(\ii,b)\in[2\ulen]\times\{0,1\}$. $\AB$ defines random functions $\RF[j]{i}{\msgi[\kwd]{i}}\in \Zp^{\nRAd\times \nRCB}$ on the fly. It updates $(j,\sky_j)$ in $\QSk$ and returns $j$. %$\sky_j=(\left(\Tw{\X[j,\ii,b]\B},\Tw{\Z[j,\ii,b]\B}\right)_{\substack{{\ii\in[L]}\\{b\in\{0,1\}}}})$. 
        \item[$\OCor$.] On every query on $j$, $\AB$ executes $\OKey(j)$ to retrieve $\sky_j$ if $(j,\dontcare)\notin\QSk$. It finally returns $\sky_j$ and updates $(j,\sky_j)$ in $\QCr$. 
        \item[$\OChal$.] Follows the description in \Cref{fig:Core-Hybrid}.
        \item[$\OTag$.] For any $j^{th}$ query $j\in[|\QTg|]$, $\AB$ simulates $\tg[{0}]$ as following:
        {\small\[
            \tg[{0}]=\begin{cases}
                \B\r+\f[j] & \text{for } \r\sample\Zp^{\nCB} \text{ if } \kwd_{i+1}=0\\
                \z & \text{for } \z\sample\Zp^{\nRB} \text{ if } \kwd_{i+1}=1\\
            \end{cases}
        \]}
    \end{description}
    Since the challenger chose $\M\sample\Zp^{\nRB\times \nCB}$ and $\AB$ chose $\B,\Bo\sample\Zp^{\nRB\times \nCB}$ uniformly at random, $(\B||\M||\Bo)$ is a full-rank matrix with overwhelming probability {\color{\confcolor}$1-\frac{\nCB+\nCB}{p-1}$}.
    Thus, $|\adv{\AA}{\hame{2,i,1}}-\adv{\AA}{\hame{2,i,1.5}}|\leq \Adv{\AB,\Gt}{|\QTg|\mhyf\matDH[\nRB,\nCB]}+\frac{2\times\nCB}{p-1}$.
    A similar argument shows that $|\adv{\AA}{\hame{2,i,1.5}}-\adv{\AA}{\hame{2,i,2}}|\leq \Adv{\AB,\Gt}{|\QTg|\mhyf\matDH[\nRB,\nCB]}+\frac{2\times\nCB}{p-1}$.
\end{proof}

\begin{proof}[Proof of \Cref{lem:H{2,i,2}-H{2,i,3}}]
    We show that, these two games are statistically indistinguishable. In $\hame{2,i,2}$, we rewrite $\RF[{\rbtSndr,\rbtRcvr}]{i}{\msgi[\rbtkw]{i}}$ as following: %and $\OF[j]{i}{\msgi[\kwd]{i}}$
    {\small    
    \begin{equation*}
        \begin{aligned}
            \RF[{\rbtSndr,\rbtRcvr}]{i}{\msgi[\rbtkw]{i}}\Bpt = \ZF[{\rbtSndr,\rbtRcvr}]{i}{\msgi[\rbtkw]{i}}\Bzs+\OF[{\rbtSndr,\rbtRcvr}]{i}{\msgi[\rbtkw]{i}}\Bos ,%&\quad &
            % \ZRF[j]{i}{\msgi[\kwd]{i}}\Bpt = \ZZF[j]{i}{\msgi[\kwd]{i}}\Bzs+\ZOF[j]{i}{\msgi[\kwd]{i}}\Bos\\
        \end{aligned}
    \end{equation*}
    }

    \noindent for all $(\rbtSndr,\rbtkw,\rbtRcvr)\in \QTg\sqcup\QChl$ where $\ZFdef[{\rbtSndr,\rbtRcvr}]{i},\OFdef[{\rbtSndr,\rbtRcvr}]{i}:\{0,1\}^{i}\rightarrow \Zp^{\nRAd\times \nCB}$ are independent random functions.
    Since, $(\Bzs||\Bos)$ generates $\Bpt$ with overwhelming probability $1-\frac{\nCB}{p}$, and $\ZFdef[{\rbtSndr,\rbtRcvr}]{i},\OFdef[{\rbtSndr,\rbtRcvr}]{i}:\{0,1\}^{i}\rightarrow \Zp^{\nRAd\times \nCB}$ are independent random functions, $\RFdef[{\rbtSndr,\rbtRcvr}]{i}:\{0,1\}^{i}\rightarrow \Zp^{\nRAd\times \nRCB}$ are random functions as well. %,$ $\ZZFdef[j]{i},\ZOFdef[j]{i},\ZRFdef[j]{i}
\end{proof}

\begin{proof}[Proof of \Cref{lem:H{2,i,3}-H{2,i,4}}]
    To argue the indistinguishability of $\hame{2,i,3}$ and $\hame{2,i,4}$, we define 

    {\small
    \begin{equation*}
            \begin{aligned}
                \ZF[{\rbtSndr,\rbtRcvr}]{i+1}{\msgi[\rbtkw]{i+1}} =
                \begin{cases}
                    \ZF[{\rbtSndr,\rbtRcvr}]{i}{\msgi[\rbtkw]{i}} & \text{if } \rbtkw_{i+1}=0\\
                    \ZF[{\rbtSndr,\rbtRcvr}]{i}{\msgi[\rbtkw]{i}}+\tZF[{\rbtSndr,\rbtRcvr}]{i}{\msgi[\rbtkw]{i}} & \text{if } \rbtkw_{i+1}=1\\                    
                \end{cases},
            \end{aligned}
        \end{equation*}
    }

    \noindent for all $(\rbtSndr,\rbtkw,\rbtRcvr)\in \QTg\sqcup\QChl$ where $\tZFdef[{\rbtSndr,\rbtRcvr}]{i}:\{0,1\}^{i}\rightarrow \Zp^{\nRAd\times \nCB}$ are independent random functions. %,\tZZFdef[j]{i}

    We devise the reduction $\AB$ next in \Cref{fig:AB-H{2,i,3}-H{2,i,4}}. 
    $\AB$ receives $\nRAd|\QChl|\mhyf$fold $\matDH[\nRA,\nCA]$ instance $(\On{\M},\On{\f[1]},\ldots,\On{\f[\nRAd|\QChl|]})$. 
    % Observe that, $\{\sky_j\}_{j\in\QSk}$ stayed the same even after entropy injection via $\XZF[j]{0}{\epsilon},\XOF[j]{0}{\epsilon},$ $\ZZF[j]{0}{\epsilon},\ZOF[j]{0}{\epsilon}$. 
    It is easy to verify that, the $\OTag$ responses are same in both $\hame{2,i,3}$ and $\hame{2,i,4}$. We analyze $\OChal$ responses for the two games next. First we assume that $\Mu\in\Zp^{\nRAd\times \nRAd}$ is invertible which happens with overwhelming probability ${\color{\confcolor}1-\frac{1}{p-1}}$. Observe that $\F[c]=\iCol{\Mu\K[c]}{\Md\K[c]+\R[c]}$ where $\R[c]$ is either $\Zro_{\nRAd}$ or uniformly random in $\Zp^{\nRAd\times \nRAd}$. We assume that $\Fu[c]$ is invertible which happens with overwhelming probability ${\color{\confcolor}1-\frac{1}{p-1}}$. Thus, $\chl[\ij,0]=\Fu[c]\chl[\ij,0]'$ are uniformly random quantities as both $\Fu[c]$ and $\chl[\ij,0]'$ are uniformly random. Now we analyze $\chl[\ij,1]$ for all $\ij\in[\ell]$ for all the challenge queries.
    \begin{enumerate}
        \item For $\kw_{i+1}=0$, the $\OChal$ responses are same in both $\hame{2,i,3}$ and $\hame{2,i,4}$.
        \item For $\kw_{i+1}=1$, 
        {\scriptsize %\color{red}
            \begin{equation*}%\hspace{-1.0cm}
                \begin{aligned}
                    \chl[\ij,1]&={(\Sum{\ii}{[\klen]}\hNt[\ii,{\msgii[\kw]{\ii}}]+\hXt[\Sndr]+\hZt[\Rcvr_{\perm(j)}])\chl[\ij,0]}+\Bzs\ZF[{\Sndr,\Rcvr_{\perm(\ij)}}]{i}{\msgi[\kw]{i}}^{\top}\chl[\ij,0]\\
                    &\qquad\qquad\qquad\qquad\qquad\qquad+\Bos\OF[{\Sndr,\Rcvr_{\perm(\ij)}}]{i}{\msgi[\kw]{i}}^{\top}\chl[\ij,0]+{\Bzs\Fd[c]\chl[\ij,0]'}\\
                    &={(\Sum{\ii}{[\klen]}\hNt[\ii,{\msgii[\kw]{\ii}}]+\hXt[\Sndr]+\hZt[\Rcvr_{\perm(j)}])\chl[\ij,0]}
                    +\Bzs\ZF[{\Sndr,\Rcvr_{\perm(\ij)}}]{i}{\msgi[\kw]{i}}^{\top}\chl[\ij,0]\\
                    &\qquad\qquad+\Bos\OF[{\Sndr,\Rcvr_{\perm(\ij)}}]{i}{\msgi[\kw]{i}}^{\top}\chl[\ij,0]+{\Bzs\Md\Mu^{-1}\Fu[c]\chl[\ij,0]'+\Bzs\R[c]\chl[\ij,0]'}\\
                    &={(\Sum{\ii}{[\klen]\setminus\{i+1\}}\hNt[\ii,{\msgii[\kw]{\ii}}]+\hXt[\Sndr]+\hZt[\Rcvr_{\perm(j)}])\chl[\ij,0]}+\hNt[i+1,{\msgii[\kw]{i+1}}]\chl[\ij,0]\\
                    &\quad+\Bzs\Md\Mu^{-1}\chl[\ij,0]+(\Bzs\ZF[{\Sndr,\Rcvr_{\perm(\ij)}}]{i}{\msgi[\kw]{i}}^{\top}+\Bos\OF[{\Sndr,\Rcvr_{\perm(\ij)}}]{i}{\msgi[\kw]{i}}^{\top})\chl[\ij,0]\\
                    &\quad+\Bzs\R[c]\chl[\ij,0]'\\
                    &={(\Sum{\ii}{[\klen]}\Nt[\ii,{\msgii[\kw]{\ii}}]+\Xt[\Sndr]+\Zt[\Rcvr_{\perm(j)}])\chl[\ij,0]}+\Bzs\ZF[{\Sndr,\Rcvr_{\perm(\ij)}}]{i}{\msgi[\kw]{i}}^{\top}\chl[\ij,0]\\
                    &\qquad\qquad\qquad+(\Bzs\R[c]\Fu[c]^{-1}+\Bos\OF[{\Sndr,\Rcvr_{\perm(\ij)}}]{i}{\msgi[\kw]{i}}^{\top})\chl[\ij,0]\\                
                \end{aligned}
            \end{equation*}
        }
        \end{enumerate}
        If $\R[c]=\Zro_{\nRAd}$, $\AB$ simulates $\hame{2,i,3}$ and if $\R[c]$ is uniformly random, we implicitly set $\tZF[{\Sndr,\Rcvr_{\perm(\ij)}}]{i}{\msgi[\kw]{i}}=\R[c]\Fu[c]^{-1}$ and simulate $\hame{2,i,4}$. Thus, $|\adv{\AA}{\hame{2,i,3}}-\adv{\AA}{\hame{2,i,4}}|\leq \Adv{\AB,\Go}{\prll|\QChl|\mhyf\matDH[\nRA,\nCA]}+\frac{|\QChl|+2}{p-1}$.

    \begin{figure*}[ht]
        \scriptsize
        %\centering
        \hspace{-1.2cm}
        \fbox{
            % \parbox[c][143mm][c]{1.13\linewidth}{
            % \centering\hspace{-4.5cm}{$\hame{1}$},\lfbox[rounded]{$\hame{2,i,1}$},\lfbox[dotted]{$\hame{2,i,2}$,\lfbox[border-style=dashed]{$\hame{2,i,3}$},\lfbox[background-color=lightgray!60!white,border-color=lightgray!60!white]{$\hame{2,i,4}$},\lfbox{$\hame{2,i,5}$}}, \lfbox[background-color=lightgray!60!white,border-color=black]{$\hame{3}$}, \lfbox[border-style=double,border-width=2pt]{$\hame{4}$}\\ \ \\
            % \hspace{-1cm}
            \begin{minipage}{0.55\textwidth}%{0.47\textwidth}
                \underline{$\Init(1^\secpp)$}
                \begin{algorithmic}[1]
                    \State $\abG\sample \ABSGen(1^\secpp)$ %=(p,\Go,\Gt,\GT,\go,\gt,e)
                    \State From $\prll|\QChl|\mhyf$fold $\matDH[\nRA,\nCA]$ instance, 
                    \item[]\hspace{0.25cm}define $\F[c]=(\f[(c-1)\prll+1]||\ldots||\f[c\prll])\in\Zp^{\nRA\times \prll}$
                    \item[]\hspace{0.25cm}for $c\in[|\QChl|]$ 
                    \State $\B$,{$\Bz,\Bo$}, {{$\Bzs,\Bos$}} $\sample\Dk[\nRB,\nCB]$ s.t.
                    \item[$-$] $(\B,\Bz,\Bo)$ generates $\Zp^{\nRB}$
                    \item[$-$] $(\Bzs,\Bos)$ generates $\Bp$
                    \item[$-$] $\Bt\Bzs=\Zro$, $\Bot\Bzs=\Zro$ 
                    \item[$-$] $\Bt\Bos=\Zro$, $\Bzt\Bos=\Zro$ 
                    \State {\color{\confcolorr}For $(\ii,b)\in[\klen]\times\{0,1\}: \N[\ii,b]=\hN[\ii,b]\sample \Zp^{(\nRAd)\times (\nRB)}$}
                    \item[] $\bslash $ Implicit: $\N[i+1,b]=\hN[i+1,b]+\Mu^{-\top}\Mdt\Bzst$
                    \State $\ppC=(\On{\B},\Tw{\B},{\color{\confcolorrr}\left\{\Tw{\N[\ii,b]\B}\right\}_{\!\substack{\!\ii\in [\klen] \\ \!b\in\{0,1\}}}})$
                    \item[] 
                    % \item[]
                \end{algorithmic}
                \underline{$\OKey(j)$}
                \begin{algorithmic}[1]
                    \State For $(\ii,b)\in[2\ulen]\times\{0,1\}$: $\J[j,\ii,b]\sample \Zp^{(\nRAd)\times (\nRB)}$ 
                \State ${\color{\confcolorr}\hX[j] = \SumiDT{1}{\ulen}\J[j,i,{\msgii[j]{i}}],\hZ[j] = \SumiDT{\ulen+1}{2\ulen}\J[j,i,{\msgii[j]{i}}]}$ %$\quad\bslash \hame{1}-\hame{4}$   
                    \State $\X[j]=\hX[j],\quad\Z[j]=\hZ[j]$
                    % \item[]//Implicit: $\X[j,i+1,b]=\hX[j,i+1,b]+\Mu^{-\top}\Mdt\Bzst$
                    % \item[]//Implicit: $\Z[j,i+1,b]=\hZ[j,i+1,b]+\Mu^{-\top}\Mdt\Bzst$
                    % \item[]//for all $j\in \QTg\sqcup\QChl$ 
                    %}}}
                    \State $\sky_j=\left(\Tw{\X[j]\B},\Tw{\Z[j]\B}\right)$
                    \State $\QSk=\QSk\sqcup \{(j,\sky_j)\}$
                    \State Return $(j,{\color{\confcolor}\Tw{\X[j]\B}})$
                    \item[]
                    % \item[]
                \end{algorithmic}
                \underline{$\OCor(j)$}
                \begin{algorithmic}[1]
                    \State $\QCr=\QCr\sqcup \{(j,\sky_j)\}$ from $\QSk$
                    \State Return $\sky_j$
                    \item[]
                \end{algorithmic}     
            \end{minipage}
            \vline \hspace{1pt}
            \begin{minipage}{0.85\textwidth}%{0.75\textwidth}
                \underline{$\OTag({\Sndrr},\kwd,{\Rcvrr})$}   
                \begin{algorithmic}[1]
                    \State $\QTg=\QTg\sqcup \{({\Sndrr},\kwd,{\Rcvrr})\}$
                    %\State $\r\sample\Zp^{\nCB}$
                    \State $\Tg=(\Tw{\tg[{0}]},\Tw{\tg[{1}]})$ where
                    % \item[]$\tg[{0}]\sample\Zp^{\nRB}$ %={\B\r}$
                    \item[]{$\tg[{0}] = \B\r+\B_{\kwd_{i+1}}\r_{\kwd_{i+1}}$ for $\r,\r_{\kwd_{i+1}}\sample\Zp^{\nCB}$}   
                    \item[]$\tg[{1}]={(\Sum{\ii}{[\klen]}\N[\ii,{\msgii[\kwd]{\ii}}]+\X[\Sndrr]+\Z[\Rcvrr])\tg[{0}]}$
                    {$+((\ZF[{\Sndrr,\Rcvrr}]{i+1}{\msgi[\kwd]{i+1}}\Bzst+\OF[{\Sndrr,\Rcvrr}]{i}{\msgi[\kwd]{i}}\Bost))\tg[{0}]$} %+(\ZZF[\Rcvrr]{i}{\msgi[\kwd]{i}}\Bzst+\ZOF[\Rcvrr]{i}{\msgi[\kwd]{i}})\Bost
                    \item[] 
                    \item[]\hspace{0.3cm}$={(\Sum{\ii}{[\klen]}\hN[\ii,{\msgii[\kwd]{\ii}}]+\hX[\Sndrr]+\hZ[\Rcvrr])\tg[{0}]}$ 
                    % \item[]\hspace{0.35cm}
                    {$+((\ZF[{\Sndrr,\Rcvrr}]{i}{\msgi[\kwd]{i}}\Bzst+\OF[{\Sndrr,\Rcvrr}]{i}{\msgi[\kwd]{i}}\Bost))\tg[{0}]$} %+(\ZZF[\Rcvrr]{i}{\msgi[\kwd]{i}}\Bzst+\ZOF[\Rcvrr]{i}{\msgi[\kwd]{i}})\Bost   
                    \State Return $\Tg$ 
                    \item[]
                \end{algorithmic}
\underline{$\OChal({\Sndr}, \kw, {\RSet})$} \hspace{2cm}  
                \begin{algorithmic}[1]
                    \State Let $\RSet=\{\Rcvr_1,\ldots,\Rcvr_{\ell}\}$
                    \State $\perm\sample \Perm{[\ell]}$ 
                    \State $\QChl=\QChl\cup \{({\Sndr},\kw,{\Rcvr_\ij})_{\ij\in[\ell]}\}$
                    % \State $\Chl=((\On{\chl[\ij,0]},\On{\chl[\ij,1]})_{\ij\in[\ell]})$ s.t.
                    \State For all $(\Sndr,\Rcvr_{\perm(\ij)})\in \{\Sndr\}\times \RSet:$ 
                    \item[]\hspace{0.35cm}Let $(\Sndr,\msgi[\kw]{i}||\dontcare||\cdots||\dontcare,\Rcvr_{\perm(\ij)})$ be the $c^{th}$ entry of $\QChl$
                    \item[]\hspace{0.35cm}$\chl[\ij,0]=\Fu[c]\chl[\ij,0]'$ for $\chl[\ij,0]'\sample\Zp^{\prll}$ 
                    %\item[]\hspace{0.3cm}$={\Sum{\ii}{[\klen]}(\hXt[\Sndr,\ii,\kw_{\ii}]+\hZt[\Rcvr_{\perm(\ij)},\ii,\kw_{\ii}])\chl[\ij,0]}$
                    \item[]\hspace{0.35cm}$\chl[\ij,1]={(\Sum{\ii}{[\klen]}\hNt[\ii,{\msgii[\kw]{\ii}}]+\hXt[\Sndr]+\hZt[\Rcvr_{\perm(\ij)}])\chl[\ij,0]}$
                    % \item[]\hspace{0.35cm}
                    {$+((\Bzs\ZF[{\Sndr,\Rcvr}]{i}{\msgi[\kw]{i}}^{\top}+\Bos\OF[{\Sndr,\Rcvr}]{i}{\msgi[\kw]{i}}^{\top}))\chl[\ij,0]$}  % +(\Bzs\ZZF[\Rcvr_{\perm(\ij)}]{i}{\msgi[\kw]{i}}^{\top}+\Bos\ZOF[\Rcvr_{\perm(\ij)}]{i}{\msgi[\kw]{i}}^{\top})
                    % \item[] for // {\color{red}$c^{th}$ new $(\Sndr,\msgi[\kw]{i},\Rcvr)$ pair where $\Rcvr\in \RSet$}
                    % \item[]$\chl[\ij,1]={\Sum{\ii}{[\klen]}(\hXt[\Sndr,\ii,\kw_{\ii}]+\hZt[\Rcvr_{\perm(\ij)},\ii,\kw_{\ii}])\chl[\ij,0]}$
                    % \item[]\hspace{0.35cm}{$+((\Bzs\XZF[\Sndr]{i+1}{\msgi[\kw]{i+1}}^{\top}+\Bos\XOF[\Sndr]{i}{\msgi[\kw]{i}}^{\top})+(\Bzs\ZZF[\Rcvr_
                    % {\perm(\ij)}]{i}{\msgi[\kw]{i}}^{\top}+\Bos\ZOF[\Rcvr_{\perm(\ij)}]{i}{\msgi[\kw]{i}}^{\top}))\chl[\ij,0]$}  
                    \State If $\kw_{i+1}=1$: $\chl[\ij,1]=\chl[\ij,1]+{\Bzs\Fd[c]\chl[\ij,0]'}$
                    \State Return $\Chl=((\On{\chl[\ij,0]},\On{\chl[\ij,1]})_{\ij\in[\ell]})$
                    \item[]
                \end{algorithmic}        
                \underline{$\Finalize(\bee\in\{0,1\})$}
                \begin{algorithmic}[1]
                    \State Return $\bee\wedge\ (\QTg\cap\QChl=\emptyset){\color{\confcolor}\wedge\ (\forall  (u,\dontcare)\in\QCr, (u,\dontcare,\dontcare),(\dontcare,\dontcare,u)\notin {\color{\confcolor}\QTg\cup}\QChl)}$
                    \item[]
                \end{algorithmic}      
            \end{minipage}
        % }
        }
        \caption{The reduction $\AB$ in the proof of \Cref{lem:H{2,i,3}-H{2,i,4}}.}
        \label{fig:AB-H{2,i,3}-H{2,i,4}}
    \end{figure*}
\end{proof}

\begin{proof}[Proof of \Cref{lem:H{2,i,4}-H{2,i,5}}]
    % Similar to the proof of \Cref{lem:H{2,i,3}-H{2,i,4}} except we define an intermediate game $\hame{2,i,4.5}$.
    To argue indistinguishability of $\hame{2,i,4}$ and $\hame{2,i,5}$ following the proof of \Cref{lem:H{2,i,3}-H{2,i,4}}. 
    Here, we define 

    {\small
    \begin{equation*}
            \begin{aligned}
                \OF[{\rbtSndr,\rbtRcvr}]{i+1}{\msgi[\rbtkw]{i+1}} =
                \begin{cases}
                    \OF[{\rbtSndr,\rbtRcvr}]{i}{\msgi[\rbtkw]{i}} & \text{if } \rbtkw_{i+1}=1\\
                    \OF[{\rbtSndr,\rbtRcvr}]{i}{\msgi[\rbtkw]{i}}+\tOF[j]{i}{\msgi[\rbtkw]{i}} & \text{if } \rbtkw_{i+1}=0\\                    
                \end{cases},
            \end{aligned}
        \end{equation*}
    }

    \noindent for all $(\rbtSndr,\rbtkw,\rbtRcvr)\in \QTg\sqcup\QChl$ where $\tOFdef[{\rbtSndr,\rbtRcvr}]{i}:\{0,1\}^{i}\rightarrow \Zp^{\nRAd\times \nCB}$ are independent random functions. %,\tZZFdef[j]{i}
    Rest of the proof is same as proving indistinguishability of $\hame{2,i,3}$ and $\hame{2,i,4}$ with the roles of $0$ and $1$ swapped.
    Thus, $|\adv{\AA}{\hame{2,i,4}}-\adv{\AA}{\hame{2,i,5}}|\leq \Adv{\AB,\Go}{\prll|\QChl|\mhyf\matDH[\nRA,\nCA]}+\frac{|\QChl|+2}{p-1}$. %$\nRAd|\QChl|\mhyf$fold 

\end{proof}

\begin{proof}[Proof of \Cref{lem:H{2,k,1}-H{3,k,1}}]
    Note that, we define $\hame{3,\klen,1}=\hame{2,\klen,1}$ as a stepping stone to inject randomness into $\SRFdef[\rbtSndr]{\klen}(\rbtkw)$. 
    Looking ahead, the goal is going to define $\SRFdef[\rbtSndr]{\klen+i}$ depending on the $(\klen+i)$-length bit string $\rbtkw\!\parallel\!\msgi[\rbtSndr]{i}$ and utlize $\Sndrr_{i+1}$ to define $\tg[0]$ for any tag query on $(\Sndrr,\dontcare,\dontcare)$.
    Observe that, $(\Sum{\ii}{[\klen]}\N[\ii,{\msgii[\rbtkw]{\ii}}]+\X[\Sndr]+\Z[\Rcvr_{\perm(j)}])=\Sum{\ii}{[\klen]}\N[\ii,{\msgii[\rbtkw]{\ii}}]+\SumiDT[{\ii}]{1}{\ulen} \J[\rbtSndr,\ii,\rbtSndr_{\ii}]+\SumiDT[{\ii}]{\ulen+1}{2\ulen} \J[\rbtRcvr,\ii,\rbtRcvr_{\ii}]$.
    Thus, to attain the goal mentioned above, we modify $\SumiDT[{\ii}]{1}{\ulen} \J[\rbtSndr,\ii,\rbtSndr_{\ii}]+\SRFdef[\rbtSndr]{\klen}(\rbtkw)$ into $\SumiDT[{\ii}]{1}{\ulen} \J[\rbtSndr,\ii,\rbtSndr_{\ii}]+\SRFdef[\rbtSndr]{\klen+\ulen}(\rbtkw\!\parallel\!\rbtSndr)$ iteratively via $\hame{3,\klen,1},\ldots,\hame{3,\klen+\ulen,1}$.
\end{proof}

\begin{proof}[Proof of \Cref{lem:H{3,i,1}-H{3,i,2}}]
    This proof is same as the proof of proof of \Cref{lem:H{2,i,1}-H{2,i,2}} except that $\SRFdef[\rbtSndr]{\klen+i}$ takes a $(\klen+i)$-bit string an input $\rbtkw\!\parallel\! \msgi[\rbtSndr]{i}$.
\end{proof}

\begin{proof}[Proof of \Cref{lem:H{3,i,2}-H{3,i,3}}]
    This proof is same as the proof of proof of \Cref{lem:H{2,i,2}-H{2,i,3}} except that $\SZFdef[\rbtSndr]{\klen+i},\SOFdef[\rbtSndr]{\klen+i}$ both take a $(\klen+i)$-bit string an input $\rbtkw\!\parallel\! \msgi[\rbtSndr]{i}$.
\end{proof}

\begin{proof}[Proof of \Cref{lem:H{3,i,3}-H{3,i,4}}]
        This proof is same as the proof of proof of \Cref{lem:H{2,i,3}-H{2,i,4}} except that $\SZFdef[\rbtSndr]{\klen+i+1}$ takes a $(\klen+i+1)$-bit string an input $\rbtkw\!\parallel\! \msgi[\rbtSndr]{i+1}$ and $\SOFdef[\rbtSndr]{\klen+i}$ takes a $(\klen+i)$-bit string an input $\rbtkw\!\parallel\! \msgi[\rbtSndr]{i}$.
\end{proof}

\begin{proof}[Proof of \Cref{lem:H{3,i,4}-H{3,i,5}}]
    This proof is same as the proof of proof of \Cref{lem:H{2,i,4}-H{2,i,5}} except that $\SZFdef[\rbtSndr]{\klen+i+1},\SOFdef[\rbtSndr]{\klen+i+1}$ both take a $(\klen+i+1)$-bit string an input $\rbtkw\!\parallel\! \msgi[\rbtSndr]{i+1}$.
\end{proof}

\begin{proof}[Proof of \Cref{lem:H{3,ku,1}-H{4,ku,1}}]
    Note that, we define $\hame{4,\klen+\ulen,1}=\hame{3,\klen+\ulen,1}$ as a stepping stone to inject randomness into $\RRFdef[\rbtRcvr]{\klen+\ulen}(\rbtkw\!\parallel\!\rbtSndr)$. 
    Looking ahead, the goal is going to define $\RRFdef[\rbtRcvr]{\klen+\ulen+i}$ depending on the $(\klen+\ulen+i)$-length bit string $\rbtkw\!\parallel\!\rbtSndr\!\parallel\!\msgi[\rbtRcvr]{i}$ and utlize $\Rcvrr_{i+1}$ to define $\tg[0]$ for any tag query on $(\dontcare,\dontcare,\Rcvrr)$.
    As we already observed, $(\Sum{\ii}{[\klen]}\N[\ii,{\msgii[\rbtkw]{\ii}}]+\X[\Sndr]+\Z[\Rcvr_{\perm(j)}])=\Sum{\ii}{[\klen]}\N[\ii,{\msgii[\rbtkw]{\ii}}]+\SumiDT[{\ii}]{1}{\ulen} \J[\rbtSndr,\ii,\rbtSndr_{\ii}]+\SumiDT[{\ii}]{\ulen+1}{2\ulen} \J[\rbtRcvr,\ii,\rbtRcvr_{\ii}]$.
Thus, to attain the goal mentioned above, we modify $\SumiDT[{\ii}]{\ulen+1}{2\ulen} \J[\rbtRcvr,\ii,\rbtRcvr_{\ii}]+\SRFdef[\rbtRcvr]{\klen+\ulen}(\rbtkw\!\parallel\!\rbtSndr)$ into $\SumiDT[{\ii}]{\ulen+1}{2\ulen} \J[\rbtRcvr,\ii,\rbtRcvr_{\ii}]+\SRFdef[\rbtRcvr]{\klen+2\ulen}(\rbtkw\!\parallel\!\rbtSndr\!\parallel\!\rbtRcvr)$ iteratively via $\hame{4,\klen+\ulen,1},\ldots,\hame{4,\klen+2\ulen,1}$.
\end{proof}

\begin{proof}[Proof of \Cref{lem:H{4,i,1}-H{4,i,2}}]
    This proof is same as the proof of proof of \Cref{lem:H{3,i,1}-H{3,i,2}} except that $\RRFdef[\rbtRcvr]{\klen+\ulen+i}$ takes a $(\klen+\ulen+i)$-bit string an input $\rbtkw\!\parallel\!\rbtSndr\!\parallel\!\msgi[\rbtRcvr]{i}$.
\end{proof}

\begin{proof}[Proof of \Cref{lem:H{4,i,2}-H{4,i,3}}]
    This proof is same as the proof of proof of \Cref{lem:H{3,i,2}-H{3,i,3}} except that $\RZFdef[\rbtRcvr]{\klen+\ulen+i},\SOFdef[\rbtRcvr]{\klen+\ulen+i}$ both take a $(\klen+\ulen+i)$-bit string an input $\rbtkw\!\parallel\!\rbtSndr\!\parallel\!\msgi[\rbtRcvr]{i}$.
\end{proof}

\begin{proof}[Proof of \Cref{lem:H{4,i,3}-H{4,i,4}}]
        This proof is same as the proof of proof of \Cref{lem:H{3,i,3}-H{3,i,4}} except that $\RZFdef[\rbtRcvr]{\klen+\ulen+i+1}$ takes a $(\klen+\ulen+i+1)$-bit string an input $\rbtkw\!\parallel\!\rbtSndr\!\parallel\!\msgi[\rbtRcvr]{i+1}$ and $\ROFdef[\rbtRcvr]{\klen+\ulen+i}$ takes a $(\klen+\ulen+i)$-bit string an input $\rbtkw\!\parallel\!\rbtSndr\!\parallel\!\msgi[\rbtRcvr]{i}$.
\end{proof}

\begin{proof}[Proof of \Cref{lem:H{4,i,4}-H{4,i,5}}]
    This proof is same as the proof of proof of \Cref{lem:H{3,i,4}-H{3,i,5}} except that $\RZFdef[\rbtRcvr]{\klen+\ulen+i+1},\SOFdef[\rbtRcvr]{\klen+\ulen+i+1}$ both take a $(\klen+\ulen+i+1)$-bit string an input $\rbtkw\!\parallel\!\rbtSndr\!\parallel\!\msgi[\rbtRcvr]{i+1}$.
\end{proof}

\begin{proof}[Proof of \Cref{lem:H{4,L,1}-H{5}}]
    Firstly, in $\hame{4,L,5}$, $\msgi[\rbtkw\!\parallel\!\rbtSndr\!\parallel\!\rbtRcvr]{L}=\rbtkw\!\parallel\!\rbtSndr\!\parallel\!\rbtRcvr$ for $(\rbtkw,\rbtSndr,\rbtRcvr)\in\{0,1\}^L$. 
    Since $\QTg\cap \QChl=\emptyset$, $(\rbtSndr,\rbtkw,\rbtRcvr)\in\QTg$ are not repeated anywhere else. 
    Moreover, as per our restriction, $\OTag$ queries are distinct.
    Therefore, uniformly random $\RF[{\rbtSndr,\rbtRcvr}]{}{\rbtkw\!\parallel\!\rbtSndr\!\parallel\!\rbtRcvr}\in\Zp^{\nRAd\times \nRCB}$ is never repeated.
    This makes $\tg[1]$ for all $(\rbtSndr,\rbtkw,\rbtRcvr)\in\QTg$ uniformly random elements. 
    Thus, $\adv{}{\hame{4,L,5}}=\adv{}{\hame{5}}$.
\end{proof}

\begin{proof}[Proof of \Cref{lem:H{5}-H{6}}]
    We devise the reduction $\AB$ next in \Cref{fig:AB-H{3}-H{4}}. 
    $\AB$ receives $\prll|\QChl|\mhyf$fold $\matDH[\nRB,\nCB]$ instance $(\On{\M},\On{\f[1]},\ldots,\On{\f[\prll|\QChl|]})$. 

    % Observe that, $\{\sky_j\}_{j\in\QSk}$ stayed the same even after entropy injection via $\XZF[j]{0}{\epsilon},\XOF[j]{0}{\epsilon},$ $\ZZF[j]{0}{\epsilon},\ZOF[j]{0}{\epsilon}$. 
    $\OTag$ responses in both $\hame{5}$ and $\hame{6}$ are sampled uniformly at random. 
    We analyze $\OChal$ responses for the two games next. 
    First we assume that $\Mu\in\Zp^{\nRAd\times \nRAd}$ is invertible which happens with overwhelming probability ${\color{\confcolor}1-\frac{1}{p-1}}$. Observe that $\F[c]=\iCol{\Mu\K[c]}{\Md\K[c]+\R[c]}$ where $\R[c]$ is either $\Zro_{\nRCB\times\prll}$ or uniformly random in $\Zp^{\nRCB\times \prll}$. 
    We assume that $\Fu[c]$ is invertible which happens with overwhelming probability ${\color{\confcolor}1-\frac{1}{p-1}}$. 
    Thus, $\chl[\ij,0]=\Fu[c]\chl[\ij,0]'$ are uniformly random quantities as both $\Fu[c]$ and $\chl[\ij,0]'$ are uniformly random. Now we analyze $\chl[\ij,1]$ for all $\ij\in[\ell]$ for all the challenge queries.
    % \begin{enumerate}
    %     % \item For $\kw_{i+1}=0$, the $\OChal$ responses are same in both $\hame{2,i,3}$ and $\hame{2,i,4}$.
    %     \item For $\kw_{i+1}=1$, 

    {\scriptsize %\color{red}
            \begin{equation*}\hspace{-0.6cm}
                \begin{aligned}
                    \chl[\ij,1]
                    &={(\Sum{\ii}{[\klen]}\hNt[\ii,{\msgii[\kw]{\ii}}]+\hXt[\Sndr]+\hZt[\Rcvr_{\perm(\ij)}])\chl[\ij,0]}+{\Bp\RF[{\Sndr,\Rcvr_{\perm(\ij)}}]{{\klen+2\ulen}}{\kw\!\parallel\!\Sndr\!\parallel\!\Rcvr_{\perm(\ij)}}^{\top}\chl[\ij,0]}\\
                    &\qquad +\Bp\Fd[c]\chl[\ij,0]'\\
                    &={(\Sum{\ii}{[\klen]}\hNt[\ii,{\msgii[\kw]{\ii}}]+\hXt[\Sndr]+\hZt[\Rcvr_{\perm(\ij)}])\chl[\ij,0]}+{\Bp\RF[{\Sndr,\Rcvr_{\perm(\ij)}}]{{\klen+2\ulen}}{\kw\!\parallel\!\Sndr\!\parallel\!\Rcvr_{\perm(\ij)}}^{\top}\chl[\ij,0]}\\
                    &\qquad +{\Bp\Md\Mu^{-1}\Fu[c]\chl[\ij,0]'+\Bp\R[c]\chl[\ij,0]'}\\
                    &={(\Sum{\ii}{[\klen]\setminus\{1\}}\hNt[\ii,{\msgii[\kw]{\ii}}]+\hXt[\Sndr]+\hZt[\Rcvr_{\perm(\ij)}])\chl[\ij,0]}+(\hNt[1,{\msgii[\kw]{1}}]+\Bp\Md\Mu^{-1})\chl[\ij,0]\\
                    &\qquad +{\Bp\RF[{\Sndr,\Rcvr_{\perm(\ij)}}]{{\klen+2\ulen}}{\kw\!\parallel\!\Sndr\!\parallel\!\Rcvr_{\perm(\ij)}}^{\top}\chl[\ij,0]+\Bp\R[c]\Fu[c]^{-1}\chl[\ij,0]}\\
                    &={(\Sum{\ii}{[\klen]}\Nt[\ii,{\msgii[\kw]{\ii}}]+\Xt[\Sndr]+\Zt[\Rcvr_{\perm(\ij)}])\chl[\ij,0]}
                    +\Bp\RF[{\Sndr,\Rcvr_{\perm(\ij)}}]{{\klen+2\ulen}}{\kw\!\parallel\!\Sndr\!\parallel\!\Rcvr_{\perm(\ij)}}^{\top}\chl[\ij,0]\\
                    &\qquad +\Bp\R[c]\Fu[c]^{-1}\chl[\ij,0]\\
                \end{aligned}
            \end{equation*}
        }

        If $\R[c]=\Zro_{\nRCB\times \prll}$, $\AB$ simulates $\hame{5}$ and if $\R[c]$ is uniformly random, $\AB$ simulates $\hame{6}$. Thus, $|\adv{\AA}{\hame{5}}-\adv{\AA}{\hame{6}}|\leq \Adv{\AB,\Go}{\prll|\QChl|\mhyf\matDH[\nRB,\nCB]}+\frac{|\QChl|+2}{p-1}$.

    \begin{figure*}[ht]
        \scriptsize
        %\centering
        \hspace{-2.5cm}
        \fbox{
            % \parbox[c][143mm][c]{1.13\linewidth}{
            % \centering\hspace{-4.5cm}{$\hame{1}$},\lfbox[rounded]{$\hame{2,i,1}$},\lfbox[dotted]{$\hame{2,i,2}$,\lfbox[border-style=dashed]{$\hame{2,i,3}$},\lfbox[background-color=lightgray!60!white,border-color=lightgray!60!white]{$\hame{2,i,4}$},\lfbox{$\hame{2,i,5}$}}, \lfbox[background-color=lightgray!60!white,border-color=black]{$\hame{3}$}, \lfbox[border-style=double,border-width=2pt]{$\hame{4}$}\\ \ \\
            % \hspace{-1cm}
            \begin{minipage}{0.55\textwidth}%{0.47\textwidth}
                \underline{$\Init(1^\secpp)$}
                \begin{algorithmic}[1]
                    \State $\abG\sample \ABSGen(1^\secpp)$ %=(p,\Go,\Gt,\GT,\go,\gt,e)
                    \State From $\nCB|\QChl|\mhyf$fold $\matDH[\nRB,\nCB]$ instance, 
                    \item[]\hspace{0.2cm}define $\F[c]=(\f[(c-1)\nRAd+1]||\ldots||\f[c\nRAd])\in\Zp^{\nRB\times \nCB}$
                    \item[]\hspace{0.2cm}for $c\in[|\QChl|]$ 
                    \State $\B\sample\Dk[\nRB,\nCB]$, $\Bp\sample\Zp^{\nRB\times\nRCB}$ s.t. $\Bpt\B=\Zro$
                    \State {\color{\confcolorr}For $(\ii,b)\in[\klen]\times\{0,1\}: \N[\ii,b]=\hN[\ii,b]\sample \Zp^{(\nRAd)\times (\nRB)}$}
                    \item[] Implicit: $\N[1,b]=\hN[1,b]+\Mu^{-\top}\Mdt\Bzst$
                    \State $\ppC=(\On{\B},\Tw{\B},
                    {\color{\confcolorrr}\left\{\Tw{\N[\ii,b]\B}\right\}_{\!\substack{\!\ii\in [\klen] \\ \!b\in\{0,1\}}}})$.
                    \item[] 
                    % \item[]
                \end{algorithmic}
                \underline{$\OKey(j)$}
                \begin{algorithmic}[1]
                    \State For $(\ii,b)\in[\klen]\times\{0,1\}$: $J[j,\ii,b]\sample \Zp^{(\nRAd)\times (\nRB)}$
                    \State ${\color{\confcolorr}
                    \hXt[j] = \SumiDT{1}{\ulen}\Jt[j,i,{\msgii[j]{i}}], 
                    \hZt[j] = \SumiDT{\ulen+1}{2\ulen}\Jt[j,i,{\msgii[j]{i}}]}$   
                    % \item[]\hspace{0.35cm}$\hX[j,\ii,b],\hZ[j,\ii,b]\sample \Zp^{(\nRAd)\times (\nRB)}$
                    \State $\X[j]=\hX[j],\quad\Z[j]=\hZ[j]$
                    % \item[]%\lfbox[rounded]{\lfbox[dotted]{\lfbox[background-color=lightgray!60!white,border-color=black]{
                    %     \parbox{0.75\textwidth}{
                    %         Set $\hX[j,1,0]=\hX[j,1,0]+\XRF[j]{0}{\epsilon}\Bpt$\\
                    %         Set $\hX[j,1,1]=\hX[j,1,1]+\XRF[j]{0}{\epsilon}\Bpt$\\
                    %         Set $\hZ[j,1,0]=\hZ[j,1,0]+\ZRF[j]{0}{\epsilon}\Bpt$\\
                    %         Set $\hZ[j,1,1]=\hZ[j,1,1]+\ZRF[j]{0}{\epsilon}\Bpt$\\
                    %     }
                    % \item[]//Implicit: $\X[j,i+1,b]=\hX[j,i+1,b]+\Mu^{-\top}\Mdt\Bzst$
                    % \item[]//Implicit: $\Z[j,i+1,b]=\hZ[j,i+1,b]+\Mu^{-\top}\Mdt\Bzst$
                    % \item[]//for all $j\in \QChl$ 
                    
                    \State $\sky_j=(\left(\Tw{\X[j,\ii,b]\B},\Tw{\Z[j,\ii,b]\B}\right)_{\substack{{\ii\in[L]}\\{b\in\{0,1\}}}})$
                    \State $\QSk=\QSk\sqcup \{(j,\sky_j)\}$
                    \State Return $(j,{\color{\confcolor}\Tw{\X[j]\B}})$
                    % \item[]
                    % \item[]
                \end{algorithmic}
            \underline{$\OCor(j)$}
                \begin{algorithmic}[1]
                    \State $\QCr=\QCr\sqcup \{(j,\sky_j)\}$ from $\QSk$
                    \State Return $\sky_j$
                    % \item[]
                \end{algorithmic}    
                \end{minipage}
            \vline \hspace{1pt}
            \begin{minipage}{0.75\textwidth}
                \underline{$\OTag({\Sndrr},\kwd,{\Rcvrr})$}   
                \begin{algorithmic}[1]
                    \State $\QTg=\QTg\sqcup \{({\Sndrr},\kwd,{\Rcvrr})\}$
                    %\State $\r\sample\Zp^{\nCB}$
                    \State $\Tg=(\Tw{\tg[{0}]},\Tw{\tg[{1}]})$ where {$\tg[{0}]\sample \Zp^{\nRB}$, $\tg[{1}]\sample \Zp^{\nRAd}$.}   
                    \State Return $\Tg$ 
                    \item[]
                \end{algorithmic}
                \underline{$\OChal({\Sndr}, \kw, {\RSet})$} \hspace{2cm}  
                \begin{algorithmic}[1]
                    \State Let $\RSet=\{\Rcvr_1,\ldots,\Rcvr_{\ell}\}$
                    % \State For all $\Usr\in \{\Sndr\}\sqcup \RSet:$
                    % \item[]\hspace{0.35cm}If $\Usr\notin\QSk:$ 
                    % \item[]\hspace{0.70cm}$\hX[\Usr,\ii,b],\hZ[\Usr,\ii,b]\sample \Zp^{(\nRAd)\times (\nRB)}\ \forall (\ii,b)\in[L]\times\{0,1\}$ 
                    % \item[]\hspace{0.70cm}$\X[j,1,b]=\hX[j,1,b]+\Mu^{-\top}\Mdt$ and $\Z[j,1,b]=\hZ[j,1,b]+\Mu^{-\top}\Mdt$ 
                    \State $\perm\sample \Perm{[\ell]}$ 
                    \State $\QChl=\QChl\cup \{({\Sndr},\kw,{\Rcvr_\ij})_{\ij\in[\ell]}\}$
                    \State For all $(\Sndr,\Rcvr_{\perm(\ij)})\in \{\Sndr\}\times \RSet:$ 
                    \item[]\hspace{0.35cm}Let $(\Sndr,\kw,\Rcvr_{\perm(\ij)})$ be the $c^{th}$ entry of $\QChl$
                    % \State $\Chl=((\On{\chl[\ij,0]},\On{\chl[\ij,1]})_{\ij\in[\ell]})$ s.t.
                    \item[]\hspace{0.35cm}$\chl[\ij,0]=\Fu[c]\chl[\ij,0]'$ for $\chl[\ij,0]'\sample\Zp^{\nRAd}$ 
                    %\item[]\hspace{0.3cm}$={\Sum{\ii}{[L]}(\hXt[\Sndr,\ii,\kw_{\ii}]+\hZt[\Rcvr_{\perm(\ij)},\ii,\kw_{\ii}])\chl[\ij,0]}$
                    \item[]\hspace{0.35cm}$\chl[\ij,1]={(\Sum{\ii}{[\klen]}\hNt[\ii,{\msgii[\kw]{\ii}}]+\hXt[\Sndr]+\hZt[\Rcvr_{\perm(\ij)}])\chl[\ij,0]}$
                    % \item[]\hspace{0.35cm}
                    % \hfill
                    ${+\Bp\RF[{\Sndr,\Rcvr_{\perm(\ij)}}]{\klen+2\ulen}{\kw\!\parallel\!\Sndr\!\parallel\!\Rcvr_{\perm(\ij)}}^{\top}\chl[\ij,0]}+\Bp\Fd[c]\chl[\ij,0]'$
                    % {$+((\Bzs\ZF[{\Sndr,\Rcvr}]{i}{\msgi[\kw]{i}}^{\top}+\Bos\OF[{\Sndr,\Rcvr}]{i}{\msgi[\kw]{i}}^{\top}))\chl[\ij,0]$}  % +(\Bzs\ZZF[\Rcvr_{\perm(\ij)}]{i}{\msgi[\kw]{i}}^{\top}+\Bos\ZOF[\Rcvr_{\perm(\ij)}]{i}{\msgi[\kw]{i}}^{\top})
                    % \item[] %for // {\color{red}$c^{th}$ new $(\Sndr,\msgi[\kw]{i},\Rcvr)$ pair where $\Rcvr\in \RSet$}
                    % \item[]$\chl[\ij,1]={\Sum{\ii}{[L]}(\hXt[\Sndr,\ii,\kw_{\ii}]+\hZt[\Rcvr_{\perm(\ij)},\ii,\kw_{\ii}])\chl[\ij,0]}$
                    % \item[]\hspace{0.35cm}{$+((\Bzs\XZF[\Sndr]{i+1}{\msgi[\kw]{i+1}}^{\top}+\Bos\XOF[\Sndr]{i}{\msgi[\kw]{i}}^{\top})+(\Bzs\ZZF[\Rcvr_
                    % {\perm(\ij)}]{i}{\msgi[\kw]{i}}^{\top}+\Bos\ZOF[\Rcvr_{\perm(\ij)}]{i}{\msgi[\kw]{i}}^{\top}))\chl[\ij,0]$}  
                    % \State If $\kw_{i+1}=1$: $\chl[\ij,1]=\chl[\ij,1]+{2\Bzs\Fd[c]\chl[\ij,0]}$
                    \State Return $\Chl=((\On{\chl[\ij,0]},\On{\chl[\ij,1]})_{\ij\in[\ell]})$
                    \item[]
                \end{algorithmic}        
                \underline{$\Finalize(\bee\in\{0,1\})$}
                \begin{algorithmic}[1]
                    \State Return $\bee\wedge\ (\QTg\cap\QChl=\emptyset){\color{\confcolor}\wedge\ (\forall  (u,\dontcare)\in\QCr, (u,\dontcare,\dontcare),(\dontcare,\dontcare,u)\notin {\color{\confcolor}\QTg\cup}\QChl)}$
                    % \item[]
                \end{algorithmic}      
            \end{minipage}
        % }
        }
        \caption{The reduction $\AB$ in the proof of \Cref{lem:H{5}-H{6}}.}
        \label{fig:AB-H{3}-H{4}}
    \end{figure*}

    % {\color{red}NOT SUFFICIENT}
    % Due to the natural restriction, $\QTg\cap \QChl=\emptyset$.
    % Thus, for all $(\Sndrr,\kwd,\Rcvrr)\in \QTg$, for all $(\Sndr,\kw,\Rcvr)\in \QChl$, $(\Sndrr,\kwd,\Rcvrr)$ \emph{does not match} $(\Sndr,\kw,\Rcvr)$. 
    % Therefore, $\RF{\Sndrr}{\kwd}{\Rcvrr}$ will be independent of $\RF{\Sndr}{\kw}{\Rcvr}$ for all $(\Sndrr,\kwd,\Rcvrr)\in \QTg$, for all $(\Sndr,\kw,\Rcvr)\in \QChl$.
    % Being defined as sum of two random functions, output of $\mathsf{RF}:\USet\times\KWd\times\USet\rightarrow \Zp^{\nRAd\times \nRCB}$ is already uniformly distributed.
    % Due to the restrictions on non-duplicated challenges, we see that both $\RF{\Sndrr}{\kwd}{\Rcvrr}$ and $\RF{\Sndr}{\kw}{\Rcvr}$ are used only once for all $(\Sndrr,\kwd,\Rcvrr)\in \QTg$, for all $(\Sndr,\kw,\Rcvr)\in \QChl$.
    % Therefore, we replace $\RF{\Sndrr}{\kwd}{\Rcvrr}$ and $\RF{\Sndr}{\kw}{\Rcvr}$ by random vectors of appropriate dimensions for all $(\Sndrr,\kwd,\Rcvrr)\in \QTg$, for all $(\Sndr,\kw,\Rcvr)\in \QChl$.
\end{proof}

\section{Experimental Results}
\label{sec:Comparison}
%!TEX spellcheck = en_US
%!TEX root = ../main.tex

% {\color{red}EXPERIMENTAL RESULTS}
We implement $\baeks$ (\Cref{fig:Construction-BAEKS}) in python $3.6.9$ using the Charm $0.50$ framework \cite{JCEng:AGMPRGR13}. 
We use MNT224 curve for pairings as it is a Type-III curve in PBC and existing works \cite{CCS:AgrCha17} also use this for their implementation. Interestingly, \cite{CCS:AgrCha17} introduced the $\bmatDH$ assumption and we prove our construction secure under $\lmatDH$ assumption which is hard under $\bmatDH$ assumption. {\color{\confcolor}We have run our experiments on Virtualbox VM running Ubuntu 18.04 where the VM is allocated one processor and 4GB RAM (fixed by the Virtualbox Setting).
}
% \Cref{fig:Group-Operations} lists the average time taken by different operations on MNT224. 

% \begin{figure}
%     \begin{tabular}{|c|c|c|}
%         \hline
%         Groups & Multiplication & Exponentiation \\
%         \hline
%     \end{tabular}
%     \caption{Average time (in miliseconds) taken by various operations on MNT224 curve.}
%     \label{fig:Group-Operations}
% \end{figure}

\Cref{fig:Count-Operations} then lists the number of various operations in $\Go$, $\Gt$ and $\GT$ for $\Setup$, $\Kgen$, $\SEnc$, $\Tgen$ and $\Test$.
Note that, $\Setup$ is run only once by a trusted third party at the starting of the system, 
$\Kgen$ is run by each user independently only once when they join. 
$\SEnc$ and $\Tgen$ are respectively ciphertext generation and trapdoor generation algorithms which will be run by the users over and over.
$\Test$ is another algorithm, that is run by the cloud server many times. 
It is only natural to aim for making $\SEnc$, $\Tgen$ and $\Test$ simple and efficient.
Observe that, all the computations does not involve $\ulen$ and only $\Setup$ and $\SEnc$ computation varies with $\klen=\poly$ linearly. 
As $\emm$ denotes security level (predecided), all other algorithms are extremely efficient.
In our $\baeks$, $\Test$ executes ${2\ell}$-many bilinear pairings to test a \emph{match} between a ciphertext and a trapdoor, in the worst case.
Setting $\ell=1$,we get a $\paeks$, where $\Test$ executes only $2$ bilinear pairings to test a \emph{match} between a ciphertext and a trapdoor.

% |A|= (m x n), |B| = (n x p)
% AB computation: #Exp  = mnp
% AB computation: #Mult = m(n-1)p

% g_1^A (2k x k) x 1
% g_1^B (3k x k) x 1
% K^T A (3k x 2k) x (2k x k) x 2\alpha 
% g_2^A (2k x k) x 1
% g_2^B (3k x k) x 1
% K B   (2k x 3k) x (3k x k) x 2\alpha 
% U^T A (3k x 2k) x (2k x k)
% V B   (2k x 3k) x (3k x k)
% U B   (2k x 3k) x (3k x k)
% V^T A (3k x 2k) x (2k x k)
% K^T A + U^T A + V^T A   : 3k x k x (\alpha+2) mult
%(K^T A + U^T A + V^T A) s: (3k x k) x (k x 1)
% A s                     : (2k x k) x (k x 1)
% K B + U B + V B         : 2k x k x (\alpha+2) mult
%(K B + U B + V B) r      : (2k x k) x (k x 1)
% B r                     : (3k x k) x (k x 1)
% c_0^T k_1               : (1 x 2k) x (2k x 1)
% c_1^T k_2               : (1 x 3k) x (3k x 1)

{\setlength{\tabcolsep}{1pt}
\begin{figure*}
    \centering
    \scriptsize
    \begin{tabular}{|c|c|c|c|c|c|c|}
        \hline
                & \multicolumn{2}{|c|}{$\Go$} & \multicolumn{2}{|c|}{$\Gt$} & \multicolumn{2}{|c|}{$\GT$} \\ \cline{2-7}
                & Mul & Exp & Mul & Exp & Mul & Pairing \\ \hline  
        $\Setup$& $6\klen\emm^2(2\emm-1)$ & $(5+12\klen\emm)\emm^2$ &$4\klen\emm^2(3\emm-1)$ & $(5+12\klen\emm)\emm^2$ & $\mhyf$ & $\mhyf$ \\ \hline  
        $\Kgen$ & $6\emm^2(2\emm-1)$ & $12\emm^3$ & $4\emm^2(3\emm-1)$ & $12\emm^3$ & $\mhyf$ & $\mhyf$ \\ \hline  
        $\SEnc$ & $(3\klen\emm+11\emm-5)\emm\ell$ & $5\emm^2\ell$ & $\mhyf$ & $\mhyf$ & $\mhyf$ & $\mhyf$ \\ \hline  
        $\Tgen$ & $\mhyf$ & $\mhyf$ & $(2\klen\emm+9\emm-5)\emm$ & $5\emm^2$ & $\mhyf$ & $\mhyf$ \\ \hline  
        $\Test$ &  $\mhyf$ & $\mhyf$ & $\mhyf$ & $\mhyf$ & $(5\emm-2)\ell$ & $5\emm\ell$ \\ \hline  
    \end{tabular}
    \caption{Number of various operations in $\Go$, $\Gt$ and $\GT$ for $\Setup$, $\Kgen$, $\SEnc$, $\Tgen$ and $\Test$ of $\baeks$. Here, $\emm$ denotes the security level ($\emm>1$), $|\KWd|=2^{\klen},|\USet|=2^{\ulen}$ and receiver set size is upper-bounded by $\ell$. \\To see computation cost of our PAEKS, one can set $\ell=1$.}
    \label{fig:Count-Operations}
\end{figure*}
}

Note that, $\lmatDH[k]$ for security level $k=1$ is an easy problem, and therefore we consider $k=2$ in all our experiments. 
In our experiments, we consider $\Users$ to be a list of numbers where we denote an arbitrary user as $i\in\Users=\{1,2,\ldots,100\}$. 
Note that, each user runs $\Kgen$ on their own, independently than the rest of the users. 
Therefore, in the experiment it is reasonable to fix a list of users $\Users$.  
We also consider different values of $\ell$ i.e. different sizes of receiver sets such as $\{1, 5, 10, 50, 100\}$ %,200,500,1000,10000\}$.
Note that, for $\ell=1$, BAEKS becomes PAEKS. 
We also have considered different $\ulen$ from $\{128,512\}$ and different $\klen$ from $\{128,512\}$.
Therefore, as a subproduct of this work, we also achieve a tightly secure PAEKS in the multi-users, multi-challenge settings with adaptive corruptions.

We only present a small description of our performance in \Cref{fig:Experiment-Time}.
The PAEKS is already extremely comparably efficient with existing works and provide tight security. 
Our simple proof-of-concept BAEKS code running on virtualbox shows practicality of our protocol. Making tournament-wise parallel multiplication of group elements running on multiple cores will certainly make our protocol much more efficient.

% We also consider keyword lengths $L\in\{128,256\}$ in our experiments.
% Thus, \Cref{fig:Setup}, \Cref{fig:Kgen}, \Cref{fig:SrchEnc}, \Cref{fig:Tgen} and \Cref{fig:Test}.

{
\setlength{\tabcolsep}{1.0pt}
\begin{figure*}
    \hspace{-0.5cm}
    % \centering
    \scriptsize
    \begin{tabular}{|c|c|c|c|c|c|c|c|c|c|c|c|c|c|c|c|c|c|c|c|c|c|c|c|c|c|c|c|}
        \hline 
            &        &\multicolumn{10}{|c|}{$\klen=128$}   & \multicolumn{10}{|c|}{$\klen=512$}\\ \cline{3-22}
            &        &\multicolumn{ 5}{|c|}{$\ulen=128$}   & \multicolumn{5}{|c|}{$\ulen=512$}&\multicolumn{ 5}{|c|}{$\ulen=128$}   & \multicolumn{5}{|c|}{$\ulen=512$}\\ \cline{3-22}   
            &        & $S$ & $K$   & $S$ & $T$ & $T$ & $S$ & $K$ & $S$ & $T$ & $T$ & $S$ & $K$ & $S$ & $T$ & $T$ & $S$ & $K$ & $S$ & $T$ & $T$\\            
            &        & $E$ & $E$   & $R$ & $R$ & $E$ & $E$ & $E$   & $R$ & $R$ & $E$ & $E$ & $E$   & $R$ & $R$ & $E$ & $E$ & $E$   & $R$ & $R$ & $E$ \\            
            &        & $T$ & $Y$   & $C$ & $A$ & $S$ & $T$ & $Y$   & $C$ & $A$ & $S$ & $T$ & $Y$   & $C$ & $A$ & $S$ & $T$ & $Y$   & $C$ & $A$ & $S$ \\            
            &        & $U$ & $G$   & $H$ & $P$ & $T$ & $U$ & $G$   & $H$ & $P$ & $T$ & $U$ & $G$   & $H$ & $P$ & $T$ & $U$ & $G$   & $H$ & $P$ & $T$ \\            
            &        & $P$ & $E$   & $E$ & $G$ & & $P$ & $E$   & $E$ & $G$ & & $P$ & $E$   & $E$ & $G$ &  & $P$ & $E$   & $E$ & $G$ & \\            
            &        & & $N$   & $N$ & $E$ & &      & $N$   & $N$ & $E$ &  &      & $N$   & $N$ & $E$ &  &      & $N$   & $N$ & $E$ & \\             
            &        & &  & $C$ & $N$ & & &   & $C$ & $N$ & & &   & $C$ & $N$ & & &   & $C$ & $N$ & \\ \hline
            \multirow{5}{*}{$\ell$} & $1$    & $80.38$  & $0.77$      & $0.14$ & $0.02$ & $0.05$ &$82.41$& $1.11$  & $0.15$ & $0.026$  & $0.05$& $307.34$  & $0.72$      & $0.21$ & $0.046$ & $0.05$ &$82.41$& $1.11$  & $0.15$ & $0.026$  & $0.05$\\
            & $5$& 79.84 & 0.76 & 0.14 & 0.08 & 0.25 & 81.67 & 1.16 & 0.16 & 0.09 & 0.26 & 308.58 & 0.71 & 0.20 & 0.11 & 0.23 & 308.67 & 1.07 & 0.21 & 0.10 & 0.24\\
            & $10$& 79.56 & 0.76 & 0.14 & 0.16 & 0.48 & 87.84 & 1.13 & 0.14 & 0.18 & 0.36 & 289.22 & 0.75 & 0.21 & 0.18 & 0.33 & 307.25 & 1.05 & 0.20 & 0.19 & 0.46\\
            & $50$& 80.58 & 0.75 & 0.14 & 0.82 & 2.44 & 81.48 & 1.14 & 0.14 & 0.81 & 2.48 & 318.98 & 0.73 & 0.20 & 0.81 & 2.38 & 307.40 & 1.07 & 0.21 & 0.82 & 2.34 \\
            & $100$   & $47.89$  & $0.43$      & $0.08$  & $1.08$ & $2.46$& $47.75$  & $0.63$      & $0.08$  & $1.09$  & $2.45$ & $237.31$  & $0.43$      & $0.12$  & $1.09$ & $2.47$ & $207.47$  & $0.63$      & $0.12$  & $1.1$  & $2.43$ \\
            \hline
    \end{tabular}
    \caption{Average time (in seconds) taken by various algorithms of $\baeks$ (\Cref{sec:BaEKS-Cons}) for receiver set of size $\ell$. For $\ell=1$, we present performance of our PAEKS.}
    \label{fig:Experiment-Time}
\end{figure*}
}

% In \Cref{fig:Experiment-Time}, we see that $\Kgen$ takes a significant amount of time. We here note that, each user will run $\Kgen$ only one, at the time of their joining.
% Functions such as $\SEnc$, $\Tgen$ and $\Test$ are quite efficient considering the receiver set size.
% We further note that, this implementation is only a proof of concept and one can further optimize the protocol by coding it in C/C++.
% We believe our construction can be deployed in the practical setting.

\section{Conclusion}\label{sec:Conclude}
This paper revamps existing BAEKS security models of \cite{ACISP:Mukherjee23,IET:Emura23} to consider ciphertext and trapdoor security from the lens of sender, receiver, and keyword privacy.
We further propose a consistency definition for BAEKS.
Following this, we propose a new \emph{statistically consistent} construction of BAEKS which we prove to be secure in the standard model. 
More precisely, our novel BAEKS construction achieves adaptive $\fullcpa$ security (in terms of both ciphertext and trapdoor) under the standard assumption (lateral) Matrix Diffie-Hellman $(\lmatDH)$.
% Interestingly, our construction still achieves asymptotic efficiency similar to that of \cite{ACISP:LHYSTH21}.
In terms of the ciphertext and key size, our scheme is a little less efficient than \cite{ACISP:Mukherjee23} but we achieve tight $\fullcpa$ security and can withstand adversary corrupting a few users. 
For future work, one might aim for a new construction that achieves more robust security in the presence of malicious adversaries making their own key pairs.

\newpage

\bibliographystyle{./styl_files/splncs04}
\bibliography{./cryptobib/abbrev3,./cryptobib/crypto,./references}

% \newpage
% \appendices
\appendix

% \section{Comments on Consistency of Emura's PAEKS}\label{sec:Emura-Correctness}
% \input{tex_files/0D_Emura}

\section{Comments on Security Proof of Ling et al.'s PEKS}\label{sec:Ling-Security}
%!TEX spellcheck = en_US
%!TEX root = ../main.tex

Ling \etal \cite{PKC:LZCHQ24} proposed two multi-user PEKS constructions. We informally recall important parts of their constructions next.

\begin{enumerate}
    \item Construction-I: 
    \begin{description}
        \item[Params.]\ $\pp=(\mathcal{G},g_1\in G_N,h_1\in H_N)$ where $\mathbb{G}= (N = p_1p_2, G_N , H_N , G_T , e)$ where $order(g_1)=order(h_1)=p_1$
        \item[Key-Pair.]\ $\pk_j=(g_1^x,g_1^y,e(g_1,h_1)^{\alpha})$, $\sk_j=(x,y,\alpha)\sample \mathbb{Z}_N^3$.
        \item[Ciphertext.]\ \ \ $U=g_1^r$, $V=g_1^{(x\kw+y)r}$, $W=\delta\cdot e(g_1,h_1)^{\alpha r}$, $\delta$ for $r\sample\mathbb{Z}_N$, $\delta\sample G_T$.
        \item[Trapdoor.]\ \ $A=h_1^s$, $B=h_1^{\alpha+s(x\kwd+y)}$ for $s\sample \mathbb{Z}_N$. 
        \item[Test.]\ $\frac{e(U,B)}{e(V,A)}\cdot \delta \iseq W$.
    \end{description}
    \item Construction-II: 
    \begin{description}
        \item[Params.]\ $\pp=(\mathcal{G},g_1\in G_N,h_1\in H_N)$ where $\mathbb{G}= (N = p_1p_2, G_N , H_N , G_T , e)$ where $order(g_1)=order(h_1)=p_1$
        \item[Key-Pair.]\ $\pk_j=(g_1^x,e(g_1,h_1)^{\alpha})$, $\sk_j=(x,\alpha)\sample \mathbb{Z}_N^3$.
        \item[Ciphertext.]\ \ \ $V=g_1^{(x+\kw)r}$, $W=\delta\cdot e(g_1,h_1)^{\alpha r}$, $\delta$ for $r\sample\mathbb{Z}_N$, $\delta\sample G_T$.
        \item[Trapdoor.]\ \ $A=h_1^{\frac{\alpha}{x+\kwd}}$. 
        \item[Test.]\ $e(V,A)\cdot \delta\iseq W$. 
    \end{description}
\end{enumerate}
Ling \etal \cite{PKC:LZCHQ24} claimed both these schemes are PEKS in the multi-user, multi-ciphertext setting and have provided two proofs.
Both the proofs define three types of ciphertexts -- $(i)$ type-0, $(ii)$ type-1 and $(iii)$ type-2. 
Both the proofs are three steps where,
\begin{itemize}
    \item $\game{0}\approx\game{1}$: First, they replace $g_1^r$ (in type-0 ciphertext) with $g_1^rg_2^{\hat{r}}$ (and make the ciphertexts type-1) by means of subgroup decision assumption.
    \item $\game{1}\approx\game{2}$: Then, they use $DDH$ on $G_2$ to inject randomness in $V$ (and make the ciphertexts type-2). Random self-reducibility of $DDH$ allows them to modify all challenge $V$. 
    \item $\game{2}$ is information-theoretic secure: Finally, they argue that type-2 ciphertext $V=g_1^{(x\kw+y)r}g_2^{\hat{z}}$ in the $G_2$ space has sufficient entropy to hide all the hashes involving different keywords $\kw$. 
\end{itemize}

\paragraph{Our Observations.}
Both these constructions use well-known structures Boneh-Boyen hash where the first Construction uses the affine hash, whereas the second construction uses exponent inversion-based hash \cite{C:BonBoy04}. To the best of our knowledge, none of these hashes are known to be tightly secure for multi-challenge setting. A critical analysis of the proof shows that, the final step where \cite{PKC:LZCHQ24} has claimed to have sufficient entropy does not really hold. To see this, let us assume a simpler situation where an adversary has received a challenge ciphertext on $\kw_{\bee}\in\{\kw_0,\kw_1\}$. Consider the adversary makes a trapdoor query on $\kwd\neq \kw$. Since the trapdoor responses are normal (without any $G_2$) component, pairing $A$ (of trapdoor) with $V$ removes entropy due to $G_2$. One can easily check the effect of the following two pairings $e(V,A)$ where $V$ is type-0 and $e(V,A)$ when $V$ is a type-2 ciphertext. Therefore, $\game{2}$ is not information theoretically secure. Therefore, the proofs are incomplete.

Moreover, \cite{PKC:LZCHQ24} did not argue consistency which is a definitive characteristic of PEKS.

\section{Differences in Security Notions} \label{sec:Different-Security}

\subsection{Unforgeability of Ciphertexts and Trapdoors.}\label{sec:Different-Security_Integrity}
% \fbox{
    \setlength{\tabcolsep}{10pt}
    \begin{figure*}
        \scriptsize
        \centering
        \begin{tabular}{|l|}
            \hline \\
    {$\Adv{\AA,\baeks}{\ctcma}=
\left| 
    \Pr 
    \left[
        \begin{aligned}
            &\Test(\Trp, \chCt)=1\\
            &\wedge\ (\chSndr,\dontcare,\dontcare)\notin \Qsk\\ 
            &\wedge\ (\pk_{\chSndr},\chkw,\pk_{\chRcvr})\notin \Qct\\    
        \end{aligned}
        \left|
%        \Test(\Trp, \chCt)=1:
        \begin{aligned}
            & \pp\leftarrow\Setup(1^\secpp,\ell), \Qsk,\Qct \leftarrow \emptyset\\
            & (\pk_{\chSndr},\chkw,\pk_{\chRSet}, \chCt)\leftarrow\AA^{\Oct(\cdot,\cdot,\cdot),\Otrp(\cdot,\cdot,\cdot),\Osk(\cdot),\Opk(\cdot)}(\pp) \\
            & \exists \chRcvr\in \chRSet \text{s.t.} \Trp\leftarrow\Tgen(\pk_{\chSndr},\chkw,\sk_{\chRcvr})\\
        \end{aligned}
        \right.
    \right]
  \right|$ 
    }\\ \ \\ 
    {
    $\Adv{\AA,\baeks}{\trapcma}=
    \left| 
        \Pr 
        \left[
            \begin{aligned}
                &\Test(\chTrp, \Ct)=1\\
                &\wedge\ (\chRcvrr,\dontcare,\dontcare)\notin \Qsk\\     
            \end{aligned}
            \left|
            \begin{aligned}
                & \pp\leftarrow\Setup(1^\secpp), \Qsk \leftarrow \emptyset\\
                & (\pk_{\chSndrr},\chkwd,\pk_{\chRcvrr}, \chTrp)\leftarrow\AA^{\Oct(\cdot,\cdot,\cdot),\Otrp(\cdot,\cdot,\cdot),\Osk(\cdot),\Opk(\cdot)}(\pp) \\
                & \Ct\leftarrow\SEnc(\sk_{\chSndrr},\chkwd,\pk_{\chRcvrr})\\
            \end{aligned}
            \right.
        \right]
      \right|
    $ 
    }  \\ \ \\
    \hline          
\end{tabular}
\caption{Unforgeability of Ciphertexts and Trapdoors}
    \label{fig:CMA}
\end{figure*}
% }

A BAEKS scheme $\baeks$ satisfies ciphertext unforgeability $\ctcma$ if for all $\ppt$ adversary $\AA$, the advantage of $\AA$ is $\Adv{\AA,\baeks}{\ctcma}=\neglgbl$ where $\AA$'s advantage is defined in \Cref{fig:CMA} provided $\AA$ is 
% {\small\[\Adv{\AA,\baeks}{\ctcma}=
% \left| 
%     \Pr 
%     \left[
%         \begin{aligned}
%             &\Test(\Trp, \chCt)=1\\
%             &\wedge\ (\chSndr,\dontcare,\dontcare)\notin \Qsk\\ 
%             &\wedge\ (\pk_{\chSndr},\chkw,\pk_{\chRcvr})\notin \Qct\\    
%         \end{aligned}
%         \left|
% %        \Test(\Trp, \chCt)=1:
%         \begin{aligned}
%             & \pp\leftarrow\Setup(1^\secpp,\ell), \Qsk,\Qct \leftarrow \emptyset\\
%             & (\pk_{\chSndr},\chkw,\pk_{\chRSet}, \chCt)\leftarrow\AA^{\Oct(\cdot,\cdot,\cdot),\Otrp(\cdot,\cdot,\cdot),\Osk(\cdot)}(\pp) \\
%             & \exists \chRcvr\in \chRSet \text{s.t.} \Trp\leftarrow\Tgen(\pk_{\chSndr},\chkw,\sk_{\chRcvr})\\
%         \end{aligned}
%         \right.
%     \right]
%   \right|
% \] 
% }
given access to following oracles: 
\begin{itemize}
    \item $\Oct$ is an oracle that on input $(\pk_{\Sndr},\kw,\pk_{\RSet})$ outputs $\Enc(\pp,\sk_{\Sndr},\kw,\pk_{\RSet})$. It maintains a list $\Qct$ of $(\pk_{\Sndr},\kw,\pk_{\Rcvr})$ for all $\Rcvr\in\RSet$ queried.

    \item $\Otrp$ is an oracle that on input $(\pk_{\Sndrr},\kwd,\pk_{\Rcvrr})$ outputs $\Tgen(\pp,\pk_{\Sndrr},\kwd,\sk_{\Rcvrr})$.
    
    \item $\Osk$ is an oracle that on input $j$ outputs $(\sk_j,\pk_j)\leftarrow\Kgen(\pp,j)$ after storing $(j,\sk_{j},\pk_{j})$ in $\Qsk$. %and $(j,\pk_{j})$ in $\Qpk$.

    \item $\Opk$ is an oracle that on input $j$ outputs $\pk_j$ where $(\sk_j,\pk_j)\leftarrow\Kgen(\pp,j)$.

%    \item $\Opk$ is an oracle that on input $j$ outputs $\pk_j$ where $(\sk_j,\pk_j)\leftarrow\Kgen(\pp,j)$ after storing $(j,\sk_{j},\pk_{j})$ in $\Qsk$ and $(j,\pk_{j})$ in $\Qpk$.
\end{itemize}

A BAEKS scheme $\baeks$ satisfies trapdoor unforgeability $\trapcma$ if for all $\ppt$ adversary $\AA$, the advantage of $\AA$ is $\Adv{\AA,\baeks}{\trapcma}=\neglgbl$ which is defined in \Cref{fig:CMA} provided $\AA$ is
% {\small\[\Adv{\AA,\baeks}{\trapcma}=
% \left| 
%     \Pr 
%     \left[
%         \begin{aligned}
%             &\Test(\chTrp, \Ct)=1\\
%             &\wedge\ (\chRcvrr,\dontcare,\dontcare)\notin \Qsk\\     
%         \end{aligned}
%         \left|
%         \begin{aligned}
%             & \pp\leftarrow\Setup(1^\secpp), \Qsk \leftarrow \emptyset\\
%             & (\pk_{\chSndrr},\chkwd,\pk_{\chRcvrr}, \chTrp)\leftarrow\AA^{\Oct(\cdot,\cdot,\cdot),\Otrp(\cdot,\cdot,\cdot),\Osk(\cdot)}(\pp) \\
%             & \Ct\leftarrow\SEnc(\sk_{\chSndrr},\chkwd,\pk_{\chRcvrr})\\
%         \end{aligned}
%         \right.
%     \right]
%   \right|
% \] 
% }
given access to following oracles: 
\begin{itemize}
    \item $\Oct$ is an oracle that on input $(\pk_{\Sndr},\kw,\pk_{\RSet})$ outputs $\Enc(\pp,\sk_{\Sndr},\kw,\pk_{\RSet})$. It maintains a list $\Qct$ of $(\pk_{\Sndr},\kw,\pk_{\Rcvr})$ for all $\Rcvr\in\RSet$ queried.

    \item $\Otrp$ is an oracle that on input $(\pk_{\Sndrr},\kwd,\pk_{\Rcvrr})$ outputs $\Tgen(\pp,\pk_{\Sndrr},\kwd,\sk_{\Rcvrr})$.
    
    \item $\Osk$ is an oracle that on input $j$ outputs $(\sk_j,\pk_j)\leftarrow\Kgen(\pp,j)$ after storing $(j,\sk_{j},\pk_{j})$ in $\Qsk$. %and $(j,\pk_{j})$ in $\Qpk$.

    \item $\Opk$ is an oracle that on input $j$ outputs $\pk_j$ where $(\sk_j,\pk_j)\leftarrow\Kgen(\pp,j)$.

    %\item $\Opk$ is an oracle that on input $j$ outputs $\pk_j$ where $(\sk_j,\pk_j)\leftarrow\Kgen(\pp,j)$ after storing $(j,\sk_{j},\pk_{j})$ in $\Qsk$ and $(j,\pk_{j})$ in $\Qpk$.
\end{itemize} 

Now we argue that if a BAEKS scheme is $\fullcpa$ secure, then it is both $\ctcma$ and $\trapcma$ secure. We argue this in two steps. First we show that, given a $\ppt$ adversary that breaks $\ctcma$ of a BAEKS construction, we can construct another $\ppt$ adversary that breaks $\fullcpa$ of the BAEKS construction. Consider $\AC$ is the $\fullcpa$ challenger and $\AA$ the $\ctcma$ adversary.

\begin{itemize}
    \item $\SetuP$: $\AC$ gives out $\pp$ that $\AB$ forwards to $\AA$.
    \item $\QPhase{}$: When $\AA$ makes oracle queries on $\Oct$ or $\Otrp$ or $\Osk$, $\AB$ forwards it to the corresponding oracles modeled by $\AC$. $\AB$ forwards the responses to $\AA$. Note that, $\AB$ forwards $\AA$'s query on $\Oct$ or $\Otrp$ by putting them on both the sides of the challenge tuple. $\AB$ maintains a list $\Qct$ which memorizes $(\pk_{\Sndr},\kw,\pk_{\Rcvr})$ for all $\Rcvr\in\RSet$ when $\AA$ makes $\Oct$ queries on $(\pk_{\Sndr},\kw,\pk_{\RSet})$.   %$\AA$ and $\AB$ respect the natural restrictions. 
    \item $\Forge$: At some point of time, $\AA$ produces a forgery $(\pk_{\chSndr},\chkw,\pk_{\chRSet}, \chCt)$ s.t. $\exists \chRcvr\in \chRSet$ for which $\Test(\Tgen(\pk_{\chSndr},\chkw,\sk_{\chRcvr}), \chCt)=1$ subject to the restriction $(\pk_{\chSndr},\chkw,\pk_{\chRcvr})$ is not in $\Qct$ and no secret key query has been made on $\chSndr$. 
    \begin{enumerate}
        \item $\AB$ makes a $\Otrp$ query on $((\pk_{\chSndr},\chkw,\pk_{\chRcvr}),(\pk_{\Sndrr},\chkw,\pk_{\chRcvr}))$ where $\Sndrr\neq \chSndr$ s.t. $(\pk_{\Sndrr},\chkw,\pk_{\chRcvr})$ is not in $\Qct$.
        \item $\AC$ responds with $\chTrp_{\bee}$ for its choice of $\bee\in\{L,R\}$ (symbolic of "Left" and "Right").
        \item $\AB$ runs $\Test(\chTrp,\chCt)$. If it is $1$, $\AB$ outputs its guess $\bee'=L$; otherwise outputs $\bee'=R$.
    \end{enumerate}
\end{itemize}
As $(\pk_{\chSndr},\chkw,\pk_{\chRcvr}), (\pk_{\Sndrr},\chkw,\pk_{\chRcvr})\notin \Qct$, $\AB$ is allowed to make the challenge $\Otrp$ query. If $\chTrp_{\bee}$ encoded $(\pk_{\chSndr},\chkw,\pk_{\chRcvr})$, the correctness of BAEKS ensures $\Test$ outputs $1$ and $\AB$ correctly guesses $L$. If $\chTrp_{\bee}$ encoded $(\pk_{\Sndrr},\chkw,\pk_{\chRcvr})$, the consistency of BAEKS will ensure $\Test$ outputs $0$ (and $\AB$ correctly guesses $R$) except with negligible probability.

It is easy to see that $\fullcpa$ security of a BAEKS scheme would also ensure $\trapcma$ security of the BAEKS scheme following a reduction very similar to the above.

% \section{Correctness and Consistency}\label{sec:Correct-Consistent}
% \input{tex_files/0A_Correct_Consistent}

% \section{Security Proof of Core-Lemma}\label{sec:Core-Security}
% % \input{tex_files/0C_Core_Lemma_Games}
% \input{tex_files/0C_Core_Lemma_Lemmas}

% \section{Security Proof of Full Security of BAEKS}\label{sec:Full-Proofs}
% \input{tex_files/0B_Full_Security_Games}
% \input{tex_files/0B_Full_Security_Lemmas}

\end{document}